\documentclass[a4paper,11pt,fleqn]{article}

\usepackage{amsmath,amssymb,amsthm,enumerate}
\usepackage{mathrsfs}

\newcommand{\p}{\partial}
\newcommand{\const}{\mathop{\rm const}\nolimits}
\newcommand{\ord}{\mathop{\rm ord}\nolimits}
\newcommand{\sco}{\mathop{\rm sco}\nolimits}
\newcommand{\wsco}{\mathop{\rm wsco}\nolimits}
\newcommand{\rank}{\mathop{\rm rank}\nolimits}

\newcommand{\ssl}{{\scriptscriptstyle(}}
\newcommand{\ssr}{{\scriptscriptstyle)}}

\newtheorem{theorem}{Theorem}
\newtheorem{lemma}[theorem]{Lemma}
\newtheorem{corollary}[theorem]{Corollary}
\newtheorem{proposition}[theorem]{Proposition}
{\theoremstyle{definition}
\newtheorem{definition}[theorem]{Definition}
\newtheorem{example}[theorem]{Example}
\newtheorem{remark}[theorem]{Remark}
}

\newcommand{\todo}[1][\null]{\ensuremath{\clubsuit}}

\newcommand{\noprint}[1]{}

\begin{document}

\par\noindent {\LARGE\bf
Singular reduction modules of differential equations\par}

{\vspace{4mm}\par\noindent
Vaycheslav M. Boyko~$^\dag$, Michael Kunzinger~$^\ddag$ and Roman O. Popovych~$^\S$
\par\vspace{2mm}\par}

{\vspace{2mm}\par\noindent\it
$^{\dag,\S}$Institute of Mathematics of NAS of Ukraine, 3 Tereshchenkivska Str., Kyiv-4, Ukraine
\par}

{\vspace{2mm}\par\noindent\it
$^{\ddag}$~Fakult\"at f\"ur Mathematik, Universit\"at Wien, Oskar-Morgenstern-Platz 1, 1090 Wien, Austria
\par}

{\vspace{2mm}\par\noindent\it
$^{\S}$~Wolfgang Pauli Institute, Oskar-Morgenstern-Platz 1, 1090 Wien, Austria
\par}

{\vspace{2mm}\par\noindent $\phantom{^{\dag,\ddag}}$\rm E-mail: \it $^\dag$boyko@imath.kiev.ua, $^\ddag$michael.kunzinger@univie.ac.at, $^\S$rop@imath.kiev.ua
 \par}

{\vspace{5mm}\par\noindent\hspace*{5mm}\parbox{150mm}{\small
The notion of singular reduction modules, i.e., of singular modules of nonclassical (conditional) symmetry,
of differential equations is introduced.
It is shown that the derivation of nonclassical symmetries for differential equations
can be improved by an in-depth prior study of the associated singular modules of vector fields.
The form of differential functions and differential equations
possessing parameterized families of singular modules is described up to point transformations.
Singular cases of finding reduction modules are related to lowering the order of the corresponding reduced equations.
As examples, singular reduction modules of evolution equations and second-order quasi-linear equations are studied.
Reductions of differential equations to algebraic equations and to first-order ordinary differential equations are
considered in detail within the framework proposed and are related to previous no-go results on nonclassical symmetries.
}\par\vspace{4mm}}

\noprint{
}



\noprint{
NOTATION
\begin{itemize}\itemsep=0ex
\item
$L$ is a differential function or the left hand side of a differential equation
\item
$\tilde L$ is a differential function such that $L|_{\mathcal Q_{\ssl r\ssr}}=(\lambda\tilde L)|_{\mathcal Q_{\ssl r\ssr}}$
           for some nonvanishing differential function~$\lambda$
\item
$\check L$ is the left hand side of a reduced equation
\item
$\hat L$ is a differential function associated with~$L$ on the manifold~$\smash{\mathcal Q_{\ssl r\ssr}}$
           by excluding of the leading derivatives of the system~$\smash{\mathcal Q_{\ssl r\ssr}}$
\item
$\bar L$ is a differential function in the representation of~$L$ admitting meta-singular modules
\item
$\Lambda$ ($\Lambda^s$) is a multiplier being a differential function
\item
$\lambda$ ($\check\lambda$, $\hat\lambda$, $\tilde\lambda$, $\bar\lambda$) is a nonvanishing multiplier being a differential function
\end{itemize}
}

\section{Introduction}

The ``nonclassical'' approach to finding solutions of differential equations in closed form
was proposed in \cite{Bluman&Cole1969}
using the particular example of the $(1{+}1)$-dimensional heat equation
in order to extend the range of applicability of symmetry methods.
Since the end of the 1980s this method has been applied to many particular differential equations
modeling real-world phenomena,
see, e.g., examples in \cite{ArrigoHillBroadbridge1994,Clarkson&Mansfield1993,Clarkson&Mansfield1994a,Popovych&Vaneeva&Ivanova2007}
and reviews in \cite{Fushchych&Shtelen&Serov1993en,Olver&Vorob'ev1996}.
Related objects, which are similar to subalgebras of Lie symmetry algebras,
are named in the literature in different ways:
nonclassical~\cite{Levi&Winternitz1989},
$Q$-conditional~\cite{Fushchych&Shtelen&Serov1993en},
conditional~\cite{Fushchych&Zhdanov1992},
partial~\cite{Vorob'ev1991} symmetries for short,
or involutive families/modules of nonclassical/conditional symmetry operators
\cite{Olver&Vorob'ev1996,Zhdanov&Tsyfra&Popovych1999} in a more complete form.
The main feature which is inherited by nonclassical symmetries from Lie symmetries is
that they allow to construct ansatzes for the unknown function
which reduce the corresponding differential equation to differential equations
with a smaller number of independent variables
\cite{ArrigoHillBroadbridge1993,Olver1994,Pucci&Saccomandi1992,Vorob'ev1991,Zhdanov&Tsyfra&Popovych1999}.
This feature relates nonclassical symmetries to the direct method by Clarkson and Kruskal~\cite{Clarkson&Kruskal1989}
and the general ansatz method \cite{Fushchych&Shtelen&Serov1993en}.
In fact, however, properties of nonclassical symmetries are more closely connected with
theories of differential constraints and formal compatibility of systems of differential equations
\cite{Kunzinger&Popovych2009,Olver1994,Pucci&Saccomandi1992}.
This is why we mostly use the term ``reduction modules'' (of vector fields)
instead of ``involutive families of conditional symmetry operators''
and say that an involutive module of vector fields reduces a differential equation
if the equation is reduced by the corresponding ansatz.

Involving the associated invariant surface condition in the conditional invariance criterion
gives rise to a few significant complications of nonclassical symmetries in comparison with Lie symmetries.
Given a differential equation~$\mathcal L$, elements of its different reduction modules
do not form objects of a nice algebraic or differential structure.
Hence it is not possible to compose single reduction operators in reduction modules
as this is done for the maximal Lie invariance algebra of~$\mathcal L$ and its subalgebras,
which consist of vector fields generating one-parameter Lie symmetry (pseudo)groups of~$\mathcal L$.
Whereas the system of determining equations for Lie symmetries is linear,
similar systems for reduction modules are nonlinear and should additionally be supplemented,
in the course of considering modules of dimension greater than one,
by the condition of involutivity, i.e., the closure of modules with respect to commutation of vector fields.
Moreover, there is no single system of determining equations even for reduction modules of a fixed dimension.
Instead, the~entire set of such modules is partitioned into subsets associated with systems of determining equations
which are rather different from each other.
Solving some of these systems may be equivalent to solving the initial equation,
which gives no-go cases of looking for reduction modules.
Such no-go cases were known for a number of particular $(1{+}1)$-dimensional evolution equations including
the linear heat equation \cite{Fushchych&Popowych1994-1,Fushchych&Shtelen&Serov&Popovych1992,Fushchych&Shtelen&Serov1993en,Mansfield1999,Webb1990},
the Burgers equation \cite{Arrigo&Hickling2002,Mansfield1999},
linear second-order evolution equations \cite{Popovych1995,Popovych2008a}
as well as for the entire classes of $(1{+}1)$-dimensional evolution equations~\cite{Zhdanov&Lahno1998},
multi-dimensional evolution equations~\cite{Popovych1998}
and even systems of such equations~\cite{Vasilenko&Popovych1999}.
Note that in the course of the study of Lie symmetries a similar no-go situation arises
for first-order ordinary differential equations \cite[Theorem~10, p.~130]{Lie1891},
see also \cite[Section~2.5]{Olver1993}.
In fact, all the above no-go cases of reduction operators are occurrences
of a no-go case common to evolution equations
and one more no-go case specific to linear second-order evolution equations.
The causes giving rise to the partition of the module set and to no-go cases
for nonclassical symmetries have not been investigated in the literature until recently.
It was not understood in what way results on no-go cases can be extended to
reduction modules of other, non-evolution, equations.

In~\cite{Kunzinger&Popovych2010a}
the partition of the set of reduction modules of a differential equation
was related with lowering the order of this equation on the manifolds determined
by the associated invariant surface conditions in the appropriate jet space.
As a result, studying singular modules of vector fields, which lower the order of the equation,
was included as the initial step in the procedure of finding nonclassical symmetries.
In order to illustrate the main ideas of the framework proposed,
we considered only the case of single partial differential equations
in one dependent and two independent variables and single reduction operators.
The notion of singular reduction operators was introduced.
The weak singularity co-order of a reduction operator~$Q$ was shown
to be equal to the essential order of the corresponding reduced equation and the number of essential parameters in
the family of $Q$-invariant solutions.
No-go assertions on singular reduction operators
of $(1{+}1)$-dimensional evolution and wave equations were derived
and then generalized to parameterized families of vector fields which reduce
partial differential equations in two independent variables
to first-order ordinary differential equations.

In the present paper we extend results of~\cite{Kunzinger&Popovych2010a}
to the case of a greater number of independent variables.
After revising and enhancing the framework of nonclassical symmetries in Section~\ref{SectionOnDefOfRedModules},
in Section~\ref{SectionOnSingularModulesOfVectorFieldsForDiffFunctions} we introduce
the concepts of singular and meta-singular modules of vector fields for differential functions.
Any meta-singular module of dimension greater than two turns out to be necessarily involutive,
in contrast to two-dimensional meta-singular modules.
The main result of Section~\ref{SectionOnDiffFunctionsAdmittingMeta-singularModules}
is Theorem~\ref{TheoremOnDiffFunctionsWithMetaSingularModule},
which describes, up to point transformations, differential functions possessing meta-singular modules.
The analogous notions of weakly singular and meta-singular modules for differential equations are introduced
in Section~\ref{SectionOnSingularModulesOfVectorFieldsForDiffEqs}.
Theorem~\ref{TheoremOnWeaklyMetaSingularModuleForDiffEq},
which characterizes differential equations admitting weakly meta-singular modules,
implies that instead of such modules it suffices to study meta-singular modules of the corresponding differential functions.
A~connection between the weak singularity co-order of reduction modules,
the essential order of the corresponding reduced equations and,
in the case of reduction to ordinary differential equations,
the number of parameters in the corresponding families of invariant solutions
is established in Section~\ref{SectionOnReductionModulesAndParametricFamiliesOfSolutions}.
It is shown that
the relation between the reducibility of a differential equation~$\mathcal L$ by an involutive module~$Q$
and the formal compatibility of the joint system of~$\mathcal L$
and the characteristic system associated to~$Q$
essentially involves the weak singularity co-order of~$Q$ for~$\mathcal L$.
Revisiting results of~\cite{Grundland&Tafel1995} within the framework of singular reduction modules,
in Section~\ref{SectionOnReductionModulesOfMaxD} we consider the specific case of reduction modules
of dimension equal to the number of independent variables, which results in the reduction to algebraic equations.
In Section~\ref{SectionOnExampleOfEvolEqs} we reformulate and extend no-go results from~\cite{Popovych1998}
on modules reducing evolution equations to ordinary differential equations
with time as the single independent variable.
This motivates the consideration of reduction modules of singularity co-order one
in Section~\ref{SectionOnReductionModulesOfSingularityCoOrder1}.
Supposing that a differential equation~$\mathcal L$ admits an $n$-dimensional meta-singular module~$M$ of singularity co-order one,
where $n$ is the number of independent variables in~$\mathcal L$,
we prove no-go assertions establishing a connection between
$(n{-}1)$-dimensional reduction modules of~$\mathcal L$ contained in~$M$ and solutions of~$\mathcal L$.
In particular, it is shown that the system of determining equations for such modules is reduced to the initial equation~$\mathcal L$
by a composition of a differential substitution and a hodograph transformation.
The final Section~\ref{SectioOnSingularModulesForQuasiLinearSecondOrderPDEs} is devoted to
singular modules for quasi-linear second-order PDEs,
and the dimension of modules is assumed to be less than the number of independent variables.
Thus, elliptic equations possess no singular modules.
Any evolution equation whose matrix of coefficients of second-order derivatives is non-degenerate
possesses only singular modules as considered in Section~\ref{SectionOnExampleOfEvolEqs} for general evolution equations.
Generalized wave equations are much more complicated from this point of view.
In particular, they may admit families of singular modules which have no interpretation in terms of meta-singular modules,
which makes a further development of the framework of singular modules desirable.

\section{Reduction modules of differential equations}
\label{SectionOnDefOfRedModules}

In this section, based on
\cite{Fushchych&Tsyfra1987,Fushchych&Zhdanov1992,Olver&Vorob'ev1996,%
Kunzinger&Popovych2009,Kunzinger&Popovych2010a,%
Popovych&Vaneeva&Ivanova2007,Zhdanov&Tsyfra&Popovych1999},
we revise and enhance the framework of nonclassical (conditional) symmetries of
differential equations.
Due to substantiating our choice with different arguments, we use the name ``reduction modules''
instead of ``involutive families of nonclassical (conditional) symmetry operators''.

Given a foliated space of $n$ independent variables $x=(x_1,\dots,x_n)$ and a single dependent variable~$u$,
consider a finite-dimensional involutive module~$Q$ of vector fields in this space,
and suppose that the module dimension~$p$ of~$Q$ (over the ring of smooth functions of $(x,u)$) is not greater than~$n$, $0<p\leqslant n$.
We additionally assume that the module~$Q$ satisfies the \emph{rank condition}, i.e.,
for each fixed value of~$(x,u)$ the projection of~$Q$ to the space of~$x$ is $p$-dimensional.
The attribute `\emph{involutive}' means that the commutator of any two vector fields from~$Q$ belongs to~$Q$.
It is obvious that any one-dimensional module is involutive.
Therefore, in the case $p=1$ we can omit the attribute `involutive' and talk only about modules.

In what follows
the indices $i$ and $j$ run from 1 to $n$,
the index~$s$ runs from 1 to $p$,
the index~$\sigma$ runs from 1 to $n-p$,
and we use the summation convention for repeated indices.
Angular brackets $\langle\dots\rangle$ are used for denoting linear spans over the ring of smooth functions of $(x,u)$.
Subscripts of functions denote differentiation with respect to the corresponding variables,
$\p_i=\p/\p x_i$ and $\p_u=\p/\p u$.
Any function is considered as its zero-order derivative.
All considerations are in the local setting.
The notion of functional independence will be understood
in the sense of total functional independence, which means that functions are in fact
functionally independent on each open subset of their common domain.

Suppose that the vector fields $Q_s=\xi^{si}(x,u)\p_i+\eta^s(x,u)\p_u$ form a basis of~$Q$,
i.e., $Q=\langle Q_1,\dots,Q_p\rangle$.
Then the rank condition is equivalent to the equality $\rank(\xi^{si})=p$.
The condition that the commutator of any pair of basis elements belongs to~$Q$, $[Q_s,Q_{s'}]\in Q$, suffices
for the module~$Q$ to be involutive.
If the vector fields~$\tilde Q_1$, \dots, $\tilde Q_p$ form another basis of~$Q$, 
then there exists a nondegenerate $p\times p$ matrix-function $(\lambda^{ss'}(x,u))$ such that
$\tilde Q_s=\lambda^{ss'}Q_{s'}$.

The first-order differential function~$Q_s[u]:=\eta^{s}(x,u)-\xi^{si}(x,u)u_i$
is called the {\it characteristic} of the vector field~$Q_s$.
In view of the Frobenius theorem, involutivity of~$Q$ is equivalent to the fact
that the characteristic system~$\mathcal Q$ of PDEs $Q_s[u]=0$
(also called the \emph{invariant surface condition}) has
$n+1-p$ functionally independent integrals $I^0(x,u)$, \dots, $I^{n-p}(x,u)$.
Therefore, the general solution of this system can implicitly be represented in the form
$F(I^0,\dots,I^{n-p})=0$, where~$F$ is an arbitrary smooth function of its arguments.

A differential function $G=G[v]$ of the dependent variables $v=(v^1,\ldots,v^m)$
which in turn are functions of a tuple of independent variables $y=(y_1,\ldots,y_l)$
will be viewed as a smooth function of $y$ and a finite number of derivatives of~$v$ with respect to~$y$.
More rigorously, the differential function $G$ is defined as a smooth function on a domain of the jet space $\mathrm J^r=\mathrm J^r(y|v)$ of some order~$r$
with independent variables $y$ and dependent variables~$v$~\cite{Olver1993}.
The order $r=\ord G$ of the differential function $G$ is defined to equal
the maximal order of derivatives (resp.\ jet variables) involved in $G$,
and $\ord G=-\infty$ if $G$ depends only on~$y$.
Each set of differential functions of a fixed positive order as well as
the set of differential functions of nonpositive order
are invariant with respect to point transformations of~$(y,v)$.

Using another basis of~$Q$ gives just another representation of the characteristic system~$\mathcal Q$ with the same set of solutions.
This is why the characteristic system~$\mathcal Q$ is associated with the module~$Q$ rather than with a fixed basis of~$Q$.
And vice versa, any family of $n+1-p$ functionally independent functions of~$x$ and $u$
is a complete set of integrals of the characteristic system of an involutive $p$-dimensional module.
Therefore, there exists a one-to-one correspondence between the set of involutive $p$-dimensional modules
and the set of families of $n+1-p$ functionally independent functions of~$x$ and $u$,
which is factorized with respect to the corresponding equivalence.
(We consider two families of the same number of functionally independent functions of the same arguments as equivalent
if any function from one of the families is functionally dependent on functions from the other family.)

A function $u=f(x)$ is called \emph{invariant with respect to the involutive module}~$Q$
(or, briefly, $Q$-\emph{invariant}) if it is a solution of the characteristic system~$\mathcal Q$.
This notion is justified by the following facts.
In view of the rank condition, we can choose a basis of~$Q$ that spans, over the underlying field,
a $p$-dimensional (Abelian) Lie algebra $\mathfrak g$ of vector fields in the space $(x,u)$.%
\footnote{%
Such a basis is constructed in the following way:
We take an arbitrary basis of~$Q$ consisting of vector fields $Q_s=\xi^{si}(x,u)\p_i+\eta^s(x,u)\p_u$.
Up to permutation of the independent variables and basis elements of~$Q$,
we can suppose in view of the rank condition that $\rank(\xi^{ss'})=p$
and change the basis to
$(\hat Q_s=\p_s+\hat\xi^{s\iota}\p_\iota+\hat\eta^s\p_u)$, where the index~$\iota$ runs from $p+1$ to~$n$
and the matrices $(\hat\xi^{s\iota})$ and $(\hat\eta^s)$ are the products of the matrix~$(\xi^{ss'})^{-1}$
by the matrices $(\xi^{s\iota})$ and $(\eta^s)$, respectively.
Since the module~$Q$ is involutive, the vector fields $\hat Q_1$, \dots, $\hat Q_p$ commute.
\label{FootnoteOnCommutativeBasis}}
The graph of each solution of the characteristic system~$\mathcal Q$
is obviously invariant with respect to the $p$-parameter local transformation group
generated by the algebra~$\mathfrak g$.

We choose a basis of~$Q$ that consists of commuting vector fields $Q_1$, \dots, $Q_p$
and, for each fixed~$s$, consider a solution $J^s=J^s(x,u)$ of the system $Q_{s'}J^s=\delta_{ss'}$,
where $\delta_{ss'}$ is the Kronecker delta.
Since the functions $I^0$, \dots, $I^{n-p}$, $J^1$, \dots, $J^p$ of $(x,u)$ are functionally independent,
one can make the change of variables
\[
\varphi=I^0(x,u),\quad \omega_\sigma=I^\sigma(x,u),\quad \omega'_s=J^s(x,u),
\]
where $\omega=(\omega_1,\dots,\omega_{n-p})$ and $\omega'=(\omega'_1,\dots,\omega'_p)$
are considered as the new independent variables
and $\varphi$ is the new dependent variable.
The variables $\omega$ and $\varphi$ are called $Q$-invariant,
and the variables $\omega'$ are called parametric for the module~$Q$.
In the new variables, the basis elements $Q_s$ take the form $Q_s=\p_{\omega'_s}$.

Next, consider an $r$th order differential equation~$\mathcal L$ of the form $L(x,u_{\ssl r\ssr})=0$
for the unknown function $u$ of the independent variables $x=(x_1,\dots,x_n).$
Here, $u_{\ssl r\ssr}$ denotes the set of all the derivatives of the function $u$ with respect to $x$
of order not greater than~$r$, including $u$ as the derivative of order zero.
We assume that the order~$r$ of the equation~$\mathcal L$ is essential,
i.e., it is minimal among the orders of equations
equivalent to~$\mathcal L$ up to nonvanishing multipliers that are differential functions of~$u$.
In the local approach the equation~$\mathcal L$ can be viewed as an algebraic equation
in the $r$th order jet space $\mathrm J^r=\mathrm J^r(x|u)$ and is identified with the manifold of its solutions in~$\mathrm J^r$,
\begin{gather*}
\mathcal L=\{ (x,u_{\ssl r\ssr}) \in \mathrm J^r\mid L(x,u_{\ssl r\ssr})=0\}.
\end{gather*}
We use the same symbol~$\mathcal L$ for this manifold and also write $\mathcal Q_{\ssl r\ssr}$
both for the system consisting of the independent differential consequences of the characteristic system~$\mathcal Q$
up to equation order~$r$ and for the manifold defined by the system~$\mathcal Q_{\ssl r\ssr}$ in $\mathrm J^r$, i.e.,
\begin{gather*}
\mathcal Q_{\ssl r\ssr}=\{ (x,u_{\ssl r\ssr}) \in \mathrm J^r\mid D^\alpha Q_s[u]=0,\,|\alpha|<r\},
\end{gather*}
where $ D^\alpha=D_1^{\alpha_1}\cdots D_n^{\alpha_n}$,
$D_i=\p_{x_i}+u_{\alpha+\delta_i}\p_{u_\alpha}$ is the operator of total differentiation with respect to the variable~$x_i$,
$\alpha=(\alpha_1,\dots,\alpha_n)$ is an arbitrary multi-index,
$\alpha_i\in\mathbb N\cup\{0\}$, $|\alpha|:=\alpha_1+\cdots+\alpha_n$,
and $\delta_i$ is the multi-index whose $i$th entry equals~1 and whose other entries are zero.
The variable $u_\alpha$ of the jet space $\mathrm J^r$ corresponds to the derivative
$\p^{|\alpha|}u/\p x_1^{\alpha_1}\dots\p x_n^{\alpha_n}$,
and $u_i\equiv u_{\delta_i}$, $u_{ij}\equiv u_{\delta_i+\delta_j}$, etc.

Without loss of generality, we can assume $F_{I^0}\ne0$ in the representation $F(I^0,\dots,I^{n-p})=0$
of the general solution of the characteristic system~$\mathcal Q$
and resolve the equation $F(I^0,\dots,I^{n-p})=0$ with respect to~$I^0$:
$I^0=\varphi(I^1,\dots,I^{n-p})$.
In view of the rank condition, this gives the representation (in general, also implicit)
\begin{equation}\label{EqGenAnsatzForm}
\mathcal A\colon\quad I^0(x,u)=\varphi(\omega),\quad \omega_\sigma=I^\sigma(x,u),
\end{equation}
for solutions of the characteristic system~$\mathcal Q$,
where $\varphi=\varphi(\omega)$ is an arbitrary smooth function of its arguments,
which is called an \emph{ansatz} for~$u$ constructed with the module~$Q$.

Making the ansatz~$\mathcal A$ for~$u$
we can express all derivatives of~$u$ in terms of~$\omega$, $\omega'$ and derivatives of~$\varphi$,
then substitute these expressions to~$L[u]$
and replace the remaining~$x$'s by their expressions in the new variables.
Alternatively, we can change the variables to $(\omega,\omega',\varphi)$ from the outset
and then take into account the constraints $\varphi_{\omega'_s}=0$.
The function obtained by the above procedure is denoted by~$L|_{\mathcal A}$.
It depends at most on $\omega$, $\omega'$ and $\varphi_{\ssl r\ssr}$,
where $\varphi_{\ssl r\ssr}$ denotes the tuple of derivatives of $\varphi$ with respect to~$\omega$
up to order~$r$.

\begin{definition}\label{DefinitionOfReductionModule}
The ansatz~$\mathcal A$ constructed with the module~$Q$ \emph{reduces} the equation~$\mathcal L$ if
there exist smooth functions $\check\lambda=\check\lambda(\omega,\omega',\varphi_{\ssl r\ssr})$ and
$\check L=\check L(\omega,\varphi_{\ssl r\ssr})$ such that the function~$\check\lambda$ does not vanish and
\[L|_{\mathcal A}=\check\lambda(\omega,\omega',\varphi_{\ssl r\ssr})\check L(\omega,\varphi_{\ssl r\ssr}).\]
Then the module~$Q$ is called a \emph{reduction module} of~$\mathcal L$,
and the equation $\check L(\omega,\varphi_{\ssl r\ssr})=0$
is a \emph{reduced equation} associated with the ansatz~$\mathcal A$.
\end{definition}

The reduction procedure should additionally be specified
in the case of reduction to algebraic equations when $p<n$,
see the proof of Theorem~\ref{TheoremOnReductionModulesAndReducedEqs}.

The set of $p$-dimensional reduction modules of the equation~$\mathcal L$ will be denoted by~$\mathcal R^p(\mathcal L)$.

Consider the following conditions on the ($r$th order) differential equation~$\mathcal L$
and the involutive module~$Q$ satisfying the rank condition:
\par\smallskip\par
(C1) $Q$ is a reduction module of the equation~$\mathcal L$;
\par\smallskip\par
(C2) $V_{\ssl r\ssr}L[u]\in\langle L[u],D^\alpha Q_s[u]=0,|\alpha|<r \rangle$ for any $V\in Q$;
\par\smallskip\par
(C3) $V_{\ssl r\ssr}L[u]\big|_{\mathcal L\cap\mathcal Q_{\ssl r\ssr}}=0$ for any $V\in Q$.

\par\smallskip\par\noindent
Here $V_{\ssl r\ssr}$ denotes the standard $r$th prolongation
of a vector field~$V=\xi^i(x,u)\p_i+\eta(x,u)\p_u$ \cite{Olver1993,Ovsiannikov1982}:
$V_{\ssl r\ssr}=V+\sum_{0<|\alpha|{}\leqslant r} \eta^{\alpha}\p_{u_\alpha}$,
where
$\eta^{\alpha}=D^\alpha V[u]+\xi^iu_{\alpha+\delta_i}$ and $V[u]=\eta-\xi^iu_i$.
In the conditions~(C2) and~(C3) it suffices to require
that $V$ runs through a basis $(Q_1,\dots,Q_p)$ of~$Q$.
What basis is chosen for representing the characteristic system~$\mathcal Q$
and checking the conditions~(C2) and~(C3) is not essential;
cf.\ \cite{Fushchych&Zhdanov1992,Zhdanov&Tsyfra&Popovych1999}.

All the conditions are preserved by point transformations of the variables~$(x,u)$.

\begin{theorem}\label{TheoremOnCriterionOnReductionModules}
The conditions {\rm (C1)} and {\rm (C2)} are equivalent and imply {\rm (C3)}.
If the tuple of differential functions $(L[u],D^\alpha Q_s[u]=0,|\alpha|<r)$
is of maximal rank on $\mathcal L\cap\mathcal Q_{\ssl r\ssr}$,
then the condition~{\rm (C3)} implies {\rm (C2)} (and thereby also~{\rm (C1)}).
\end{theorem}

Besides the case of maximal rank, there are other, more specific, cases
when the conditions (C1)--(C3) are simultaneously satisfied,
e.g., if $\mathcal L\cap\mathcal Q_{\ssl r\ssr}=\mathcal Q_{\ssl r\ssr}$.

A proof of Theorem~\ref{TheoremOnCriterionOnReductionModules} is presented in~\cite{Boyko&Kunzinger&Popovych2016b}.
It relies on the following assertion.

\begin{lemma}\label{LemmaOnNonvanishingMultiplier}
Let smooth functions~$f$, $\Lambda^1$, \dots, $\Lambda^p$,
and an involutive module $Q=\langle Q_1,\dots,Q_p\rangle$ of vector fields
that are defined on a neighborhood~$O_{z^0}$ of a point~$z^0\in\mathbb R^l$ for some $l\in\mathbb N$
satisfy the conditions $Q_sf(z)=\Lambda^s(z)f(z)$ for any $z\in O_{z^0}$, $s=1,\dots,p$, and $\dim Q|_{z^0}=p$.
Then there exist a neighborhood~$\check O_{z^0}\subset O_{z^0}$ of~$z^0$ and
smooth functions~$\check f$ and $\lambda$ defined on~$\check O_{z^0}$
such that $\lambda(z)\ne0$, $Q^s\check f(z)=0$ and $f(z)=\lambda(z)\check f(z)$ for any~$z\in\check O_{z^0}$.
\end{lemma}

In previous papers on reduction modules, a different terminology was used
(see, e.g., \cite{Bluman&Cole1969,Fushchych&Shtelen&Serov1993en,Fushchych&Tsyfra1987,Olver&Rosenau1987,Zhdanov&Tsyfra&Popovych1999}).
Usually the condition~{\rm (C3)} was considered as the main one and was called the \emph{conditional invariance criterion}.
Then the differential equation~$\mathcal L$ is called \emph{conditionally invariant} with respect to the involutive module~$Q$,
whereas the module~$Q$ is called an \emph{involutive module of conditional symmetry}
(or $Q$-conditional symmetry, or nonclassical symmetry, etc.) operators of the equation~$\mathcal L$.
A version of the condition~(C2) for systems of differential equations for several unknown functions
appeared in~\cite{Grundland&Tafel1995}.
In contrast to the case of a single differential equation for a single unknown function,
the version of the condition~(C2) for systems is sufficient but not necessary
for an ansatz constructed with the module~$Q$ to reduce the system under consideration, cf.\ \cite[Section~5]{Kunzinger&Popovych2009}.
An alternative approach to conditional invariance is to demand that the joint system of~$\mathcal L$ and~$\mathcal Q_{\ssl r\ssr}$
is formally compatible in the sense of the absence of nontrivial differential consequences~\cite{Olver1994,Olver&Vorob'ev1996}.
But a delicate point is which representation of the joint system should in fact be considered to decide formal compatibility,
cf.\ \cite[footnote~1]{Kunzinger&Popovych2009}
and Section~\ref{SectionOnReductionModulesAndParametricFamiliesOfSolutions} of the present paper.
If the conditional invariance criterion is not satisfied but nevertheless
the equation~$\mathcal L$ has $Q$-invariant solutions, 
then one talks about weak invariance of the equation~$\mathcal L$
with respect to the module~$Q$~\cite{Olver&Rosenau1987,Olver&Vorob'ev1996,Pucci&Saccomandi1992}.

There are reduction modules related to classical Lie symmetries.
Let $\mathfrak g$ be a $p$-dimensional Lie invariance algebra of the equation~$\mathcal L$,
whose basis vector fields $Q_1$, \dots, $Q_p$ satisfy the condition
$\rank (\xi^{si})=\rank(\xi^{si},\eta^s)=p'$, where $p'\leqslant p$.
Then the span of $Q_1$, \dots, $Q_p$ over the ring of smooth functions of~$(x,u)$
is a $p'$-dimensional involutive module which belongs to~$\mathcal R^{p'\!}(\mathcal L)$.
Modules of this kind are called \emph{Lie reduction modules}.
Other reduction modules are called \emph{non-Lie}.

The following assertion on reduction modules is important for the considerations of this paper
(cf.\ \cite{Zhdanov&Tsyfra&Popovych1999}).

\begin{lemma}\label{LemmaOnReformulationOfCondInvCriterion}
Given an $r$th order differential equation~$\mathcal L$: $L[u]=0$,
a $p$-dimensional ($0<p\leqslant n$) involutive module~$Q$ satisfying the rank condition
and differential functions $\tilde L[u]$ and $\lambda[u]\ne0$ of an order not greater than~$r$
such that \smash{$(L-\lambda\tilde L)|_{\mathcal Q_{\ssl r\ssr}}=0$},
the module~$Q$ is a reduction module of~$\mathcal L$ if and only if
it is a reduction module of the equation~$\tilde{\mathcal L}$: $\tilde L[u]=0$.
An ansatz constructed with the module~$Q$ reduces~$\mathcal L$ and~$\tilde{\mathcal L}$
to equations that may differ at most by a nonvanishing multiplier.
\end{lemma}

The classification of reduction modules can be notably enhanced and simplified by involving
Lie symmetry and equivalence transformations of (classes of) differential equations.
By~$\mathfrak M^p$ we denote the set of $p$-dimensional modules of vector fields in the space of~$(x,u)$.
Any point transformation of $(x,u)$ induces a one-to-one mapping of~$\mathfrak M^p$ into itself
via push-forward of vector fields.
Namely, the transformation~$g$: $\tilde x_i=X^i(x,u)$, $\tilde u=U(x,u)$ generates
the mapping $g_*\colon \mathfrak M^p\to\mathfrak M^p$ such that
for any $Q\in\mathfrak M^p$ and $V\in Q$
the vector field~$V=\xi^i(x,u)\p_i+\eta(x,u)\p_u$ is mapped to the vector field
$g_*V=\tilde\xi^i\p_{\tilde x_i}+\tilde\eta\p_{\tilde u}$, where
$\tilde\xi^i(\tilde x,\tilde u)=VX^i(x,u)$,
$\tilde\eta(\tilde x,\tilde u)=VU(x,u)$.

Given a group $G$ of point transformations in the space of~$(x,u)$,
the modules~$Q$ and~$\tilde Q$ (of the same dimension) are called \emph{equivalent} with respect to~$G$
if there exists some $g\in G$ such that $\tilde Q=g_*Q$.

\begin{lemma}\label{LemmaOnInducedMapping}
Suppose that $Q\in \mathcal R^p(\mathcal L)$,
a point transformation $g$ maps a differential equation~$\mathcal L$ to a differential equation~$\tilde{\mathcal L}$
and the image $g_*Q$ satisfies the rank condition.
Then $g_*Q\in \mathcal R^p(\tilde{\mathcal L})$.
\end{lemma}

\begin{corollary}\label{CorollaryOnEquivReductionOperatorWrtSymGroup}
Let $G$ be the point symmetry group of an equation~$\mathcal L$.
Then the equivalence of $p$-dimensional modules of vector fields
with respect to the group $G$ generates an equivalence relation in~$\mathcal R^p(\mathcal L)$.
\end{corollary}

Next, we consider a class~$\mathcal L|_{\mathcal S}$ of differential equations~$\mathcal L_\theta$: $L(x,u_{\ssl r\ssr},\theta_{\ssl q\ssr})=0$.
Here $L$ is a fixed function of $x$, $u_{\ssl r\ssr}$ and $\theta_{\ssl q\ssr}$,
where by $\theta$ we denote the tuple of arbitrary (parametric) differential functions
$\theta(x,u_{\ssl r\ssr})=(\theta^1(x,u_{\ssl r\ssr}),\ldots,\theta^k(x,u_{\ssl r\ssr}))$,
traversing the set~${\mathcal S}$ of solutions of an auxiliary system,
and $\theta_{\ssl q\ssr}$ stands for the set of all the derivatives of $\theta$ of order not greater than $q$
with respect to $x$ and $u_{\ssl r\ssr}$.
The auxiliary system consists of differential equations $S(x,u_{\ssl r\ssr},\theta_{\ssl q'\ssr}(x,u_{\ssl r\ssr}))=0$
and differential inequalities $\Sigma(x,u_{\ssl r\ssr},\theta_{\ssl q'\ssr}(x,u_{\ssl r\ssr}))\ne0$ ($>0$, $<0$,~\dots) on $\theta$,
where both $x$ and $u_{\ssl r\ssr}$ play the role of independent variables.
Henceforth we call the functions $\theta$ arbitrary elements.
We write $G^\sim$ and~$\mathcal G^\sim$ for the equivalence group and the equivalence groupoid
of the class~$\mathcal L|_{\mathcal S}$, respectively.
Roughly speaking, the group~$G^\sim$ consists of the transformations of $(x,u_{\ssl r\ssr},\theta)$
that preserve the form of the equations from~$\mathcal L|_{\mathcal S}$
and are point transformations with respect to $(x,u)$ when $\theta$ is fixed.
In fact, there are various kinds of equivalence groups~\cite[Section~2.3]{Popovych&Kunzinger&Eshraghi2010}.
The groupoid~$\mathcal G^\sim$ is the set
$\{(\theta,\tilde\theta,g)\mid\theta,\tilde\theta\in \mathcal S,\,g\in\mathrm T(\theta,\tilde\theta)\}$
naturally equipped with the groupoid structure via the composition of transformations.
Here $\mathrm T(\theta,\tilde\theta)$ denotes the set of point transformations of $(x,u)$
that map the equation~$\mathcal L_\theta$ to the equation~$\mathcal L_{\tilde\theta}$.
See \cite{Bihlo&Cardoso-Bihlo&Popovych2012,Popovych&Kunzinger&Eshraghi2010}
for rigorous definitions of notions related to classes of differential equations.

By $P$ we denote the set of all pairs of the form $(\mathcal L_\theta,Q)$,
where $\mathcal L_\theta$ is an equation from~$\mathcal L|_{\mathcal S}$ and a
module~$Q$ from~$\mathcal R^p(\mathcal L_\theta)$.
It follows from Lemma~\ref{LemmaOnInducedMapping}
that the action of transformations from the equivalence group~$G^\sim$ or from the equivalence groupoid~$\mathcal G^\sim$
on $\mathcal L|_{\mathcal S}$ and $\{\mathcal R^p(\mathcal L_\theta)\mid\theta\in{\mathcal S}\}$
induces an equivalence relation on~$P$~\cite{Popovych&Vaneeva&Ivanova2007}.

\begin{definition}\label{DefinitionOfEquivOfRedOperatorsWrtEquivGroup}
Let $\theta,\theta'\in{\mathcal S}$,
$Q\in\mathcal R^p(\mathcal L_\theta)$, $Q'\in\mathcal R^p(\mathcal L_{\theta'})$.
The pairs~$(\mathcal L_\theta,Q)$ and~$(\mathcal L_{\theta'},Q')$
are called {\em $G^\sim$-equivalent} if there exists a transformation $\mathcal T\in G^\sim$
mapping the equation~$\mathcal L_\theta$ to the equation~$\mathcal L_{\theta'}$, and $Q'=(\mathcal T^\theta)_*Q$.
Here $\mathcal T^\theta$ is the point transformation of $(x,u)$ obtained from~$\mathcal T$ by fixing~$\theta$.
The pairs~$(\mathcal L_\theta,Q)$ and~$(\mathcal L_{\theta'},Q')$
are called {\em $\mathcal G^\sim$-equivalent} (or, simply, \emph{pointwise equivalent})
if there exists a transformation $g\in\mathrm T(\theta,\tilde\theta)$ such that $Q'=g_*Q$.
\end{definition}

We will interpret the classification of reduction modules with respect to~$G^\sim$ (or $\mathcal G^\sim$)
as the classification in~$P$ with respect to the corresponding equivalence relation,
a problem which can be investigated similarly to
the usual group classification in classes of differential equations.
Namely, at first we construct the reduction modules that are defined for all values of $\theta$.
Then we classify, with respect to $G^\sim$ (or $\mathcal G^\sim$), the values of $\theta$
for which the equation~$\mathcal L_\theta$ admits additional reduction modules.

\section{Singular modules of vector fields for differential functions}
\label{SectionOnSingularModulesOfVectorFieldsForDiffFunctions}

Let $L=L[u]$ be a differential function of order $\ord L=r>0$
and let $Q$ be a $p$-dimensional ($0<p<n$) involutive module which is spanned by
the vector fields $Q_s=\xi^{si}(x,u)\p_i+\eta^s(x,u)\p_u$ defined in the space of $(x,u)$
and satisfying the rank condition $\rank(\xi^{si})=p$.

\begin{definition}\label{DefinitionOfSingularModuleForDiffFunction}
The module~$Q$ is called \emph{singular} for the differential function~$L$
if there exists a differential function $\tilde L=\tilde L[u]$ of an order less than~$r$
such that $L|_{\mathcal Q_{\ssl r\ssr}}=\tilde L|_{\mathcal Q_{\ssl r\ssr}}$.
Otherwise $Q$ is called a \emph{regular} module for the differential function~$L$.
If the minimal order of differential functions
whose restrictions to $\mathcal Q_{\ssl r\ssr}$ coincide with $L|_{\mathcal Q_{\ssl r\ssr}}$ equals $k\in\{-\infty,0,1,\dots,r\}$, 
then the module~$Q$ is said to be of \emph{singularity co-order $k$} for the differential function~$L$.
The module~$Q$ is called \emph{ultra-singular} for the differential function~$L$
if $L|_{\mathcal Q_{\ssl r\ssr}}\equiv0$.
\end{definition}

In particular, if a module is regular for the differential function~$L$, then its singularity co-order is $r=\ord L$.
The singularity co-order of the module~$Q$ for the differential function~$L$ will be denoted by $\sco_LQ$. 
It is obvious that any differential function of nonpositive order admits no singular involutive modules, 
which explains the supposition $\ord L=r>0$ in Definition~\ref{DefinitionOfSingularModuleForDiffFunction}.

The case $p=n$ is special.
Given an $n$-dimensional involutive module~$Q$ which is generated by vector fields satisfying the rank condition,
for any $r$th order differential function~$L=L[u]$ there exists a nonpositive-order differential function $\tilde L=\tilde L[u]$
such that $L|_{\mathcal Q_{\ssl r\ssr}}=\tilde L|_{\mathcal Q_{\ssl r\ssr}}$.
This is why in this case it is natural to assume that
the module~$Q$ is singular for~$L$ only if $\sco_LQ=-\infty$.

Next we show how to algorithmically construct a function~$\tilde L$ satisfying the conditions of Definition~\ref{DefinitionOfSingularModuleForDiffFunction}.
First note that without loss of generality we can consider the basis
$(\hat Q_s=\p_s+\hat\xi^{s\iota}\p_\iota+\hat\eta^s\p_u)$, where the index~$\iota$ runs from $p+1$ to~$n$,
cf.\ footnote~\ref{FootnoteOnCommutativeBasis}.
Following Riquier's compatibility theory, the consideration of this basis can be interpreted as
designating derivatives of~$u$ that contain differentiations with respect to~$x_s$
to be \emph{principal derivatives} for the the characteristic system~$\mathcal Q$.
Then any $\mathcal Q$-principal derivative of order not greater than~$r$ can be expressed,
on the manifold~$\smash{\mathcal Q_{\ssl r\ssr}}$,
via derivatives of~$u$ with respect to $x_{p+1}$, \dots, $x_n$ only
and the coefficients $\hat\xi^{s\iota}$ and $\hat\eta^s$.
For example, for the first- and second-order derivatives we have
\begin{equation}\begin{split}\label{EqSingularRedModuleExpressionsForU_2}
&u_s=\hat\eta^s-\hat\xi^{s\iota} u_\iota,\\
&u_{s\iota}=\hat\eta^s_\iota-\hat\xi^{s\iota'}_\iota u_{\iota'}
+\hat\eta^s_uu_\iota-\hat\xi^{s\iota'}_uu_{\iota'}u_\iota-\hat\xi^{s\iota'} u_{\iota\iota'},\\
&u_{ss'}=\hat\eta^s_{s'}-\hat\xi^{s\iota}_{s'}u_\iota
+(\hat\eta^s_u-\hat\xi^{s\iota}_uu_\iota)(\hat\eta^{s'}-\hat\xi^{s'\iota'} u_{\iota'})\\
&\qquad\ \ {}-\hat\xi^{s\iota}(\hat\eta^{s'}_\iota-\hat\xi^{s'\iota'}_\iota u_{\iota'}
+\hat\eta^{s'}_uu_\iota-\hat\xi^{s'\iota'}_uu_{\iota'}u_\iota-\hat\xi^{s'\iota'} u_{\iota\iota'}).
\end{split}\end{equation}
Substituting the expressions for the $\mathcal Q$-principal derivatives up to order~$r$ into~$L$,
we obtain a differential function~$\hat L$
depending only on $x$, $u$ and derivatives of~$u$ with respect to~$x_{p+1}$, \dots, $x_n$.
We will call $\hat L$ a \emph{differential function associated with~$L$
on the manifold~$\smash{\mathcal Q_{\ssl r\ssr}}$}.
The module~$Q$ is singular for the differential function~$L$
if and only if the order of $\hat L$ is less than~$r$.
The co-order of singularity of~$Q$ equals the order of~$\hat L$.
The module~$Q$ is ultra-singular if and only if $\hat L\equiv0$.
Therefore, we can take~$\hat L$ as~$\tilde L$.
This shows that testing for an involutive module satisfying the rank condition
to be singular for a differential function
is indeed feasible in an entirely algorithmic way
and can easily be included in existing programs for symbolic calculations of symmetries.

\begin{example}\label{ExamplesOfSingularModules}
We present particular singular modules, $Q$, for differential functions, $L$,
associated with PDEs of physical relevance, $L=0$.

1. $n=3$, $r=2$, $L=u_1-(f(u)u_2)_2-(f(u)u_3)_3$ with $f\ne0$,
the (1+2)-dimensional nonlinear isotropic diffusion equation,
$p=2$, $Q=\langle\p_2,\p_3\rangle$, $\hat L=u_1$, $\sco_LQ=1$.

2. $n=2$, $r=4$, $L=g(x_2)u_{11}-(f(x_2)u_{22})_{22}$ with $fg\ne0$,
the Euler--Bernoulli beam equation,
$p=1$, $Q=\langle\p_2\rangle$, $\hat L=g(x_2)u_{11}$, $\sco_LQ=2$.

3. $n=4$, $r=2$, $L=u_{11}-u_{22}-u_{33}-u_{44}$,
the (1+3)-dimensional homogeneous linear wave equation,
$p=3$, $Q=\langle\p_1+\p_4,\p_2,\p_3\rangle$, $\hat L=0$, $\sco_LQ=-\infty$,
and, moreover, $Q$ is ultra-singular for~$L$.

4. $n=4$, $r=2$, $L=u_{11}-u_{22}-u_{33}-u_{44}-u$,
the (1+3)-dimensional Klein--Gordon equation,
$p=3$, $Q=\langle\p_1+\p_4,\p_2,\p_3\rangle$, $\hat L=u$, $\sco_LQ=0$.

5. $n=3$, $r=2$, $L=u_1-(f(u)u_2)_2-g(u)u_3$ with $f\ne0$,
a (1+2)-dimensional degenerate nonlinear diffusion--convection equation,
a) $p=1$, $Q=\langle\p_2\rangle$, $\hat L=u_1-g(u)u_3$, $\sco_LQ=1$;
b) for $Q=\langle\p_2,\p_3\rangle$ with $p=2$, $\hat L=u_1$, thus we also have $\sco_LQ=1$;
c) $p=2$, $Q=\langle\p_2,\p_1-g(u)\p_3\rangle$, $\hat L=0$, $\sco_LQ=-\infty$,
and, moreover, $Q$ is ultra-singular for~$L$.
\end{example}

\begin{proposition}\label{PropositionOnModulesWithSingularSubmodules}
Suppose that $L=L[u]$ is a differential function of $u=u(x)$, $x=(x_1,\dots,x_n)$,
and let $Q$ be an involutive module of vector fields defined in the space of $(x,u)$,
which is of dimension less than~$n$ and satisfies the rank condition.
Then $\sco_LQ\leqslant\sco_L\check Q$ for any involutive submodule~$\check Q$ of~$Q$.
In particular, the module~$Q$ is singular for~$L$ if it contains a submodule singular for~$L$.
\end{proposition}

\begin{proof}
If $\check Q$ is an involutive submodule of~$Q$, 
then it necessarily satisfies the rank condition and
$\smash{\mathcal Q_{\ssl r\ssr}}\subseteq\smash{\check{\mathcal Q}_{\ssl r\ssr}}$, where $r=\ord L$.
If the differential function~$L$ coincides with a differential function~$\tilde L$
on the manifold $\smash{\check{\mathcal Q}_{\ssl r\ssr}}$,
the same is true on the manifold $\smash{\mathcal Q_{\ssl r\ssr}}$.
Therefore, $\sco_LQ\leqslant\sco_L\check Q$.
\end{proof}

\begin{definition}\label{DefinitionOfMetaSingularModuleForDiffFuction}
A $(p{+}1)$-dimensional ($0<p<n$) module~$M$ is called \emph{meta-singular} for the differential function~$L$
if any $p$-dimensional involutive submodule of~$M$ which satisfies the rank condition is singular for~$L$ and
the module~$M$ contains a family $\mathfrak M=\{Q^\Phi\}$ of such submodules
parameterized by an arbitrary function $\Phi=\Phi(x,u)$ of all independent and dependent variables.
The \emph{singularity co-order} of the meta-singular module~$M$, which is denoted by $\sco_LM$,  
is the maximum of the singularity co-orders of its involutive $p$-dimensional submodules satisfying the rank condition.
\end{definition}

\begin{definition}\label{DefinitionOfMetaRegularModuleForDiffFuction}
A $(p{+}1)$-dimensional ($0<p<n$) module~$M$ is called \emph{meta-regular} for the differential function~$L$
if the module~$M$ contains a family $\mathfrak M=\{Q^\Phi\}$ of $p$-dimensional involutive submodules that are regular for~$L$
and parameterized by an arbitrary function $\Phi=\Phi(x,u)$ of all independent and dependent variables, 
and thus $\sco_LM=\ord L$.
\end{definition}

Here and in what follows the parameterization by an arbitrary function is understood as a parameterization
modulo solutions of a system of differential equations in this function.
The \emph{singularity co-order} of the family $\mathfrak M=\{Q^\Phi\}$ is also defined as
the maximum of the singularity co-orders of its elements 
and denoted by $\sco_L\mathfrak M$.

For convenience, the case of singularity co-order $-\infty$ is assumed to include 
parameterized families of ultra-singular modules as well as
meta-singular modules whose involutive $p$-dimensional submodules satisfying the rank condition are ultra-singular.
As shown below, this case in fact cannot occur; see Corollary~\ref{CorollaryOn2DUltraSingularVectorFields}. 

\begin{example}\label{ExamplesOfMetaSingularModules}
For each differential function~$L$ of Example~\ref{ExamplesOfSingularModules},
we present a module~$M$ meta-singular for~$L$ and the corresponding family $\mathfrak M=\{Q^\Phi\}$.
The specific form of the coefficient functions in~$Q^\Phi$ is chosen in order to guarantee involutivity of~$Q^\Phi$.

1. $M=\langle\p_2,\p_3,\p_u\rangle$,
$\mathfrak M=\{Q^\Phi=\langle\p_2-(\Phi_2/\Phi_u)\p_u,\p_3-(\Phi_3/\Phi_u)\p_u\rangle\}$,
$\sco_LM=1$.

2. $M=\langle\p_2,\p_u\rangle$,
$\mathfrak M=\{Q^\Phi=\langle\p_2+\Phi\p_u\rangle\}$,
$\sco_LM=2$.

3. $M=\langle\p_1+\p_4,\p_2,\p_3,\p_u\rangle$,
$\mathfrak M=\{Q^\Phi=\langle\p_2+\p_4-((\Phi_2+\Phi_4)/\Phi_u)\p_u,\p_2-(\Phi_2/\Phi_u)\p_u,\p_3-(\Phi_3/\Phi_u)\p_u\rangle\}$,
$\sco_LM=1$.

4. $M=\langle\p_1+\p_4,\p_2,\p_3,\p_u\rangle$,
$\mathfrak M=\{Q^\Phi=\langle\p_2+\p_4-((\Phi_2+\Phi_4)/\Phi_u)\p_u,\p_2-(\Phi_2/\Phi_u)\p_u,\p_3-(\Phi_3/\Phi_u)\p_u\rangle\}$,
$\sco_LM=1$.

5.
a) $M=\langle\p_2,\p_u\rangle$, $\mathfrak M=\{Q^\Phi=\langle\p_2+\Phi\p_u\rangle\}$, $\sco_LM=1$;
b) $M=\langle\p_2,\p_3,\p_u\rangle$, $\mathfrak M=\{Q^\Phi=\langle\p_2-(\Phi_2/\Phi_u)\p_u,\p_3-(\Phi_3/\Phi_u)\p_u\rangle\}$, $\sco_LM=1$.
\end{example}

The condition that the module~$M$ is involutive is not explicitly included 
in Definitions~\ref{DefinitionOfMetaSingularModuleForDiffFuction} and~\ref{DefinitionOfMetaRegularModuleForDiffFuction}.
By these definitions, the meta-singular (resp.\ meta-regular) module~$M$ should only contain a family of $p$-dimensional involutive submodules
parameterized by an arbitrary function of all variables.
At the same time, in the case $p\geqslant2$ this is equivalent to the fact that the module~$M$ is involutive.

\begin{proposition}\label{PropositionOnMultiDimModulesWithFamilyOfInvModules}
A $(p{+}1)$-dimensional module~$M$, where $p\geqslant2$, contains a family of $p$-dimensional involutive submodules
parameterized by an arbitrary function of all independent and dependent variables
if and only if the module~$M$ is involutive.
\end{proposition}

\begin{proof}
\looseness=-1
Suppose that the module~$M$ contains a family~$\mathfrak M$ of $p$-dimensional involutive submodules
parameterized by an arbitrary function of all variables modulo functions with a smaller number of arguments.
Then we can choose a basis $(Q_0,\dots,Q_p)$ of~$M$ in such a way that
the vector fields $Q_1$, \dots, $Q_p$ generate an involutive submodule of~$M$ from~$\mathfrak M$
and, moreover, commute.
By a change of the variables~$(x,u)$ these vector fields
are reduced to shift operators, $Q_s=\p_s$.
Up to combining with $Q_s$ and multiplying by a nonvanishing function, the vector field~$Q_0$ can be chosen in the form
$Q_0=\xi^{0\iota}(x,u)\p_\iota+\eta^0(x,u)\p_u$, where one of the coefficients $\xi^{0\iota}$, $\iota=p+1,\dots,n$, or $\eta^0$
is equal to 1.
The entire set of $p$-dimensional submodules of~$M$ is partitioned into the subsets
\begin{gather*}
S_0=\{Q^{\bar\theta}=\langle Q_s+\theta^sQ_0,\, s=1,\dots,p\rangle\},\\
S_s=\{Q^{\bar\theta^s}=\langle Q_0,\, \dots,\, Q_{s-1},\, Q_{s'}+\theta^{s'}Q_s,\, s'=s+1,\dots,p\rangle\},
\end{gather*}
where $\bar\theta=(\theta^1,\dots,\theta^p)$, $\bar\theta^s=(\theta^{s+1},\dots,\theta^p)$,
and all $\theta$'s run through the set of smooth functions of~$(x,u)$.%
\footnote{
This can be seen as follows: Suppose that $R=\langle R_1,\dots,R_p\rangle$
is a $p$-dimensional submodule of $M$ and write $R_s = \sum_{k=0}^p\lambda^{sk}Q_k$,
where $\lambda^{sk}$ are smooth functions of $(x,u)$.
Set $\Lambda:=(\lambda^{sk})_{s=1,\dots,p}^{k=0,\dots,p}$ and
$\Lambda':=(\lambda^{sk})_{s=1,\dots,p}^{k=1,\dots,p}$.
We have $\rank\Lambda=p$.
If $\rank\Lambda'=p$, by elementary row operations we may generate a basis of~$R$ that identifies it as
an element of $S_0$.
Otherwise we may transform $\Lambda$ into a form with $\lambda^{10}=1$
and all other entries of the first column and first row vanishing. Applying the same reasoning
as before to the resulting lower right submatrix we obtain either a basis of type $S_1$
or we may again simplify the first column and first row of this submatrix as before.
Thus the claim follows by induction.}
By construction, if $Q^{\bar\theta^s}$ is an involutive module, then the chosen basis elements commute.
This implies that the components of $\bar\theta^s$ satisfy the equations
$Q_0\theta^{s'}=0$, \dots, $Q_{s-1}\theta^{s'}=0$, $s'=s+1,\dots,p$, and can hence be expressed via functions
of at most $n+1-s$ arguments.
Therefore, the module~$M$ contains a family~$\mathfrak M$ of $p$-dimensional involutive submodules
parameterized by an arbitrary function of all variables modulo functions with a smaller number of arguments
if and only if such a family is a subset of~$S_0$.

If a module~$Q^{\bar\theta}$ from~$S_0$ is involutive, by the Frobenius theorem
the overdetermined system of the first-order linear partial differential equations $Q_s\Phi+\theta^sQ_0\Phi=0$
with respect to the unknown function~$\Phi=\Phi(x,u)$ has a solution with $Q_0\Phi\ne0$,
i.e., the parameter-functions $\theta^s$ can be represented in the form $\theta^s=-\Phi_s/Q_0\Phi$.
As the basis elements $Q_s+\theta^sQ_0$ of the module~$Q^{\bar\theta}$ should commute, we have
\begin{gather*}
[Q_s+\theta^sQ_0,Q_{s'}+\theta^{s'}Q_0]=
(\theta^{s'}_s+\theta^sQ_0\theta^{s'}-\theta^s_{s'}-\theta^{s'}Q_0\theta^s)Q_0+\theta^{s'}Q_{0,s}-\theta^sQ_{0,s'}
\\\qquad
=\frac{\Phi_{s'}Q_{0,s}\Phi-\Phi_sQ_{0,s'}\Phi}{(Q_0\Phi)^2}Q_0+\frac{\Phi_s}{Q_0\Phi}Q_{0,s'}-\frac{\Phi_{s'}}{Q_0\Phi}Q_{0,s}=0,
\end{gather*}
where $Q_{0,s}=[Q_s,Q_0]=\xi^{0\iota}_s\p_\iota+\eta^0_s\p_u$.

Suppose that not all the vector fields $Q_{0,s}$ are zero.
Let vector fields $\tilde Q_1$, \dots, $\tilde Q_{\tilde p}$ form a basis of the module $\langle Q_{0,s}\rangle$.
Then the vector fields $Q_0$, $\tilde Q_1$, \dots, $\tilde Q_{\tilde p}$ are linearly independent,
the vector fields $Q_{0,s}$ are represented as $\sum_{\tilde s=1}^{\tilde p}\lambda^{s\tilde s}\tilde Q_{\tilde s}$
for some smooth functions $\lambda^{s\tilde s}=\lambda^{s\tilde s}(x,u)$ and
the above commutation relations imply the following system for the function~$\Phi$:
\[
\Phi_{s'}Q_{0,s}\Phi-\Phi_sQ_{0,s'}\Phi=0, \quad
\lambda^{s'\tilde s}\Phi_s=\lambda^{s\tilde s}\Phi_{s'}, \quad \tilde s=1,\dots,\tilde p.
\]
As some of the coefficients $\lambda^{s\tilde s}$ are necessarily nonzero,
in this case the function~$\Phi$ runs at most through the solution set of a system of first-order linear partial differential equations
and, therefore, the number of its arguments is in fact less than~$n+1$.
This contradicts the existence of a family of $p$-dimensional involutive submodules of~$M$
parameterized by an arbitrary function of all variables modulo functions with a smaller number of arguments.

As a result, all the vector fields $Q_{0,s}=[Q_s,Q_0]$ are zero, i.e., the module~$M$ is involutive.

Conversely, if the module~$M$ is involutive, we choose a basis which consists of commuting vector fields $Q_0$,~\dots, $Q_p$.
Each submodule $Q^\Phi=\langle Q_s-(Q_s\Phi)/(Q_0\Phi)Q_0\rangle$ is involutive and of dimension~$p$.
Here $\Phi=\Phi(x,u)$ runs through the set of smooth functions of~$(x,u)$ with $Q_0\Phi\ne0$, i.e.,
the complement of the solution set of the equation $Q_0\Phi=0$.
The general solution of this equation is parameterized by a single function of~$n$ arguments.
The modules $Q^\Phi$ and $Q^{\tilde\Phi}$ associated with the functions~$\Phi$ and~$\tilde\Phi$ coincide
if and only if $(Q_0\Phi)Q_s\tilde\Phi=(Q_s\Phi)Q_0\tilde\Phi$.
We will call such functions~$\Phi$ and~$\tilde\Phi$ equivalent.
Thus, the set traversed by~$\Phi$ should additionally be factorized with respect to this equivalence relation.
A necessary and sufficient condition for the equivalence of functions~$\Phi$ and~$\tilde\Phi$ is
that there exists a smooth function~$F$ of $n-p+1$ arguments for which
$\tilde\Phi=F(\omega^1,\dots,\omega^{n-p},\Phi)$,
where $\omega^1$, \dots, $\omega^{n-p}$ form a complete set of functionally independent solutions
of the system $Q_0\omega=0$, $Q_s\omega=0$.
As the number of arguments of~$F$ is less than $n+1$, the factorization does not affect the degree of arbitrariness of~$\Phi$.
\end{proof}

\begin{remark}
Families of involutive submodules belonging to the set~$S_0$ can be parameterized in different ways.
Thus, in the above proof these submodules are parameterized via
the representation of the coefficients $\theta^s$ in the form $\theta^s=-Q_s\Phi/Q_0\Phi$,
and the set traversed by $\Phi=\Phi(x,u)$ should be factorized with respect to an equivalence relation.
Alternatively, we can choose, e.g., any of the $\theta^s$ instead of~$\Phi$ as a parameterized function.
After $\Psi^1=\theta^1$ is chosen as such a function, the condition $[Q_1+\theta^1Q_0,Q_2+\theta^2Q_0]=0$ implies
the first-order linear partial differential equation $Q_1\theta^2+\theta^1Q_0\theta^2=Q_2\theta^1+\theta^2Q_0\theta^1$
with respect to~$\theta^2$, whose general solution is parameterized by an arbitrary function~$\Psi^2$ of $n$~arguments.
In the same way, the condition $[Q_i+\theta^iQ_0,Q_3+\theta^3Q_0]=0$, $i=1,2$,
gives the system of two first-order linear partial differential equations
$Q_i\theta^3+\theta^iQ_0\theta^3=Q_3\theta^i+\theta^3Q_0\theta^i$.
This system is associated with the module generated by the vector fields
$Q_i+\theta^iQ_0+(Q_3\theta^i+\theta^3Q_0\theta^i)\p_{\theta^3}$,
which is involutive in view of the equation for~$\theta^2$.
It follows from the Frobenius theorem that the general solution
of the system is parameterized by an arbitrary function~$\Psi^3$ of $n-1$~arguments.
Iterating the procedure for each~$\theta^s$, we bijectively parameterize involutive modules from the submodule set~$S_0$
by the tuple $(\Psi^1,\dots,\Psi^p)$, where $\Psi^s$ runs through the set of smooth functions of $n+2-s$ arguments.
\end{remark}

\begin{corollary}\label{CorollaryOnInvolitivityOfMultiDimSingularModules}
Any $(p{+}1)$-dimensional meta-singular (resp.\ meta-regular) module~$M$ for a differential function~$L$, where $p\geqslant2$,
is involutive.
\end{corollary}

\begin{proposition}\label{PropositionOnReducedFormOfFamiliesOfSingularModules}
A $(p{+}1)$-dimensional module~$M$, where $p\geqslant2$, is meta-singular (resp.\ meta-regular) for a differential function~$L$
if and only if
it contains, up to point transformations in the space of $(x,u)$,
a family~$\mathfrak M$ of $p$-dimensional involutive submodules singular (resp.\ regular) for~$L$
of the form \[Q^\Phi=\langle \p_s-(\Phi_s/\Phi_u)\p_u,\, s=1,\dots,p\rangle\]
parameterized by an arbitrary smooth function $\Phi=\Phi(x,u)$ of all independent and dependent variables with $\Phi_u\ne0$.
The singularity co-order of the family~$\mathfrak M$ for the differential function~$L$
coincides with that of the entire module~$M$, $\sco_L\mathfrak M=\sco_LM$.
\end{proposition}

\begin{proof}
Suppose that the module~$M$ is meta-singular (resp.\ meta-regular) for a differential function~$L$.
As the module~$M$ is involutive in view of Corollary~\ref{CorollaryOnInvolitivityOfMultiDimSingularModules},
we can choose a basis of~$M$ consisting of commuting vector fields $Q_0$, \dots, $Q_p$
that satisfy the rank condition $\rank(\xi^{si})=p$.
A~change of the variables~$(x,u)$ reduces these vector fields to shift operators, $Q_s=\p_s$ and $Q_0=\p_u$.
A submodule of the form $Q^{\bar\theta}=\langle Q_s+\theta^sQ_0,\, s=1,\dots,p\rangle$ of~$M$ is involutive
if and only if
the coefficients~$\theta^s$ can be represented in the form $\theta^s=-\Phi_s/\Phi_u$
for some function~$\Phi=\Phi(x,u)$ with $\Phi_u\ne0$,
cf.\ the proof of Proposition~\ref{PropositionOnMultiDimModulesWithFamilyOfInvModules}.
Thus, the family $\mathfrak M=\{Q^\Phi=\langle \p_s-(\Phi_s/\Phi_u)\p_u,\, s=1,\dots,p\rangle\}$,
parameterized by the arbitrary function $\Phi$ with $\Phi_u\ne0$, is of the required form.

The converse statement and the equality with singularity co-orders are obvious in the regular case 
and hence it suffices to prove them only in the singular case. 

Let the module~$M$ contain, up to point transformations in the space of $(x,u)$,
a family $\mathfrak M=\{Q^\Phi\}$ of the required form and let $\sco_L\mathfrak M=k<r=\ord L$.
Hence we have that $\sco_LQ^\Phi\leqslant k$ for any allowed value of the parameter-function~$\Phi$.
In the initial coordinates, the basis elements of a submodule~$Q^\Phi$ take the form
$Q_s-(Q_s\Phi)/(Q_0\Phi)Q_0$, where $Q_0$, \dots, $Q_p$ are commuting vector fields.
It suffices to prove that any $p$-dimensional involutive submodule~$P$ of~$M$,
which satisfies the rank condition and does not belong to the family $\{Q^\Phi\}$, is singular for~$L$.
Up to a permutation of the vector fields $Q_1$, \dots, $Q_p$,
a basis of~$P$ consists of the vector fields $Q_0$ and $Q_{s'}+\theta^{s'}Q_1$, $s'=2,\dots,p$,
where the coefficients $\theta^{s'}$ are smooth functions of~$(x,u)$.
For convenience, by a point transformation of the variables~$(x,u)$
we reduce the vector fields $Q_0$, $Q_2$, \dots, $Q_p$ and $Q_1$
to the shift operators $\p_1$, $\p_2$, \dots, $\p_p$ and $\p_u$, respectively.
As the submodule~$P$ is involutive, the coefficients $\theta^{s'}$ possess the representation
$\theta^{s'}=-\Psi_{s'}/\Psi_u$, where $\Psi$ is a smooth function of~$x_2$, \dots, $x_n$ and $u$
with $\Psi_u\ne0$.
Consider the family of involutive modules of the~form
\[Q^\varepsilon=\langle Q_0+\varepsilon Q_1,\, Q_{s'}+\theta^{s'\varepsilon}Q_1,\, s'=2,\dots,p\rangle\]
parameterized by a constant~$\varepsilon$ running through a neighborhood of zero.
Here the coefficients $\theta^{s'\varepsilon}$ are obtained via replacing
the argument~$u$ in~$\theta^{s'}$ by $u-\varepsilon x_1$,
$\theta^{s'\varepsilon}=\theta^{s'}(x_2,\dots,x_n,u-\varepsilon x_1)$.
For each nonzero value of~$\varepsilon$ the module~$Q^\varepsilon$ belongs to the family $\{Q^\Phi\}$.
This is obvious after the transition from the chosen basis to the basis
\[(Q_1+\varepsilon^{-1}Q_0,\,Q_{s'}-\varepsilon^{-1}\theta^{s'\varepsilon}Q_0,\, s'=2,\dots,p).\]
Therefore, any module~$Q^\varepsilon$ with $\varepsilon\ne0$ is singular for the differential function~$L$.
Consider the differential function~$\hat L^\varepsilon$ which is associated with~$L$
on the manifold~$\smash{\mathcal Q^\varepsilon_{\ssl r\ssr}}$ via the exclusion of the derivatives~$u_\alpha$
with $\alpha_1+\dots+\alpha_p>0$ from~$L$ using the equations $u_1=\varepsilon$, $u_{s'}=\theta^{s'\varepsilon}$
and their differential consequences.
The function~$\hat L^\varepsilon$ is smooth in the parameter~$\varepsilon$, the variables~$x$ and
derivatives of~$u$ with respect to~$x_{p+1}$,~\dots,~$x_n$.
As $\ord\hat L^\varepsilon\leqslant k$ for any $\varepsilon\ne0$, the same statement is true for $\varepsilon=0$
by continuity.
This means that the submodule $P=Q^\varepsilon|_{\varepsilon=0}$ is singular for~$L$, and $\sco_LP\leqslant k=\sco_L\mathfrak M$.

It is obvious that $\sco_LM\geqslant\sco_L\mathfrak M$.
As any $p$-dimensional involutive submodule~$Q$ of~$M$ satisfies the inequality $\sco_LQ\leqslant \sco_L\mathfrak M$,
we obtain that $\sco_LM=\sco_L\mathfrak M$.
\end{proof}

The case of two-dimensional meta-singular (resp.\ meta-regular) modules is special.
As any one-dimensional module of vector fields is involutive,
the fact that a two-dimensional module is meta-singular (resp.\ meta-regular) for a differential function does not imply
that this module is involutive.
More specifically, the conclusion of Proposition \ref{PropositionOnMultiDimModulesWithFamilyOfInvModules} is not true if $p=1$.
This is why the reduced form of submodules of a two-dimensional meta-singular (resp.\ meta-regular) module depends on whether
this module is involutive or not.

\begin{proposition}\label{PropositionOnReducedFormOfFamiliesOfOneDimSingularModules}
A two-dimensional module~$M$ is meta-singular (resp.\ meta-regular) for a differential function~$L$
if and only if
it contains, up to point transformations in the space of $(x,u)$,
a family~$\mathfrak M$ of one-dimensional submodules singular (resp.\ regular) for~$L$
with basis vector fields in a reduced form
parameterized by an arbitrary smooth function $\theta=\theta(x,u)$.
The reduced form is \[\p_1+\theta\p_u\quad\mbox{or}\quad\p_1+u\p_2+\xi^3\p_3+\dots+\xi^n\p_n+\theta\p_u\]
if the module~$M$ is involutive or not involutive, respectively.
Here $\xi^i=\xi^i(x,u)$, $i=3,\dots,n$ are fixed smooth functions.
The singularity co-order of the family~$\mathfrak M$ for the differential function~$L$
coincides with that of the entire module~$M$, $\sco_L\mathfrak M=\sco_LM$.
\end{proposition}

In order to make the further consideration for $p=1$ consistent with the case of~$p\geqslant2$,
the parameter-function $\theta$ can be represented in the form $\theta=-\Phi_1/\Phi_u$,
where $\Phi=\Phi(x,u)$ is an arbitrary smooth function with $\Phi_u\ne0$.

\begin{proof}
If the module~$M$ is involutive, the proof is similar to that of Proposition~\ref{PropositionOnMultiDimModulesWithFamilyOfInvModules}.
We therefore only consider the case when the module~$M$ is not involutive.

Let the module~$M$ be meta-singular (resp.\ meta-regular) for a differential function~$L$.
We choose a basis $(Q_0,Q_1)$ of~$M$ such that the vector field~$Q_1$ satisfies the rank condition \mbox{$\rank(\xi^{1i})=1$}
and reduce the vector field~$Q_0$ by a change of the variables~$(x,u)$ to the shift operator with respect to~$u$, $Q_0=\p_u$.
Up to permutation of the variables~$x_1$, \dots, $x_p$, we can assume that $\xi^{11}\ne0$.
Then we replace the basis element~$Q_1$ by $(\xi^{11})^{-1}(Q_1-\eta^1Q_0)$ in order to set $\eta^1=0$ and $\xi^{11}=1$.
As the module~$M$ is not involutive, the commutator $[Q_0,Q_1]=\xi^{12}_u\p_2+\dots+\xi^{1n}_u\p_n$ does not vanish.
Hence we can assume up to permutation of the variables~$x_2$, \dots, $x_p$ that $\xi^{12}_u\ne0$.
The change of variables $\tilde x^s=x^s$ and $\tilde u=\xi^{12}(x,u)$
with the simultaneous replacement of~$Q_0$ by $(\xi^{12}_u)^{-1}Q_0$
reduces the basis elements of~$M$ to the form $Q_0=\p_u$ and~$Q_1=\p_1+u\p_2+\xi^3\p_3+\dots+\xi^n\p_n$.
Then the family $\mathfrak M=\{\langle Q_1+\theta Q_0\rangle\}$ of one-dimensional submodules of~$M$
singular for~$L$,
where the parameter~$\theta=\theta(x,u)$ runs through the set of smooth functions of all independent and dependent variables,
is of the required form.

The converse statement and the equality with singularity co-orders are again obvious in the regular case 
and hence it suffices to prove them only in the singular case. 

Let the module~$M$ contain, up to point transformations in the space of $(x,u)$,
a family $\mathfrak M$ of the required form with $\sco_L\mathfrak M=k<r=\ord L$. 
Hence we have that
$\sco_L\langle Q_1+\theta Q_0\rangle\leqslant k$ for any value of the parameter-function~$\theta$.
After returning to the initial coordinates, it suffices to prove that the submodule~$\langle Q_0\rangle$ of~$M$ is singular for~$L$
if the vector field~$Q_0$ also satisfies the rank condition in these coordinates.
For convenience we reduce~$Q_0$ by a change of coordinates to the shift operator~$\p_1$.
Consider the family of modules of the form $Q^\varepsilon=\langle Q_0+\varepsilon(Q_1-\xi^{11}Q_0)\rangle$
parameterized by a constant~$\varepsilon$ running though a neighborhood of zero.
For each nonzero value of~$\varepsilon$ the module~$Q^\varepsilon$ belongs to the family~$\mathfrak M$
as in this case we have $Q^\varepsilon=\langle Q_1+(\varepsilon^{-1}-\xi^{11})Q_0\rangle$.
Therefore, any module~$Q^\varepsilon$ with $\varepsilon\ne0$ is singular for the differential function~$L$.
Consider the differential function~$\hat L^\varepsilon$ which is associated with~$L$
on the manifold~$\smash{\mathcal Q^\varepsilon_{\ssl r\ssr}}$ via the exclusion of the derivatives~$u_\alpha$
with $\alpha_1>0$ from~$L$ using the equation $u_1=\varepsilon(\eta^{11}-\xi^{12}u_2-\dots-\xi^{1n}u_n)$ and its differential consequences.
The function~$\hat L^\varepsilon$ is smooth in the totality of the parameter~$\varepsilon$, the variables~$x$ and
derivatives of~$u$ with respect to~$x_2$,~\dots,~$x_n$.
As the order of~$\hat L^\varepsilon$ is not greater than~$k$ for any nonzero~$\varepsilon$, the same statement is true for $\varepsilon=0$ by continuity.
This means that the submodule $\langle Q_0\rangle=Q^\varepsilon|_{\varepsilon=0}$ is singular for~$L$,
and $\sco_L\langle Q_0\rangle\leqslant k=\sco_L\mathfrak M$.

The set of one-dimensional submodules of~$M$ is exhausted by~$\langle Q_0\rangle$
and the elements of the family~$\mathfrak M$. Hence $\sco_L\mathfrak M=\sco_LM$.
\end{proof}

\section{Differential functions admitting meta-singular modules}
\label{SectionOnDiffFunctionsAdmittingMeta-singularModules}

Up to point transformations, we can describe the general form of differential functions admitting meta-singular modules of vector fields.

\begin{theorem}\label{TheoremOnDiffFunctionsWithMetaSingularModule}
An $r$th order ($r>0$) differential function~$L$ with one dependent and $n$ independent variables possesses
a co-order $k$ meta-singular $(p{+}1)$-dimensional module of vector fields
with $0\leqslant k<r$ and $0<p<n$
if and only if
it can be represented, up to point transformations, in the form
\begin{equation}\label{Eq2DWithSingularVectorFieldsOfCoOrderK}
L=\bar L(x,\Omega_{r,k,p}),
\end{equation}
where
$\Omega_{r,k,p}=\big(\omega_\alpha,\,|\alpha|\leqslant r,\,\alpha_{p+1}+\dots+\alpha_n\leqslant k\big)$,
and the function~$\bar L$ essentially depends on some $\omega_\alpha$ with $\alpha_{p+1}+\dots+\alpha_n=k$.
Here $\omega_\alpha=u_\alpha$ or, only for the case $p=1$,
$\omega_\alpha=D_2^{\alpha_2}\cdots D_n^{\alpha_n}(D_1+uD_2+\xi^3D_3+\dots+\xi^nD_n)^{\alpha_1} u$
for some fixed smooth functions $\xi^i=\xi^i(x,u)$, $i=3,\dots,n$.
\end{theorem}

\begin{proof}
Suppose that a differential function~$L$ possesses a $(p{+}1)$-dimensional co-order~$k$ meta-singular module~$M$.
Up to combining basis elements and change of variables, a basis of the module~$M$ consists of
either the vector fields $Q_s=\p_s$ and $Q_0=\p_u$ or, if $p=1$ and the module~$M$ is not involutive,
the vector fields $Q_1=\p_1+u\p_2+\xi^3\p_3+\dots+\xi^n\p_n$ and $Q_0=\p_u$,
cf.\ Propositions~\ref{PropositionOnReducedFormOfFamiliesOfSingularModules}
and~\ref{PropositionOnReducedFormOfFamiliesOfOneDimSingularModules}.
Although the form of the initial differential function~$L$ will also be transformed by the change of variables,
for simplicity we will continue to use the old notations for all new values.

We choose a family~$\mathfrak M=\{Q^\Phi=\langle Q_s-(\Phi_s/\Phi_u)\p_u,\, s=1,\dots,p\rangle\}$
of $p$-dimensional involutive submodules of~$M$ that are parameterized by an arbitrary function $\Phi=\Phi(x,u)$.
Then we fix an arbitrary point $\smash{z^0=(x^0,u_{\ssl r\ssr}^0)\in \mathrm J^r}$
and consider the values of the parameter-function~$\Phi$ for which \smash{$z^0\in\mathcal Q^\Phi_{\ssl r\ssr}$}.
(Here by \smash{$\mathcal Q^\Phi_{\ssl r\ssr}$} we denote the manifold that is defined in the jet space~$\mathrm J^r$
by the system with the same notation~\smash{$\mathcal Q^\Phi_{\ssl r\ssr}$}
consisting of the independent differential consequences of the characteristic system~$\mathcal Q^\Phi$
up to equation order~$r$.)
This condition for~$\Phi$ implies that the values of the derivatives of~$\Phi$ with respect to only~$x_1$,~\dots, $x_n$
at the point $(x^0,u^0)$, which contain differentiation with respect to some~$x_s$,
are expressed via $\smash{u_{\ssl r\ssr}^0}$ and values of derivatives of~$\Phi$ in $(x^0,u^0)$,
containing differentiation with respect to~$u$.
The latter values are not constrained.
For instance, if the module~$M$ is involutive, we get
\begin{gather*}
\Phi_s(x^0,u^0)=-u_s^0\Phi_u(x^0,u^0), \\
\Phi_{si}(x^0,u^0)=-u_{si}^0\Phi_u(x^0,u^0)-u_i^0\Phi_{su}(x^0,u^0)-u_s^0\Phi_{iu}(x^0,u^0)-u_s^0u_i^0\Phi_{uu}(x^0,u^0),\quad \dots\, .
\end{gather*}

We introduce the new coordinates $(x_i,\omega_\alpha,|\alpha|\leqslant r)$
in $\mathrm J^r$ instead of the standard ones $(x_i,u_\alpha,|\alpha|\leqslant r)$.
If the module~$M$ is involutive, this change of coordinates is just a re-labeling of variables
in order to guarantee consistency with the special case of non-involutive modules for $p=1$.
In the latter case, this is a valid change of coordinates since the Jacobian matrix $(\p\omega_\alpha/\p u_{\alpha'})$ is nondegenerate:
it is a triangular matrix with all diagonal entries equal to $1$ if we implement the graded lexicographic order of multi-indices,
\[
\alpha\prec\beta\ \Leftrightarrow\ |\alpha|<|\beta|\vee(|\alpha|=|\beta|\wedge\alpha_1<\beta_1)\vee(|\alpha|=|\beta|\wedge\alpha_1=\beta_1\wedge\alpha_2<\beta_2)\vee\cdots.
\]

Denote by $\hat L^\Phi$ the differential function obtained from~$L$ by the procedure of excluding,
in view of the system $\smash{\mathcal Q^\Phi_{\ssl r\ssr}}$,
the derivatives of~$u$ that involve differentiations with respect to~$x_s$
and are thus assumed $\mathcal Q^\Phi$-principal.
For a multi-index $\alpha=(\alpha_1,\dots,\alpha_n)$ we set
$\check\alpha=(\alpha_1,\dots,\alpha_p)$ and $\hat\alpha=(\alpha_{p+1},\dots,\alpha_n)$,
i.e., $\alpha=(\check\alpha,\hat\alpha)$. The symbol~$\check0$ denotes the tuple of $p$~zeros.
As $Q^\Phi$ is a co-order~$k$ singular module for~$L$,
the function $\hat L^\Phi$ does not depend on the derivatives $u_{\ssl\check0,\hat\alpha\ssr}$, $|\hat\alpha|=k+1,\dots,r$.
We use this condition step-by-step, starting from the greatest value of~$|\hat\alpha|$
and re-writing the derivatives in the new coordinates of~$\mathrm J^r$ and in terms of~$L$.
In the course of this procedure, we take into account the equality
$\omega_\beta=\psi^\beta[u]$ satisfied on the manifold $\smash{\mathcal Q^\Phi_{\ssl r\ssr}}$
for each multi-index~$\beta$ with $|\beta|\leqslant r$.
Here the differential function $\psi^\beta=\psi^\beta[u]$ is defined by the equality
$\psi^\beta=D_{p+1}^{\beta_{p+1}}\cdots D_n^{\beta_n}(Q^\Phi_1)^{\beta_1}\cdots(Q^\Phi_p)^{\beta_p}u$,
and hence it is of $\smash{|\hat\beta|}$th order and possesses the representation
\[\psi^\beta=(\p_u(Q^\Phi_1)^{\beta_1}\cdots(Q^\Phi_p)^{\beta_p}u\bigr)u_{\ssl\check0,\hat\beta\ssr}+\tilde\psi^\beta[u]\]
with some differential function $\tilde\psi^\beta=\tilde\psi^\beta[u]$ of order less than~$|\hat\beta|$.
Therefore, for each $\alpha$ with $\check\alpha=\check0$ and $|\hat\alpha|\leqslant r$ the chain rule implies that
\begin{equation}\label{EqnChainRuleForDiffFunctionsWithMetaSingularModule}
\hat L^\Phi_{u_\alpha}(z^0)=\sum_{\beta\colon |\beta|\leqslant r,\,|\hat\beta|\geqslant|\hat\alpha|}L_{\omega_\beta}(z^0)\psi^\beta_{u_\alpha}(z^0).
\end{equation}

\noprint{
Note that
\[
\omega_\alpha=D_2^{\alpha_2}\cdots D_n^{\alpha_n}(D_1+uD_2+\xi^3D_3+\dots+\xi^nD_n)^{\alpha_1} u
=D_2^{\alpha_2}\cdots D_n^{\alpha_n}(Q^\Phi_1)^{\alpha_1}u
\]
on $\smash{\mathcal Q^\Phi_{\ssl r\ssr}}$ in the special case with $p=1$.
We also have a similar representation for $\omega_\alpha$ on $\smash{\mathcal Q^\Phi_{\ssl r\ssr}}$ if the module~$M$ is involutive:
\[
\omega_\alpha=D_{p+1}^{\alpha_{p+1}}\cdots D_n^{\alpha_n}D_1^{\alpha_1}\cdots D_p^{\alpha_p} u
=D_{p+1}^{\alpha_{p+1}}\cdots D_n^{\alpha_n}(Q^\Phi_1)^{\alpha_1}\cdots(Q^\Phi_p)^{\alpha_p}u
\]
These representations for $\omega_\alpha$ implies that
$\omega_\alpha=(\p_u(Q^\Phi_1)^{\alpha_1}\cdots(Q^\Phi_p)^{\alpha_p}u\bigr)u_{(\check0,\hat\alpha)}+\psi^\alpha[u]$
on $\smash{\mathcal Q^\Phi_{\ssl r\ssr}}$,
where $\psi^\alpha=\psi^\alpha[u]$ is a differential function of order less than~$|\hat\alpha|$.
}

Thus, in the new coordinates for each $\alpha$ with $\check\alpha=\check0$ and $|\hat\alpha|=r$
the equation $\smash{\hat L^\Phi_{u_\alpha}(z^0)=0}$
can be written in the form $\smash{L_{\omega_\alpha}(z^0)=0}$.
Indeed, in this case we have that $\psi^\beta_{u_\alpha}=1$ if $\beta=\alpha$ and $\psi^\beta_{u_\alpha}=0$ otherwise.
This completes the first step.

In the second step we fix a value of~$\alpha$ with $\check\alpha=\check0$ and $|\hat\alpha|=r-1$.
As $\smash{L_{\omega_\beta}(z^0)=0}$ if $\check\beta=\check0$ and $|\hat\beta|=r$,
the summation multi-index in~\eqref{EqnChainRuleForDiffFunctionsWithMetaSingularModule} with the fixed~$\alpha$
can be assumed to run through the set $B_1=\{\beta\mid |\check\beta|\leqslant 1,\,|\hat\beta|=r-1\}$.
The derivative $\psi^\beta_{u_\alpha}$ is equal to~$1$, $\p_uQ^\Phi_su$ and~$0$ for
$\beta=(\check0,\hat\alpha)$, $\beta=(\delta_s,\hat\alpha)$ and all other values of~$\beta$ from~$B_1$, respectively.
Here $\delta_s$ is the $p$-tuple with the $s$th entry equal to 1 and the other entries equal to 0.
Therefore, the equation $\smash{\hat L^\Phi_{u_{(\check0,\hat\alpha)}}(z^0)=0}$ implies that
\[
L_{\omega_{(\check0,\hat\alpha)}}(z^0)+L_{\omega_{(\delta_s,\hat\alpha)}}(z^0)
\bigl(\p_uQ^\Phi_su\bigr)\big|_{(x,u)=(x^0,u^0)}=0.
\]
Note that $\p_uQ^\Phi_su=-(\Phi_s/\Phi_u)_u$.
We split with respect to the value $\Phi_{su}(x^0,u^0)$ as it is unconstrained.
Thereby we arrive at the equations $L_{\omega_{(\check0,\hat\alpha)}}(z^0)=0$ and $L_{\omega_{(\delta_s,\hat\alpha)}}(z^0)=0$.

Iterating this procedure, before the $\mu$th step, $\mu\in\{3,\dots,r-k\}$, we derive the equations
$L_{\omega_\beta}(z^0)=0$, where the multi-index~$\beta$ runs through values
for which $r-\mu+2\leqslant|\hat\beta|\leqslant r$ and $|\check\beta|\leqslant r-|\hat\beta|$.
Then for each fixed value of~$\alpha$ with $\check\alpha=\check0$ and $|\hat\alpha|=r-\mu+1$
the summation multi-index in~\eqref{EqnChainRuleForDiffFunctionsWithMetaSingularModule}
can be assumed to run through the set $B_\mu=\{\beta\mid |\check\beta|\leqslant \mu-1,\,|\hat\beta|=r-\mu+1\}$.
For $\beta\in B_\mu$ the derivative $\psi^\beta_{u_\alpha}$equals
$\p_u(Q^\Phi_1)^{\beta_1}\cdots(Q^\Phi_p)^{\beta_p}u$ if $\hat\beta=\hat\alpha$
and is zero otherwise.
Therefore, the equation $\smash{\hat L^\Phi_{u_{(\check0,\hat\alpha)}}(z^0)=0}$ implies the condition
\[
\sum_{|\check\beta|\leqslant\mu-1,\, \hat\beta=\hat\alpha}
L_{\omega_\beta}(z^0)\bigl(\p_u(Q^\Phi_1)^{\beta_1}\cdots(Q^\Phi_p)^{\beta_p}u\bigr)\big|_{(x,u)=(x^0,u^0)}=0.
\]
The values $(\p_1^{\beta_1}\cdots\p_p^{\beta_p}\Phi_u)(x^0,u^0)$, $0<|\check\beta|\leqslant\mu-1$, are unconstrained.
Hence by splitting with respect to them or, equivalently, by splitting with respect to
$\smash{\bigl(\p_u(Q^\Phi_1)^{\beta_1}\cdots(Q^\Phi_p)^{\beta_p}u\bigr)\big|_{(x,u)=(x^0,u^0)}}$, $0<|\check\beta|\leqslant\mu-1$,
we obtain the equations $L_{\omega_\beta}(z^0)=0$, $|\hat\beta|=r-\mu+1$ and $|\check\beta|\leqslant\mu-1$.

Finally, after the $(r-k)$th step we derive the system
$L_{\omega_\alpha}(z^0)=0$, $|\hat\alpha|>k$ and $|\alpha|\leqslant r$,
which implies the condition~\eqref{Eq2DWithSingularVectorFieldsOfCoOrderK}.

Conversely, let an $r$th order differential function~$L$ be of the form~\eqref{Eq2DWithSingularVectorFieldsOfCoOrderK}
(after a point transformation).
For an arbitrary smooth function~$\Phi=\Phi(x,u)$ with $\Phi_u\ne0$, we consider
the involutive module~$Q^\Phi$ generated by either the vector fields $Q^\Phi_s=\p_s-(\Phi_s/\Phi_u)\p_u$
in the general case or the vector field $Q^\Phi_1=\p_1+u\p_2+\xi^3\p_3+\dots+\xi^n\p_n-(\Phi_s/\Phi_u)\p_u$
in the special case with $p=1$.
Using~$\bar L$ and~$Q^\Phi$, we construct the differential function $\tilde L^\Phi=\bar L(x,\tilde\Omega_{r,k,p})$, where
\[
\tilde\Omega_{r,k,p}=\big(\omega_\alpha=D_{p+1}^{\alpha_{p+1}}\cdots D_n^{\alpha_n}(Q^\Phi_1)^{\alpha_1}\cdots(Q^\Phi_p)^{\alpha_p}u,\,
|\hat\alpha|\leqslant k,\,|\alpha|\leqslant r\big).
\]
By construction, $\ord \tilde L^\Phi\leqslant k$ for all values of the parameter-function~$\Phi$ with $\Phi_u\ne0$.
Moreover, $\ord \tilde L^\Phi=k$ for almost all values of this parameter-function except those
which satisfy a system of differential equations.
As
\[
L|_{\mathcal Q^\Phi_{\ssl r\ssr}}=\tilde L^\Phi|_{\mathcal Q^\Phi_{\ssl r\ssr}},
\]
where $\Phi$ runs through the set of smooth functions of $(x,u)$ with nonvanishing derivatives with respect to~$u$,
the family $\mathfrak M=\{Q^\Phi\}$ is a co-order~$k$ singular family of $p$-dimensional involutive modules
for the differential function~$L$ in the new coordinates.
We return to the old coordinates.
In view of Proposition~\ref{PropositionOnReducedFormOfFamiliesOfSingularModules} if $p\geqslant2$
or Proposition~\ref{PropositionOnReducedFormOfFamiliesOfOneDimSingularModules} if $p=1$,
the module of vector fields which contains the family $\mathfrak M$ is
a co-order~$k$ meta-singular $(p{+}1)$-dimensional module for the differential function~$L$.
\end{proof}

Excluding the special case of two-dimensional non-involutive meta-singular modules,
the result presented in Theorem~\ref{TheoremOnDiffFunctionsWithMetaSingularModule}
can be formulated in the following way:
A differential function with one dependent and $n$~independent variables admits
a co-order~$k$ meta-singular $(p{+}1)$-dimensional involutive module~$M$
if and only if it can be reduced by a point transformation of the variables to a differential function~$L$, 
where the differentiation with respect to $n-p$ fixed independent variables
in each derivative of the dependent variable among (essential) arguments of~$L$
is of aggregate order not greater than~$k$.

\begin{corollary}\label{CorollaryOn2DUltraSingularVectorFields}
Any differential function with one dependent and $n$ independent variables of positive order does not admit
any meta-singular $(p{+}1)$-dimensional ($0<p<n$) module of singularity co-order~$-\infty$.
\end{corollary}

\begin{proof}
Assume, to the contrary, that 
there exists an $r$th order ($r>0$) differential function~$L$ with one dependent and $n$ independent variables 
that possesses a meta-singular $(p{+}1)$-dimensional ($0<p<n$) module~$M$ of singularity co-order~$-\infty$.
Following the proof of Theorem~\ref{TheoremOnDiffFunctionsWithMetaSingularModule}, 
we re-combine basis elements of~$M$ and change variables 
in order to reduce the chosen basis of~$M$ to the same canonical form as in that proof. 
We also choose a family~$\mathfrak M$ of $p$-dimensional involutive submodules of~$M$ 
that are parameterized by an arbitrary function $\Phi$ of $(x,u)$. 
Repeating the further steps of the proof leads to the conclusion 
that the differential function~$L$ is of the form~\eqref{Eq2DWithSingularVectorFieldsOfCoOrderK} with $k=0$. 
Since $\ord L=r$, the function~$\bar L$ in this representation essentially depends on some $\omega_\alpha$'s with $|\alpha|=r$. 
For each $\alpha$ with $|\alpha|=r$, the expression for $\omega_\alpha$ in view of the system $\smash{\mathcal Q^\Phi_{\ssl r\ssr}}$ 
involves the derivative~$\Phi_\alpha$, and the analogous expressions for other $\omega$'s do not. 
Therefore, the differential function~$\hat L^\Phi$ is of order 0 
for all values of the parameter-function~$\Phi$ except solutions of the equation~$\hat L^\Phi_u=0$,
treated as an $r+1$th order differential equation with respect to~$\Phi$.
This property is preserved by point transformations because of the arbitrariness of~$\Phi$. 
Therefore, in the initial variables we also have $\sco_LM\geqslant0$, 
which contradicts the assumption.
\end{proof}

\begin{remark}\label{RemarkOnSubmodulesOfSingularityCo-ordersLessThanSingularityCo-orderOfMeta-singularModule}
Clearly a meta-singular  (resp.\ meta-regular) $(p{+}1)$-dimensional module~$M$ for a differential function~$L$ 
may contain $p$-dimensional involutive modules whose singularity co-orders
are less than the singularity co-order~$\sco_LM=:k$ of the entire module~$M$.
Consider a family~$\mathfrak M=\{Q^\Phi=\langle Q_s-(\Phi_s/\Phi_u)Q^0,\, s=1,\dots,p\rangle\}$
of $p$-dimensional involutive submodules of~$M$ which are parameterized by an arbitrary function $\Phi=\Phi(x,u)$
and assume that $\sco_L\mathfrak M=k$.
Here the basis vector fields~$Q_s$ and~$Q^0$ of~$M$ are assumed to be reduced to the form presented
in the beginning of the proof of Theorem~\ref{TheoremOnDiffFunctionsWithMetaSingularModule}.
Then the values of~$\Phi$ for which $\sco_LQ^\Phi<k$ are solutions of the system
\[
\sum_{|\check\alpha|\leqslant r-k}
\bar L_{\omega_\alpha}(x,\tilde\Omega_{r,k,p})\big(\p_u(Q^\Phi_1)^{\alpha_1}\cdots(Q^\Phi_p)^{\alpha_p}u\big)=0, \quad |\hat\alpha|=k,
\]
where $\bar L$ and $\tilde\Omega_{r,k,p}$ are defined in Theorem~\ref{TheoremOnDiffFunctionsWithMetaSingularModule} and its proof, respectively.
In other words, the regular values of~$\Phi$ associated with the submodules of the maximal singularity co-order~$k$ in~$\mathfrak M$
satisfy, for some $\hat\alpha$ with $|\hat\alpha|=k$, the inequality
\[
\sum_{|\check\alpha|\leqslant r-k}
\bar L_{\omega_\alpha}(x,\tilde\Omega_{r,k,p})\big(\p_u(Q^\Phi_1)^{\alpha_1}\cdots(Q^\Phi_p)^{\alpha_p}u\big)\ne0.
\]
\end{remark}

\begin{remark}\label{RemarkOnMeta-regularModulesOfDiffFunctions}
In general, a differential function~$L$ (of order $r>0$) admits an infinite number of meta-regular modules of various dimensions. 
Indeed, suppose that at a point $z^0=(x^0,u_{\ssl r\ssr}^0)$ of the $r$th order jet space~$\mathrm J^r$, 
the differential function~$L$ has a noncharacteristic direction $(c_1,\dots,c_n)$,~i.e., 
\[
C:=\sum_{|\alpha|=r}L_{u_\alpha}(x^0,u_{\ssl r\ssr}^0)c_1^{\alpha_1}\dots c_n^{\alpha_n}\ne0,
\]
cf.\ \cite[Definition 2.75]{Olver1993}.
We change the independent variables~$x$, $\tilde x_i=X^i(x)$ with $X^1_i(x^0)=c_i$, in a neighborhood of~$x^0$ and set $\tilde u=u$, 
which induces a local coordinate change in~$\mathrm J^r$. 
Denote by $\tilde z^0$ the new coordinates of the point $z^0$.
In the new coordinates, 
the derivative $L_{\tilde u_{r\delta_1}}$ coincides with~$C$ at the point $\tilde z^0$ 
and hence it does not vanish in a neighborhood of this point. 
Then any module $M=\langle Q^0,\dots,Q^p\rangle$ with $0<p<n$, where $Q^0=\p_{\tilde u}$ and 
$\langle Q^1,\dots,Q^p\rangle$ is a submodule of $\langle\p_{\tilde x_2},\dots,\p_{\tilde x_n}\rangle$, 
is a meta-regular module for~$L$. 
Pushing forward elements of~$M$ by the inverse change of coordinates, 
we construct a meta-regular module for~$L$ in the initial variables~$(x,u)$. 
Due to the functional freedom in choosing~$X^1$ and, additionally if $p<n-1$, 
the functional freedom in choosing~$p$-dimensional submodules in $\langle\p_{\tilde x_2},\dots,\p_{\tilde x_n}\rangle$, 
the differential function~$L$ admits an infinite number of meta-regular modules of any dimension~$p$ with $0<p<n$.
\end{remark}

\section{Singular modules of vector fields for differential equations}\label{SectionOnSingularModulesOfVectorFieldsForDiffEqs}

An involutive module~$Q$ satisfying the rank condition is called \emph{(strongly) singular} for a differential equation~$\mathcal L$
if it is singular for the differential function~$L[u]$ constituting the left hand side of the canonical representation $L[u]=0$
of the equation~$\mathcal L$. 
Strongly regular modules are defined similarly.
Usually we will omit the attribute ``strongly''.

As left hand sides of differential equations are defined up to multipliers which are nonvanishing differential functions,
the conditions from Definition \ref{DefinitionOfSingularModuleForDiffFunction} can be weakened
when considering differential equations.

\begin{definition}\label{DefinitionOfWeaklySingularModuleForDiffEq}
A $p$-dimensional ($0<p<n$) involutive module~$Q$ which satisfies the rank condition
is called \emph{weakly singular}
for the differential equation~$\mathcal L$: $L[u]=0$ of essential order~$r>0$
if there exists a differential function $\tilde L=\tilde L[u]$ of an order less than~$r$
and a nonvanishing differential function $\lambda=\lambda[u]$ of an order not greater than~$r$
such that $L|_{\mathcal Q_{\ssl r\ssr}}=(\lambda\tilde L)|_{\mathcal Q_{\ssl r\ssr}}$.
Otherwise we call~$Q$ a \emph{weakly regular} module for the differential equation~$\mathcal L$.
If the minimal order of differential functions whose restrictions on $\mathcal Q_{\ssl r\ssr}$ coincide,
up to nonvanishing functional multipliers, with $L|_{\mathcal Q_{\ssl r\ssr}}$ is equal to $k\in\{-\infty,0,1,\dots,r\}$,
then the module~$Q$ is said to be \emph{weakly singular of co-order $k$} for the differential equation~$\mathcal L$.
\end{definition}

In particular, as in the case of strong regularity, weakly regular modules for the differential equation~$\mathcal L$
are defined to have weak singularity co-order $r=\ord L$.
An involutive module~$Q$ is considered to be weakly ultra-singular for the differential equation $\mathcal L$: $L[u]=0$
if it is strongly ultra-singular for~$\mathcal L$.
We write $\wsco_{\mathcal L}Q$ for the weak singularity co-order of the module~$Q$ for the equation~$\mathcal L$.

Strong singularity implies weak singularity and consequently weak regularity implies strong regularity.
The weak singularity co-order is always less or equal and may be strictly less than the strong singularity co-order.
Thus, strongly regular modules may be singular in the weak~sense.

\begin{example}\label{ExampleOnWeakSingularity}
To illustrate the relation between strong and weak singularity co-order, consider the equation $x_2u_{111}+x_1u_{222}=e^{u_{33}}(u_{3}+u)$.
It possesses the two-dimensional singular module $\langle\p_1,\p_2\rangle$
whose strong and weak singularity co-orders equal~2 and~1, respectively.
The same module $\langle\p_1,\p_2\rangle$ is strongly regular and
is of weak singularity co-order~1 for the equation $x_2u_{11}+x_1u_{22}=e^{u_{33}}(u_3+u)$.
\end{example}

Let $\hat L$ be the differential function associated with~$L$ on the manifold~$\mathcal Q_{\ssl r\ssr}$
by excluding $\mathcal Q$-principal derivatives.
Without loss of generality, we can replace
the differential functions~$\lambda$ and $\tilde L$
in Definition~\ref{DefinitionOfWeaklySingularModuleForDiffEq}
by their counterparts associated with them on~$\mathcal Q_{\ssl r\ssr}$.
Then the equality $L=\lambda\tilde L$ on the manifold~$\mathcal Q_{\ssl r\ssr}$
is equivalent to the representation $\hat L=\lambda\tilde L$.
We suppose that $\hat L$ is of maximal rank in a derivative~$u_\alpha$ of the highest order~$k$ appearing in this differential function, i.e.,
$\smash{\hat L_{u_{(\check 0,\hat\alpha)}}\ne0}$ for some $\hat\alpha$ with $|\hat\alpha|=k$ on the solution manifold of the equation~$\hat L=0$
(see the notation in the proof of Theorem~\ref{TheoremOnDiffFunctionsWithMetaSingularModule}).
Then we can set $\tilde L=\hat L$ and $\lambda=1$,
which means that the weak singularity co-order of~$Q$ for the equation~$\mathcal L$: $L=0$ equals the order~$k$ of $\hat L$
and, consequently, the strong singularity co-order of~$Q$ for this equation.
Therefore, in this case there is an entirely algorithmic procedure for testing whether an
involutive module is weakly singular for a partial differential equation.

\begin{proposition}\label{PropositionUselessnessOfWeaklySingularFamiliesOfVectorFields}
A differential equation~$\mathcal L$: $L[u]=0$ with one dependent and $n$~independent variables possesses
a co-order~$k$ weakly singular $p$-dimensional module of vector fields ($0<p\leqslant n$) if and only if
this module is co-order~$k$ strongly singular for~$\mathcal L$
(possibly in a representation differing from $L[u]=0$ by multiplication by a nonvanishing differential function~of~$u$).
\end{proposition}

\begin{proof}
In the notation of Definition~\ref{DefinitionOfWeaklySingularModuleForDiffEq},
the equation $L=0$ is equivalent to the equation $\tilde L=0$, and $\sco_{\tilde L}Q=k$.
\end{proof}

\begin{example}\label{ExampleOnWeakAndStrongSingularityForEquivEqs}
Multiplying the equations from Example~\ref{ExampleOnWeakSingularity}
by the nonvanishing differential function~$e^{-u_{33}}$,
we obtain the equations $e^{-u_{33}}(x_2u_{111}+x_1u_{222})=u_{3}+u$
and $e^{-u_{33}}(x_2u_{11}+x_1u_{22})=u_3+u$, respectively.
Then the module $\langle\p_1,\p_2\rangle$ is a singular module
of both weak and strong singularity co-orders~1 for these equations.
\end{example}

The following assertion is proven in the same way as Proposition~\ref{PropositionOnModulesWithSingularSubmodules}.

\begin{proposition}\label{PropositionOnModulesWithWeaklySingularSubmodules}
Suppose that $\mathcal L$ is a differential equation with respect to the unknown function~$u$
of $n$~independent variables~$x$
and let $Q$ be an involutive module of vector fields defined in the space of $(x,u)$,
which is of dimension less than~$n$ and satisfies the rank condition.
Then \mbox{$\wsco_{\mathcal L}Q\leqslant\wsco_{\mathcal L}\check Q$}
for any involutive submodule~$\check Q$ of~$Q$.
In particular, the module~$Q$ is weakly singular for~$\mathcal L$
if it contains a submodule weakly singular for~$\mathcal L$.
\end{proposition}

Modules that are meta-singular (resp.\ meta-regular) for differential equations in the strong sense 
are defined in a way similar to strongly singular (resp.\ regular) modules.

\begin{definition}\label{DefinitionOfStronglylyMetaSingularModuleForDiffEq}
A $(p{+}1)$-dimensional ($0<p<n$) module~$M$ is called \emph{meta-singular} (resp.\ \emph{meta-regular}) 
for the differential equation~$\mathcal L$ \emph{in the strong sense}
if it is meta-singular (resp.\ meta-regular) 
for the differential function~$L[u]$ constituting the left hand side of the canonical representation $L[u]=0$
of the equation~$\mathcal L$.
\end{definition}

The notion of meta-singular (resp.\ meta-regular) modules for differential equations in the weak sense 
is defined by analogy with the case of differential functions.

\begin{definition}\label{DefinitionOfWeaklyMetaSingularModuleForDiffEq}
A $(p{+}1)$-dimensional ($0<p<n$) module~$M$ is called \emph{meta-singular in the weak sense} for the differential equation~$\mathcal L$
if any $p$-dimensional involutive submodule of~$M$ that satisfies the rank condition is weakly singular for~$\mathcal L$ and
the module~$M$ contains a family $\mathfrak M=\{Q^\Phi\}$ of such submodules
parameterized by an arbitrary function $\Phi=\Phi(x,u)$ of all independent and dependent variables.
The \emph{weak singularity co-order} of the meta-singular module~$M$ for~$\mathcal L$, 
which is denoted by $\wsco_{\mathcal L}M$,
coincides with the maximum of the weak singularity co-orders of its involutive submodules satisfying the rank condition.
\end{definition}

\begin{definition}\label{DefinitionOfWeaklyMetaRegularModuleForDiffEq}
A $(p{+}1)$-dimensional ($0<p<n$) module~$M$ is called \emph{meta-regular in the weak sense} for the differential equation~$\mathcal L$ 
if the module~$M$ contains a family $\mathfrak M=\{Q^\Phi\}$ of $p$-dimen\-sional involutive submodules that are weakly regular for~$\mathcal L$
and parameterized by an arbitrary function $\Phi=\Phi(x,u)$ of all independent and dependent variables, 
and thus $\sco_LM=\ord L$.
\end{definition}

\begin{theorem}\label{TheoremOnWeaklyMetaSingularModuleForDiffEq}
An $r$th order ($r>0$) differential equation~$\mathcal L$: $L[u]=0$ of maximal rank with one dependent and $n$~independent variables possesses
a co-order $k$ weakly meta-singular $(p{+}1)$-dimensional module~$M$ of vector fields if and only if
the differential function~$L$ can be represented, up to point transformations of variables, in the form
\begin{equation}\label{EqWithWeaklyMetaSingularModule}
L=\bar\lambda[u]\bar L(x,\Omega_{r,k,p}),
\end{equation}
where $\bar\lambda$ is a nonvanishing differential function of order not greater than~$r$,
$\bar L$ is a smooth function of~$x$ and $\Omega_{r,k,p}=\big(\omega_\alpha,\,|\alpha|\leqslant r,\,\alpha_{p+1}+\dots+\alpha_n\leqslant k\big)$
that essentially depends on some $\omega_\alpha$ with $\alpha_{p+1}+\dots+\alpha_n=k$.
Here $\omega_\alpha=u_\alpha$ or, only for the case $p=1$,
$\omega_\alpha=D_2^{\alpha_2}\cdots D_n^{\alpha_n}(D_1+uD_2+\xi^3D_3+\dots+\xi^nD_n)^{\alpha_1} u$
for some fixed smooth functions $\xi^i=\xi^i(x,u)$, $i=3,\dots,n$.
The value of~$k$ should be minimal among all possible representations
of the form~\eqref{EqWithWeaklyMetaSingularModule} for the differential function~$L$.
Then the module~$M$ is co-order $k$ strongly meta-singular
for the equation $\bar L(x,\Omega_{r,k,p})=0$, which is equivalent to~$\mathcal L$.
\end{theorem}

\begin{proof}
We will freely use the notations and definitions from the proof of Theorem~\ref{TheoremOnDiffFunctionsWithMetaSingularModule}.

To begin with, assume that the differential equation~$\mathcal L$: $L[u]=0$ is of maximal rank and admits
a co-order~$k$ weakly meta-singular $(p{+}1)$-dimensional module of vector fields.
Up to point transformations and changes of module basis, it suffices to consider only
the family~$\mathfrak M=\{Q^\Phi=\langle Q_s-(\Phi_s/\Phi_u)Q^0,\, s=1,\dots,p\rangle\}$
of $p$-dimensional involutive submodules of~$M$ which are parameterized by an arbitrary function $\Phi=\Phi(x,u)$.
Here $Q^0$ is reduced to the shift operator~$\p_u$
and the vector fields $Q^s$ take either the form $Q_s=\p_s$ or,
if $p=1$ and the module~$M$ is not involutive, the form $Q_1=\p_1+u\p_2+\xi^3\p_3+\dots+\xi^n\p_n$.

Given any point $\smash{z^0=(x^0,u_{\ssl r\ssr}^0)\in\mathcal L\subset \mathrm J^r}$
we choose the modules from $\mathfrak M=\{Q^\Phi\}$ for which $\smash{z^0\in\mathcal Q^\Phi_{\ssl r\ssr}}$.
It then follows that the values of the derivatives of~$\Phi$ with respect to only~$x_1$, \dots,~$x_n$
at the point $(x^0,u^0)$, which contain differentiation with respect to some~$x_s$,
can be expressed via $\smash{u_{\ssl r\ssr}^0}$ and values of derivatives of~$\Phi$ in $(x^0,u^0)$,
containing differentiation with respect to~$u$.
These latter values are then unconstrained.

We next derive from~$L$ the differential function $\hat L^\Phi$ by excluding,
in view of the system $\smash{\mathcal Q^\Phi_{\ssl r\ssr}}$,
$\mathcal Q^\Phi$-principal derivatives of~$u$, which involve differentiations with respect to~$x_s$.
The fact that for any function~$\Phi$ with $\Phi_u\ne0$ the module $Q^\Phi$ is at most of co-order~$k$ weak singularity for~$\mathcal L$ leads to
the condition $\hat L^\Phi_{u_{(\check0,\hat\alpha)}}(z_0)=0$, $|\hat\alpha|=k+1,\dots,r$.
As in the proof of Theorem~\ref{TheoremOnDiffFunctionsWithMetaSingularModule} we now iteratively employ this condition,
starting out with the greatest value~$r$ of~$|\hat\alpha|$ and re-writing the derivatives in the new coordinates
$(x_i,\omega_\alpha,|\alpha|\leqslant r)$ of~$\mathrm J^r$ and in terms of~$L$.
Thereby we obtain
\[
L_{\omega_\alpha}(z^0)=0, \quad|\hat\alpha|>k,\quad |\alpha|\leqslant r,
\]
which is satisfied for any $z^0\in\mathcal L$.
Applying the Hadamard lemma to each of these equations and using Lemma~\ref{LemmaOnNonvanishingMultiplier},
we obtain~\eqref{EqWithWeaklyMetaSingularModule}.

Conversely, suppose that the $r$th order differential function~$L$ is of the form~\eqref{EqWithWeaklyMetaSingularModule}
(after a point transformation).
For an arbitrary smooth function~$\Phi=\Phi(x,u)$ with $\Phi_u\ne0$, we consider
the involutive module~$Q^\Phi$ generated by either the vector fields $Q^\Phi_s=\p_s-(\Phi_s/\Phi_u)\p_u$
in the general case or the vector field $Q^\Phi_1=\p_1+u\p_2+\xi^3\p_3+\dots+\xi^n\p_n-(\Phi_s/\Phi_u)\p_u$
in the special case with $p=1$.
Using~$\bar L$ and~$Q^\Phi$, we construct the differential function $\tilde L^\Phi=\bar L(x,\tilde\Omega_{r,k,p})$, where
\[
\tilde\Omega_{r,k,p}=\big(\omega_\alpha=D_{p+1}^{\alpha_{p+1}}\cdots D_n^{\alpha_n}(Q^\Phi_1)^{\alpha_1}\cdots(Q^\Phi_p)^{\alpha_p}u,\,
|\hat\alpha|\leqslant k,\,|\alpha|\leqslant r\big).
\]
By construction, the order of the differential function~$\tilde L^\Phi$ is not greater than~$k$
for all allowed values of the parameter-function~$\Phi$.
Moreover, $\ord \tilde L^\Phi=k$ for almost all values of this parameter-function except those
which satisfy a system of differential equations.
As
\[
L\big|_{\mathcal Q^\Phi_{\ssl r\ssr}}
=(\bar\lambda\bar L)\big|_{\mathcal Q^\Phi_{\ssl r\ssr}}
=(\bar\lambda\tilde L^\Phi)\big|_{\mathcal Q^\Phi_{\ssl r\ssr}},
\]
where $\Phi$ runs through the set of smooth functions of $(x,u)$ with nonvanishing derivative with respect to~$u$,
the family $\mathfrak M=\{Q^\Phi\}$ is a co-order~$k$ weakly singular family of $p$-dimensional involutive modules
for the differential equation~$\mathcal L$ in the new coordinates.
We return to the old coordinates.
In view of Proposition~\ref{PropositionOnReducedFormOfFamiliesOfSingularModules}, if $p\geqslant2$,
or Proposition~\ref{PropositionOnReducedFormOfFamiliesOfOneDimSingularModules} if $p=1$,
the module~$M$ of vector fields which is spanned by the family $\mathfrak M$ is
a co-order~$k$ meta-singular $(p{+}1)$-dimensional module for the differential function~$L/\bar\lambda$.
Therefore, this module is
a co-order~$k$ weakly meta-singular $(p{+}1)$-dimensional module for the differential equation~$\mathcal L$.
\end{proof}

\begin{remark}\label{RemarkOnWeaklyAndStronglyMetaSingularModuleForDiffEq}
Theorem~\ref{TheoremOnWeaklyMetaSingularModuleForDiffEq} implies
that, up to equivalence of equations by nonvanishing multipliers that are differential functions,
we do not need to specify which kind of meta-singularity, weak or strong,
of differential equations is meant.
\end{remark}

\section{Singularity of reduction modules and order\\ of reduced equations}\label{SectionOnReductionModulesAndParametricFamiliesOfSolutions}

\begin{definition}\label{DefinitionOfSingularRedOps}
An involutive module~$Q$ is called a \emph{singular reduction module} of a differential equation~$\mathcal L$
if $Q$ is both a reduction module of~$\mathcal L$ and a weakly singular module of~$\mathcal L$.
\end{definition}

Recall that differential equations are equivalent
if they differ by a multiplier that is a nonvanishing differential function.
The essential order of a differential equation~$\mathcal L$
is the minimal order of equations in the equivalence class of~$\mathcal L$.

The following assertions are obtained as direct extensions of the corresponding results
for the case $n=2$~\cite{Kunzinger&Popovych2010a}.

\begin{theorem}\label{TheoremOnReductionModulesAndReducedEqs}
Let $Q$ be a $p$-dimensional reduction module ($0<p\leqslant n$) of an equation~$\mathcal L$. 
Then the weak singularity co-order of~$Q$ for~$\mathcal L$
coincides with the essential order of a corresponding reduced differential equation.
\end{theorem}

\begin{proof}
Suppose that $p<n$ and $k:=\wsco_{\mathcal L}Q>0$.
Then any point transformation preserves the value of~$\wsco_{\mathcal L}Q$.
Locally, by a point transformation we can achieve the situation
that in the new variables the module~$Q$ has a basis $(Q_s=\p_s,\,s=1,\dots,p)$.
(Again we use the same notation for old and new variables).
An ansatz~$\mathcal A$ constructed with~$Q$ in the new variables is $u=\varphi(\omega)$,
where $\varphi=\varphi(\omega)$ is the new unknown function, $\omega=(\omega_1,\dots,\omega_{n-p})$,
and $\omega_1=x_{p+1}$, \dots, $\omega_{n-p}=x_n$ are the invariant independent variables.
The system~$\mathcal Q_{\ssl r\ssr}$ consists of the equations $u_\alpha=0$,
where the multi-index $\alpha=(\alpha_1,\dots,\alpha_n)$ satisfies the conditions
$\alpha_1+\dots+\alpha_p>0$, $|\alpha|\leqslant r:=\ord L$.

As $Q\in\mathcal R^p(\mathcal L)$, there are differential functions
$\smash{\check\lambda=\check\lambda[\varphi]}$ and $\smash{\check L=\check L[\varphi]}$
of an order not greater than~$r$ such that $L|_{\mathcal A}=\check\lambda\check L$
(cf.\ \cite{Zhdanov&Tsyfra&Popovych1999}).
The function~$\check\lambda$ does not vanish and may depend on the variables~$x_s$ as parameters.
Modulo equivalence generated by nonvanishing multipliers we may assume that
the function~$\check L$ is of minimal order~$\check r$.
Then the reduced equation~$\check{\mathcal L}$: $\check L=0$ has essential order~$\check r$.

Now since $\wsco_{\mathcal L}Q=k$ it follows that
there exists a strictly $k$th order differential function $\tilde L=\tilde L[u]$
and a nonvanishing differential function $\tilde\lambda=\tilde\lambda[u]$ of an order not greater than~$r$,
which depend at most on~$x$ and derivatives of~$u$ with respect to~$x_{p+1}$, \dots, $x_n$
such that $(L-\tilde\lambda\tilde L)|_{\mathcal Q_{\ssl r\ssr}}=0$.

Suppose first that $\check r$ would be less than~$k$. In this case
we can use $\tilde\lambda_{\rm new}=\check\lambda|_{u\rightsquigarrow\varphi}$
and $\tilde L_{\rm new}=\check L|_{u\rightsquigarrow\varphi}$
in the definition of weak singularity
to arrive at the contradiction $\wsco_{\mathcal L}Q\leqslant\ord\tilde L_{\rm new}=\check r<k$.
Thus $\check r\geqslant k$.
(Here ``$u\rightsquigarrow\varphi$'' means
that the derivatives $\p_{p+1}^{\alpha_{p+1}}\dots\p_n^{\alpha_n}u$
should be substituted instead of the derivatives $\p_{\omega_{p+1}}^{\alpha_{p+1}}\dots\p_{\omega_n}^{\alpha_n}\varphi$.)

On the other hand, if $\check r>k$ we conclude that $\check\lambda\check L=(\tilde\lambda\tilde L)|_{\mathcal A}$,
where the variables~$x_s$ play the role of parameters.
Fixing a value $x_s^0$ of $x_s^{}$ for each~$s$, we obtain the representation
\[
\check L=\frac{\tilde\lambda|_{\mathcal A}}{\check\lambda}\biggr|_{x_s^{}=x_s^0}
\ \tilde L\,\bigg|_{\mathcal A,\; x_s^{}=x_s^0}.
\]
However, as $\ord\tilde L|_{\mathcal A,\; x_s^{}=x_s^0}\leqslant k< \check r$,
this representation contradicts the condition
that $\check r$ is the essential order of the reduced equation~$\check{\mathcal L}$.

We conclude that $\check r=k>0$.
Then the value of $\check r$ does not depend on the choice of ansatz
among those constructed with~$Q$.
The inverse change of variables preserves the claimed property.

Consider now the case $k:=\wsco_{\mathcal L}Q\leqslant0$.
In general, the values $k=-\infty$ and $k=0$
are not invariant with respect to point transformations
since they may be mapped to each other by such transformations.
The essential order of the related reduced equations may depend on the choice of ansatzes
in the set of ansatzes constructed with~$Q$.
This is why straightening out the basis vector fields~$Q_s$ and choosing an ansatz
should be consistent with the structure of the restriction of~$L$ to the manifold~$\mathcal Q_{\ssl r\ssr}$
in order to result in the proper reduction procedure.
In view of the condition $k\leqslant0$ there exists
a smooth function $\tilde L=\tilde L(x,u)$ with $\ord\tilde L=k$
and a nonvanishing differential function $\tilde\lambda=\tilde\lambda[u]$ of an order not greater than~$r$
such that $L=\tilde\lambda\tilde L$ on $\mathcal Q_{\ssl r\ssr}$.
Lemma~\ref{LemmaOnReformulationOfCondInvCriterion} implies
that $Q$ is a reduction module of the equation $\tilde{\mathcal L}$: $\tilde L=0$,
and thus $Q_s\tilde L=\Lambda^s\tilde L$ for some smooth function~$\Lambda^s=\Lambda^s(x,u)$.
Due to Lemma~\ref{LemmaOnNonvanishingMultiplier},
this gives the representation $\tilde L=\check\lambda\check L$,
where $\check\lambda$ is a nonvanishing smooth function of $(x,u)$,
and the smooth function $\check L=\check L(x,u)$ satisfies the system $Q_s\check L=0$.
Therefore, the function~$\check L$ is an invariant of the module~$Q$
and hence it can be presented as a function of the basis invariants~$I^0$,~\dots, $I^{n-p}$,
$\check L=\zeta(I^0,\dots,I^{n-p})$ for a smooth function~$\zeta$ of its arguments.
Note that $k\leqslant\ord\check L\leqslant 0$
since we have $L=\tilde\lambda\check\lambda\check L$ on $\mathcal Q_{\ssl r\ssr}$,
and $\tilde\lambda\check\lambda$ is a nonvanishing differential function~of~$u$.

If $k=0$, then $\ord\check L=0$
and thus the function~$\zeta$ essentially depends at least on one basis invariant.
Up to permutation of the basis invariants, we can assume that this invariant is~$I^0$.
Then the ansatz~\eqref{EqGenAnsatzForm} reduces~$\mathcal L$ to the equation $\zeta(\varphi,\omega)=0$,
and the order of this reduced equation as equation for~$\varphi$ equals zero.

Finally suppose that $k=-\infty$.
In view of the rank condition for~$Q$,
there exists an element of the basis~$(I^0,\dots,I^{n-p})$ of $Q$-invariants
whose derivative with respect to~$u$ does not vanish.
Up to permutation of the basis,
we can assume that this invariant is~$I^0$.
We now change coordinates in the space of~$(x,u)$:
$y_0=I^0(x,u)$, $y_\sigma=I^\sigma(x,u)$, $z_s=J^s(x,u)$,
where, as in Section~\ref{SectionOnDefOfRedModules},
each $J^s=J^s(x,u)$ is a solution of the system $Q_{s'}J^s=\delta_{ss'}$,
and $\delta_{ss'}$ is the Kronecker delta.
Denote $y=(y_0,\dots,y_{n-p})$ and $z=(z_1,\dots,z_p)$.
The equality $\tilde L_u=0$ under the representation \smash{$\tilde L=\check\lambda(y,z)\zeta(y)$}
can be written as
\[
\zeta_{y_0}+\frac{I^\sigma_u}{I^0_u}\zeta_{y_\sigma}=-\frac{\lambda_u}{\lambda I^0_u}\zeta,
\]
where the coefficients of~$\zeta_{y_\sigma}$ and~$\zeta$ are smooth functions
that should be expressed in the coordinates $(y,z)$.
We fix a value $z=z^0$ in these coefficients and modify
the coordinates~$y$ in order to straighten out the vector field
$\p_{y_0}+(I^\sigma_u/I^0_u)\big|_{z=z^0}\p_{y_\sigma}$ to $\p_{y_0}$.
We also modify consistently the functional basis of $Q$-invariants.
Lemma~\ref{LemmaOnNonvanishingMultiplier} implies the representation $\zeta=\breve\lambda\breve\zeta$,
where $\breve\lambda$ is a nonvanishing smooth function of~$y$
and the smooth function \smash{$\breve\zeta$} of~$y$ does not depend on~$y_0$.
Then the ansatz~\eqref{EqGenAnsatzForm} reduces the equation~$\mathcal L$
to the equation $\breve\zeta(\omega)=0$,
which is of order $-\infty$ with respect to~$\varphi$.
\end{proof}

The proof of the case $p=n$ is also given in Section~\ref{SectionOnReductionModulesOfMaxD}
before Proposition~\ref{PropositionOnReductionModulesOfMaxD}.
The distinguishing feature of this case in comparison the case $\wsco_{\mathcal L}Q\leqslant0$ with $p<n$
is that there is no essential freedom in choosing invariant dependent variable
since there is only one functionally independent $Q$-invariant if $Q$ is an $n$-dimensional module.

Non-ultra-singular $p$-dimensional reduction modules ($0<p\leqslant n$)
of weak singularity co-order~$-\infty$ lead to inconsistent reduced equations.

The properties of ultra-singular modules of vector fields as reduction modules are obvious.
The selection of ansatzes is not needed for them.

\begin{proposition}\label{PropositionOnUltraSingularRedOpsAndFamiliesOfInvSolutions}
1) Any $p$-dimensional module~$Q$ of vector fields ($0<p\leqslant n$)
which is ultra-singular for a differential equation~$\mathcal L$ is a reduction module of this equation.
An ansatz constructed with~$Q$ reduces~$\mathcal L$ to the identity.
Therefore, the family of $Q$-invariant solutions of~$\mathcal L$ is parameterized by an arbitrary smooth function
of $n-p$ $Q$-invariant variables.

2) If $Q$ is a $p$-dimensional involutive module of vector fields and
the family of $Q$-invariant solutions of~$\mathcal L$ is parameterized by an arbitrary smooth function
of $n-p$ $Q$-invariant variables, then $Q$ is an ultra-singular module for~$\mathcal L$.
\end{proposition}

Given a partial differential equation~$\mathcal L$, consider a reduction module~$Q$ for~$\mathcal L$ of dimension~$n-1$, i.e., $p=n-1$.
Such a module with $\wsco_{\mathcal L}Q>0$ reduces the equation~$\mathcal L$ to an ordinary differential equation
whose essential order, in view of Theorem~\ref{TheoremOnReductionModulesAndReducedEqs},
is equal to $\wsco_{\mathcal L}Q$.
If $\wsco_{\mathcal L}Q=0$, then the reduced equation is algebraic.
This allows us to relate nonnegative $\wsco_{\mathcal L}Q$ with the maximal number of parameters
in families of $Q$-invariant solutions of~$\mathcal L$.
An exception occurs for $\wsco_{\mathcal L}Q=-\infty$, when
$Q$ reduces~$\mathcal L$ to an identity if $Q$ is an ultra-singular module for~$L$,
while otherwise the corresponding reduced equation is inconsistent.

\begin{proposition}\label{PropositionOnSingularCodimOneRedModulesAndFamiliesOfInvSolutions}
Given an $(n{-}1)$-dimensional involutive module~$Q$ of vector fields and a differential function~$L[u]$,
suppose that the equation $\hat L=0$ with a lowest-order differential function~$\hat L=\hat L[u]$,
associated with~$L[u]$ on the manifold~$\mathcal Q_{\ssl r\ssr}$ (with $r=\ord L$) up to a nonvanishing multiplier,
can be solved with respect to the highest-order derivative of~$u$ appearing in this equation.
Then any two of the following properties imply the third one.

1) $Q$ is a reduction module of the equation~$\mathcal L$: $L=0$.

2) The weak singularity co-order of~$Q$ for~$\mathcal L$ equals~$l$, where $0\leqslant l\leqslant r$.

3) The equation~$\mathcal L$ possesses an $l$-parameter family of $Q$-invariant solutions,
and any $Q$-invariant solution of~$\mathcal L$ belongs to this family.
\end{proposition}

\begin{proof}
If $Q$ is an $(n{-}1)$-dimensional reduction module of an equation~$\mathcal L$ with $\wsco_{\mathcal L}Q\geqslant0$,
then the weak singularity co-order of~$Q$ for~$\mathcal L$ equals
the maximal number $N_{\mathcal L,Q}$ of essential continuous parameters in families of $Q$-invariant solutions of~$\mathcal L$.
Indeed, the weak singularity co-order of~$Q$ for~$\mathcal L$ coincides with
the essential order~$\check r$ of the reduced ordinary differential equation~$\check{\mathcal L}$ associated with~$Q$,
$\check r=\wsco_{\mathcal L} Q$.
The maximal number of essential continuous parameters in solutions of~$\check{\mathcal L}$ also equals~$\check r$.
The substitution of these solutions into the corresponding ansatz leads to parametric families of $Q$-invariant solutions of~$\mathcal L$,
and all $Q$-invariant solutions of~$\mathcal L$ are obtained in this way.
Therefore, $N_{\mathcal L,Q}=\check r$.

In view of the proposition's supposition for~$\hat L$, the equation~$\check{\mathcal L}$ can be written in normal form and hence
has an $\check r$-parameter general solution that contains all solutions of~$\check{\mathcal L}$.
Substituting it into the corresponding ansatz, this solution gives an $\check r$-parameter family of $Q$-invariant solutions of~$\mathcal L$.
There are no other $Q$-invariant solutions of~$\mathcal L$.
Therefore, conditions~2 and~3 are equivalent if condition~1 holds.

For any $(n{-}1)$-dimensional involutive module~$Q$ of vector fields we have $N_{\mathcal L,Q}\leqslant k$, where by~$k$ we denote $\wsco_{\mathcal L} Q$.
Let us prove that $Q$ is a reduction module of~$\mathcal L$ if $N_{\mathcal L,Q}=k$.

\looseness=-1
Point transformations of the variables do not change the claimed property for $k>0$,
and thus we can use the variables and notations
from the same case of the proof of Theorem~\ref{TheoremOnReductionModulesAndReducedEqs}.
For the ansatz~$\mathcal A$: $u=\varphi(\omega)$ constructed with~$Q$ in the new variables,
consider the differential function $\breve L[\varphi]=\hat L|_{\mathcal A}$.
It depends on the variables~$x_s$ as parameters, and $\ord\breve L=k$.
Due to our assumption on~$\hat L$,
we can resolve the equation $\breve L=0$ with respect to the highest-order derivative~$\varphi^{(k)}$:
$\varphi^{(k)}=R[\varphi]$, where $\ord R<k$.
If $R_{x_s}\ne0$ for some~$s$, splitting with respect to~$x_s$
in the resolved equation results in
an ordinary differential equation $\tilde R[\varphi]=0$ of an order lower than~$k$.
Any $Q$-invariant solution of~$\mathcal L$ has the form $u=\varphi(\omega)$,
where the function~$\varphi$ satisfies, in particular, the equation $\tilde R[\varphi]=0$.
This contradicts the condition that $N_{\mathcal L,Q}=k$.
Therefore, $R_{x_s}=0$, i.e.,
the equation $\varphi^{(k)}=R[\varphi]$ is a reduced equation
obtained from the equation~$\mathcal L$ by the substitution
of the ansatz~$\mathcal A$: $u=\varphi(\omega)$ constructed with the module~$Q$,
i.e., $Q$ is a reduction module of~$\mathcal L$.

The condition $k=0$ means that $\hat L$ is a function of~$(x,u)$.
In view of our assumption on~$\hat L$,
the equation~$\hat{\mathcal L}$: $\hat L=0$ can be solved with respect to~$u$.
This gives a function $u=f(x)$, which is the only possible solution
of the joint system~$\hat{\mathcal L}$ and~$\mathcal Q$.
At the same time, by condition~3 the equation~$\mathcal L$,
or equivalently~$\hat{\mathcal L}$, possesses a $Q$-invariant solutions.
Therefore, the function $u=f(x)$ is the only solution of the equation $\hat L=0$,
and this solution is $Q$-invariant.
This is why the equation $\hat L=0$ is equivalent,
up to a nonvanishing multiplier depending at most on $(x,u)$,
to an equation involving only $Q$-invariants,
which means by definition that $Q$ is a reduction module of~$\hat{\mathcal L}$
and, by Lemma~\ref{LemmaOnReformulationOfCondInvCriterion}, of~$\mathcal L$.
\end{proof}

The precise relation between the reducibility of the equation~$\mathcal L$ by the module~$Q$
and the formal compatibility of the joint system of~$\mathcal L$ and~$\mathcal Q$
is in fact not as simple as one might believe.
See \cite[footnote~1]{Kunzinger&Popovych2009}
for the case of one-dimensional reduction modules.
The choice of the system whose formal compatibility
should be checked depends on~$\wsco_{\mathcal L}Q$
and on what definition of formal compatibility is used.
Consider, e.g., the definition given in \cite{Seiler2010}.
Let $\mathcal E_\rho$ be a system of $\kappa$~differential equations $\smash{E^1[u]=0}$, \dots, $\smash{E^\kappa[u]=0}$,
which involves derivatives of~$u$ up to order~$\rho$.
The system~$\mathcal E_\rho$ is interpreted as a system of algebraic equations in the jet space~$\mathrm J^\rho$
and defines a manifold in~$\mathrm J^\rho$, which is also denoted by~$\mathcal E_\rho$.
The $\varsigma$th order prolongation $\mathcal E_{\rho+\varsigma}$
of the system~$\mathcal E_\rho$, $\varsigma\in\mathbb N$, is the system in~$\mathrm J^{\rho+\varsigma}$
consisting of the equations $\smash{D^{\alpha}E^\nu[u]=0}$, $\nu=1,\dots,\kappa$, $|\alpha|\leqslant\varsigma$.
The projection of the corresponding manifold on $\mathrm J^{\rho+\varsigma-\varsigma'}$,
where $\varsigma'\in\mathbb N$ and $\varsigma'\leqslant\varsigma$,
is denoted by $\smash{\mathcal E_{\rho+\varsigma-\varsigma'}^{\ssl \varsigma'\ssr}}$.
The system~$\mathcal E_\rho$ is called \emph{formally integrable} (or \emph{formally compatible}) if
$\smash{\mathcal E_{\rho+\varsigma}^{(1)}=\mathcal E_{\rho+\varsigma}}$ for any $\varsigma\in\mathbb N$.

There are two obstacles for the harmonization
of the above definition of formal compatibility
and the definition of reduction module.
The first obstacle is created by (in general) different orders
of the equation~$\mathcal L$: $L[u]=0$ and the characteristic system~$\mathcal Q$: $Q_s[u]=0$.
This is overcome by attaching differential consequences the characteristic system~$\mathcal Q$
to the joint system of~$\mathcal L$ and~$\mathcal Q$ before testing its compatibility.
The second obstacle is the lowering of the order of $L[u]$ on the manifold~$\mathcal Q_{\ssl r\ssr}$
if $Q$ is a singular module for the equation $L[u]=0$.
Hence instead of the equation~$\mathcal L$ we should use
the equation $\tilde{\mathcal L}$: $\tilde L[u]=0$,
where $\tilde L$ is a differential function such that $L=\tilde\lambda\tilde L$ on $\mathcal Q_{\ssl r\ssr}$
for some nonvanishing differential function~$\lambda$ of~$u$, and $\ord\tilde L=\wsco_{\mathcal L}Q$.
The module~$Q$ reduces the differential equation~$\mathcal L$ if and only if
\begin{itemize}\itemsep=0ex
\item
the system $\tilde L[u]=0$, $D^\alpha Q_s[u]=0$, $|\alpha|<\wsco_{\mathcal L}Q$, in the case $\wsco_{\mathcal L}Q>0$ or
\item
the system $\tilde L[u]=0$, $D_\iota\tilde L[u]=0$, $Q_s[u]=0$ in the case $\wsco_{\mathcal L}Q=0$
\end{itemize}
is formally compatible.
Here the index~$\iota$ runs through a subset $\{\iota_1,\dots,\iota_{n-p}\}$
of the range $\{1,\dots,n\}$ such that the module $\langle Q_1,\dots,Q_p,\p_{\iota_1},\dots,\p_{\iota_{n-p}}\rangle$
satisfies the rank condition.

The case $\wsco_{\mathcal L}Q=-\infty$, where $\tilde L$ does not identically vanish,
is irrelevant for the framework of formal compatibility since
then the equation $\tilde{\mathcal L}$ itself is inconsistent as an equation with respect to~$u$.
If~$Q$ is an ultra-singular module for~$\mathcal L$, i.e.\ $\tilde L\equiv0$,
then the ``formal compatibility'' of the joint system of~$\mathcal L$ and~$\mathcal Q$
is trivially equivalent to the involutivity of the module~$Q$.

\section{Reduction modules of maximal dimension}
\label{SectionOnReductionModulesOfMaxD}

As remarked in Section~\ref{SectionOnSingularModulesOfVectorFieldsForDiffFunctions},
the case when the dimension of modules of vector fields coincides with the number of independent variables,
$p=n$, is special for the singularity of modules for differential functions.
This case is also special for reduction of differential equations.
In contrast to reduction modules of lower dimensions,
$n$-dimensional reduction modules of any differential equation~$\mathcal L$ with $n$ independent variables
reduce this equation to mere algebraic equations instead of differential equations
with a smaller number of independent variables.
Moreover, this is the only case when both regular and singular reduction modules can be studied
in a uniform way.

Let $Q$ be an involutive module of dimension $p=n$ that satisfies the rank condition.
Then we can choose the basis of $Q$ consisting of the vector fields $Q_s=\p_s+\eta^s(x,u)\p_u$.
As the module~$Q$ is involutive, the basis elements~$Q_s$ should commute,
and hence the coefficients $\eta^s=\eta^s(x,u)$ satisfy the system
\begin{equation}\label{EqInvCondForModulesReducingToAlgEqs}
\eta^s_{s'}+\eta^{s'}\eta^s_u=\eta^{s'}_s+\eta^s\eta^{s'}_u.
\end{equation}
The Frobenius theorem implies that the system $Q_s\Phi=\Phi_s+\eta^s\Phi_u=0$
with respect to the function $\Phi=\Phi(x,u)$ has a single functionally independent solution,
which is not constant.
In~other words, the coefficients~$\eta^s$ can be represented in the form $\eta^s=-\Phi_s/\Phi_u$
for some smooth function $\Phi=\Phi(x,u)$ with $\Phi_u\ne0$.
An implicit ansatz constructed for~$u$ with the module~$Q$ is $\Phi(x,u)=\varphi$,
where $\varphi$ is the new unknown function.
It is nullary since the module dimension~$p$ of~$Q$ equals the number~$n$ of independent variables~$x$.

Suppose that $Q$ is a reduction module of an $r$th order differential equation~$\mathcal L$: $L[u]=0$.
All the derivatives of $u$ with respect to~$x$ from order~1 up to order~$r$
are expressed, on the manifold~$\mathcal Q_{\ssl r\ssr}$, via the variables $x$ and~$u$:
\[
u_\alpha=h^\alpha(x,u):=(\p_1+\eta^1\p_u)^{\alpha_1}\cdots(\p_n+\eta^n\p_u)^{\alpha_n}u,\quad
1\leqslant|\alpha|\leqslant r.
\]
As the vector fields~$Q_s$ commute,
the above representation for derivatives of~$u$ is well defined since it does not depend on
the ordering of differential operators in the expressions on the right hand side.
Using this representation for excluding the appropriate derivatives $u_\alpha$ from~$L$,
we get a differential function $\hat L=\hat L[u]$ of nonpositive order, i.e., a function of~$(x,u)$.
Varying~$\eta^s$, it can also be interpreted as a differential function
with the independent variables~$x$ and~$u$ and the dependent variables~$\eta^s$.
Therefore, in view of Lemma~\ref{LemmaOnReformulationOfCondInvCriterion}
the fact that $Q$ is a reduction module of~$\mathcal L$ is equivalent to the condition
$Q_s\hat L=\Lambda^s\hat L$ for some smooth functions $\Lambda^s=\Lambda^s(x,u)$,
which holds on the corresponding open subset of the zero-order jet space.
This condition jointly with equations~\eqref{EqInvCondForModulesReducingToAlgEqs} gives
the complete system of determining equations for the coefficients~$\eta^s$.
Under the representation $\eta^s=-\Phi_s/\Phi_u$, it is interpreted as a condition for the function~$\Phi$,
where $L^\Phi=\hat L|_{-\Phi_s/\Phi_u\rightsquigarrow\eta^s}$ is considered as a differential function
with the independent variables~$x$ and~$u$ and the dependent variable~$\Phi$.

In order to handle the condition $Q_s\hat L=\Lambda^s\hat L$
if the equation $\hat L=0$ is of maximal rank,
we can solve the equation $\hat L=0$ with respect to a variable (either one of the~$x$'s or~$u$) and
exclude this variable from the equations $Q_s\hat L=0$ using the obtained expression.
As the functions~$\hat L$ and~$\eta^s$ (resp.\ $\Phi$) depend on the same arguments,
the above procedure does not result in a usual differential equation.

\looseness=-1
This is why we use another, more convenient, way.
By Lemma~\ref{LemmaOnNonvanishingMultiplier},
the condition $Q_s\hat L=\Lambda^s\hat L$ implies
that there exist a nonvanishing smooth function $\lambda=\lambda(x,u)$
and a smooth function $\check L=\check L(x,u)$
such that $\hat L=\lambda\check L$ and $Q_s\check L=0$.
The last condition means that $\check L=\zeta(\Phi)$
for some smooth function~$\zeta$ of one argument.
Then the equation $\hat L=\lambda\check L$ in terms of~$\Phi$
takes the form $L^\Phi=\tilde\lambda\zeta(\Phi)$.

In fact, this equation is precisely the condition of reducing the equation~$\mathcal L$
by the ansatz $\Phi(x,u)=\varphi$ associated with the module~$Q$.
Indeed, if the function~$u=u(x)$ is implicitly defined by this ansatz,
then the derivatives of~$u$ (of nonzero orders) are found from differential consequences
of the equations $u_s=-\Phi_s/\Phi_u$.
Therefore, the substitution of the ansatz to~$\mathcal L$
leads to the equation $\hat L|_{\Phi(x,u)=\varphi}=0$,
which is equivalent, in view of the reduction condition,
to the algebraic equation $\zeta(\varphi)=0$ with respect to the constant~$\varphi$.

It is obvious that the essential order of the reduced equation $\zeta(\varphi)=0$
coincides with $\wsco_{\mathcal L}Q\in\{-\infty,0\}$.
The module~$Q$ is ultra-singular for~$\mathcal L$ if and only if
the function~$\zeta$ identically vanishes.
If $\wsco_{\mathcal L}Q=-\infty$ and the module~$Q$ is not ultra-singular,
then the corresponding reduced algebraic equation is not consistent with respect to~$\varphi$.

Summing up the above considerations results in the following assertion:

\begin{proposition}\label{PropositionOnReductionModulesOfMaxD}
An $n$-dimensional module~$Q$ is a reduction module of a differential equation~$\mathcal L$: $L[u]=0$
with $n$ independent variables~$x$
if and only if this module is spanned by the vector fields $\p_s-(\Phi_s/\Phi_u)\p_u$, $s=1,\dots,n$,
where the function $\Phi=\Phi(x,u)$ satisfies the equation $L^\Phi=\tilde\lambda\zeta(\Phi)$
for some smooth function~$\zeta$ of one argument
and some nonvanishing function $\tilde\lambda$ of $(x,u)$,
and the differential function~$L^\Phi=L^\Phi[\Phi]$ is obtained from $L[u]$ by the exclusion of
the derivatives of~$u$ using differential consequences of the equations $u_s=-\Phi_s/\Phi_u$.
The ansatz $\Phi(x,u)=\varphi$ reduces the equation~$\mathcal L$
to the algebraic equation $\zeta(\varphi)=0$ with respect to the constant~$\varphi$,
and the essential order of the reduced equation
coincides with $\wsco_{\mathcal L}Q\in\{-\infty,0\}$.
\end{proposition}

\begin{theorem}\label{TheoremOnOneParFamiliesOfSolutionsAndnDReductionModules}
Up to the equivalence of solution families, for any differential equation~$\mathcal L$
with respect to a single unknown function of $n$ independent variables
there exists a one-to-one correspondence between one-parameter families of its solutions
and its $n$-dimensional ultra-singular reduction modules.
Namely, each module of the above kind corresponds to
the family of solutions which are invariant with respect to this module.
The problems of the construction of all one-parameter solution families of the equation~$\mathcal L$
and the exhaustive description of its $n$-dimensional ultra-singular reduction modules
are completely equivalent.
\end{theorem}

\begin{proof}
Let $Q$ be an $n$-dimensional ultra-singular reduction module of the equation~$\mathcal L$.
It follows from Proposition~\ref{PropositionOnReductionModulesOfMaxD} that
the ansatz $\Phi(x,u)=\varphi$ constructed with the module~$Q$ reduces the equation~$\mathcal L$ to the identity.
In other words, for each value of the constant~$\varphi$ this ansatz implicitly defines a solution of~$\mathcal L$.

\looseness=-1
Conversely,
suppose that $\mathcal F=\{u=f(x,\varkappa)\}$ is a family of solutions of~$\mathcal L$
parameterized by the single parameter~$\varkappa$.
As this parameter is essential and, therefore, the derivative $f_\varkappa$ is nonzero,
we can express $\varkappa$ from the equality $u=f(x,\varkappa)$.
As a result, we obtain that $\varkappa=\Phi(x,u)$ for some function $\Phi=\Phi(x,u)$ with $\Phi_u\ne0$.
Consider the module $Q=\langle Q_1,\dots,Q_n\rangle$,
where $Q_s=\p_s+\eta^s\p_u$ and the coefficients $\eta^s=\eta^s(x,u)$ are defined by $\eta^s=-\Phi_s/\Phi_u$.
It is an $n$-dimensional involutive module and
$Q_s[u]=0$ for any element of~$\mathcal F$.
The ansatz $u=f(x,\varphi)$, where $\varphi$ is the new unknown (nullary) function,
is associated with~$Q$ and reduces~$\mathcal L$ to the identity.
This means that $Q$ is an ultra-singular reduction module for~$\mathcal L$.

One-parameter families $\mathcal F=\{u=f(x,\varkappa)\}$ and $\tilde{\mathcal F}=\{u=\tilde f(x,\tilde\varkappa)\}$
are defined to be equivalent if they consist of the same functions and differ only by parameterizations, i.e.,
if there exists a function $\zeta=\zeta(\varkappa)$ such that
$\zeta_\varkappa\ne0$ and $\tilde f(x,\zeta(\varkappa))=f(x,\varkappa)$.
This is true if and only if the functions~$\Phi=\Phi(x,u)$ and $\tilde\Phi=\tilde\Phi(x,u)$
associated with the families $\mathcal F$ and $\tilde{\mathcal F}$, respectively,
are functionally dependent.
More precisely, $\tilde\Phi=\zeta(\Phi)$.
As $\Phi_u\tilde\Phi_u\ne0$, the functional dependence of $\Phi$ and~$\tilde\Phi$
is equivalent to the equality $\Phi_s/\Phi_u=\tilde\Phi_s/\tilde\Phi_u$.
Therefore, equivalent one-parameter families of solutions correspond to
the same ultra-singular reduction module~$Q$ of~$\mathcal L$
and, conversely, any two one-parameter families of $Q$-invariant solutions are equivalent.
\end{proof}

\begin{corollary}\label{CorollaryOnReductionOfDetEqsForModulesReducingToAlgEqs}
The system of the determining equations for coefficients of $n$-dimensional ultra-singular reduction modules of~$\mathcal L$,
which consists of equations~\eqref{EqInvCondForModulesReducingToAlgEqs} and the equation $\hat L=0$,
is reduced by the composition of the nonlocal substitution
$\eta^s=-\Phi_s/\Phi_u$, where $\Phi$ is a function of $(x,u)$, and the hodograph transformation
\begin{gather*}
\mbox{the new independent variables:}\qquad \tilde x_i=x_i, \quad \varkappa=\Phi,
\\
\lefteqn{\mbox{the new dependent variable:}}\phantom{\mbox{the new independent variables:}\qquad }\tilde u=u
\end{gather*}
to the initial equation~$\mathcal L$ in the function $\tilde u=\tilde u(\tilde t,\tilde x,\varkappa)$
with $\varkappa$ playing the role of a parameter.
\end{corollary}

Note that the reduction of differential equations to algebraic ones using nonclassical symmetries
was considered in~\cite{Grundland&Tafel1995}.
An assertion similar to Theorem~\ref{TheoremOnOneParFamiliesOfSolutionsAndnDReductionModules}
was obtained therein.
Using the notion of singular reduction modules makes
Theorem~\ref{TheoremOnOneParFamiliesOfSolutionsAndnDReductionModules} more precise.

\section{Motivating example: evolution equations}\label{SectionOnExampleOfEvolEqs}

We investigate $n$-dimensional singular reduction modules
of $(1{+}n)$-dimensional evolution equations of the general form
\begin{equation}\label{EqGenEvolEq}
u_t=H(t,x,u_{\ssl r,x\ssr})
\end{equation}
for the unknown function $u$ depending on the variables~$t=x_0$ and $x=(x_1,\dots,x_n)$.
(For convenience, in this section we set the number of independent variables to $n+1$ instead of~$n$
and additionally single out the variable~$x_0$.)
Here $u_{\ssl r,x\ssr}$ denotes the set of all derivatives
of the functions~$u$ with respect to the space variables~$x$
of order not greater than~$r$, including $u$ as derivative of order zero. Also,
$u_t=\p u/\p t$ and
$r=\max\{|\alpha|\mid H_{u_\alpha}\ne0\}\geqslant2$,
i.e., we assume the order of the equations under consideration to be not less than two.
We fix an arbitrary equation~$\mathcal L$ of the form~\eqref{EqGenEvolEq}.
An $n$-dimensional reduction module of~$\mathcal L$ reduces~$\mathcal L$ to an ordinary differential equation.
In general, such a module~$Q$ is spanned by $n$ vector fields of the form
\[
Q_s=\tau^s(t,x,u)\p_t+\xi^{si}(t,x,u)\p_i+\eta^s(t,x,u)\p_u,
\]
with $\rank(\tau^s,\xi^{si})=n$.
(We use the same convention for the ranges of indices as in the whole paper, i.e.,
the indices $i$ and $j$ run from 1 to $n$ and the index $s$ runs from 1 to $p$
noting that here $p$ coincides with~$n$.)

A systematic study of singular reduction modules of evolution equations
was initiated in~\cite{Fushchych&Shtelen&Serov&Popovych1992}.
It was shown that the system of determining equations
for singular one-dimensional reduction modules
of the (1+1)-dimensional linear heat equation
consists of a single (1+2)-dimensional nonlinear evolution equation,
which is reduced by a nonlocal transformation to the initial equation with
an additional implicit independent variable playing the role of a parameter.
Later this no-go assertion was extended to more general (1+1)-dimensional
linear evolutions equations~\cite{Fushchych&Popowych1994-1,Popovych1995},
general (1+1)-dimensional evolution equations~\cite{Zhdanov&Lahno1998} and
multi-dimensional evolution equations~\cite{Popovych1998}, cf.\ the introduction.
At the same time, the calculation of all one-dimensional reduction modules
of the $(1{+}n)$-dimensional linear hear equation with $n>2$ in explicit form
in~\cite{Popovych&Korneva1998} showed
that the existence of the above no-go cases is not explained
just by vanishing $t$-components of vector fields from related reduction modules.
Following \cite{Kunzinger&Popovych2010a}, we set up these results
within the framework of singular reduction modules.

In contrast to the specific case $n=1$ considered in~\cite{Kunzinger&Popovych2010a},
for general values of~$n$ we can completely describe
singular $n$-dimensional modules of the equation~$\mathcal L$ only after more precisely knowing the form of~$H$.
At the same time, by direct generalization of the case $n=1$ we easily find a family of such modules
for any equation of the form~\eqref{EqGenEvolEq}.

\begin{proposition}\label{PropositionOnSingularModulesOfEvolEqs}
If the coefficient of~$\p_t$ in any vector field from an $n$-dimensional involutive module~$Q$ is zero, 
then $Q$ is a singular module for the differential function $L=u_t-H(t,x,u_{\ssl r,x\ssr})$.
The co-order of singularity of~$Q$ equals one, $\sco_{\mathcal L}Q=1$.
\end{proposition}

\begin{proof}
Suppose that $\tau^s=0$. Then $\rank(\xi^{si})=n$ and up to changing basis of~$Q$ we can set
\begin{equation}\label{EqReducedForOfSingularModulesForEvolEqs}
Q_s=\p_s+\eta^s(t,x,u)\p_u.
\end{equation}
All the derivatives of $u$ with respect to~$x$ from order~1 up to order~$r$ can be expressed,
on the manifold~$\mathcal Q_{\ssl r\ssr}$, via $t$, $x$ and~$u$:
\[
u_\alpha=h^\alpha(t,x,u):=(\p_1+\eta^1\p_u)^{\alpha_1}\cdots(\p_n+\eta^n\p_u)^{\alpha_n}u.
\]
Here and in what follows $\alpha=(\alpha_0,\dots,\alpha_n)$,
$1\leqslant|\alpha|\leqslant r$ and $\alpha_0=0$.
In view of the module~$Q$ being involutive, the vector fields~$Q_s$ commute.
(Moreover, we have $\eta^s=-\Phi_s/\Phi_u$ for some smooth function $\Phi=\Phi(x,u)$ with $\Phi_u\ne0$.)
Therefore, the above representation for~$u_\alpha$ is well defined since it does not depend on
the ordering of differential operators in the expression on right hand side.
Using this representation for excluding the appropriate derivatives $u_\alpha$ from~$L$,
we get the differential function $\tilde L=u_t-\tilde H(t,x,u)$, where
$\tilde H=\tilde H(t,x,u)$ is the function obtained from~$H$ by the substitution of~$h^\alpha$ for~$u_\alpha$.
The order of~$\tilde L$ equals~1.
Hence the module~$Q$ is singular for the differential function~$L$, and its singularity co-order equals~1.
\end{proof}

\begin{corollary}\label{CorollaryOSingularVectorFieldsOfEvolEqs}
The module~$M=\langle\p_1,\dots,\p_n,\p_u\rangle$ is a co-order one strictly meta-singular $(n{+}1)$-dimensional module
for any $(1{+}n)$-dimensional evolution equation.
\end{corollary}

\begin{proof}
An $n$-dimensional involutive submodule~$Q$ of~$M$ satisfies the rank condition if and only if
it is spanned by the vector fields $\p_s-(\Phi_s/\Phi_u)\p_u$ for some smooth function $\Phi=\Phi(x,u)$.
It follows from Proposition~\ref{PropositionOnSingularModulesOfEvolEqs} that
the submodule~$Q$ is a co-order one strictly singular module for any evolution equation of the form~\eqref{EqGenEvolEq}.
Varying the parameter-function~$\Phi$, we obtain a family of such submodules parameterized
by an arbitrary function of all independent and dependent variables.
\end{proof}

The vector fields~$\p_s$ and~$\p_u$ generating the meta-singular module~$M$ commute
and the differential function~$L$ contains only first-order differentiation with respect to~$t$
(namely, in the form of the derivative~$u_t$).
This perfectly agrees with Theorem~\ref{TheoremOnDiffFunctionsWithMetaSingularModule}.

We denote by $\mathcal R^n_0(\mathcal L)$ the set of $n$-dimensional reduction modules of~$\mathcal L$ that are contained in~$M$.
Consider a module~$Q$ from $\mathcal R^n_0(\mathcal L)$ and choose a basis for it consisting of vector fields~$Q^s$
of the form~\eqref{EqReducedForOfSingularModulesForEvolEqs}.
The system ${\rm DE}_0(\mathcal L)$ of determining equations for the coefficients~$\eta^s$ is naturally partitioned into two subsystems.
The first subsystem
\begin{equation}\label{EqInvCondForSingularReductionModulesOfEvolEqs}
\eta^s_{s'}+\eta^{s'}\eta^s_u=\eta^{s'}_s+\eta^s\eta^{s'}_u
\end{equation}
follows from the condition that the module~$Q$ is involutive and hence
the basis elements~$Q^s$ commute. 
The second subsystem
\begin{equation}\label{EqsDetForSingularReductionModulesOfEvolEqs}
\eta^s_t+\tilde H\eta^s_u=\tilde H_s+\eta^s\tilde H_u
\end{equation}
is a consequence of the conditional invariance criterion
applied to the equation~$\mathcal L$ and the module~$Q$.
The function $\tilde H=\tilde H(t,x,u)$ is defined in the proof of Proposition~\ref{PropositionOnSingularModulesOfEvolEqs}.
It coincides with~$H$ on the manifold $\mathcal Q_{\ssl r\ssr}$.
The total number of equations in the joint system ${\rm DE}_0(\mathcal L)$
of~\eqref{EqInvCondForSingularReductionModulesOfEvolEqs} and~\eqref{EqsDetForSingularReductionModulesOfEvolEqs}
is $n(n+1)/2$ and thereby greater than the number of the unknown functions~$\eta^s$ if $n>1$.
Hence the system ${\rm DE}_0(\mathcal L)$ looks strongly overdetermined in the multidimensional case.
In fact, the equations of the subsystems agree well with each other.
The subsystem~\eqref{EqInvCondForSingularReductionModulesOfEvolEqs} implies
the representation $\eta^s=-\Phi_s/\Phi_u$ of~$\eta^s$ via a single arbitrary function~$\Phi=\Phi(t,x,u)$.
Substituting this representation into the subsystem~\eqref{EqsDetForSingularReductionModulesOfEvolEqs},
we obtain a system of $n$ partial differential equations in the single function~$\Phi$,
that is equivalent to a single equation in~$\Phi$.
Indeed, we have
\begin{gather*}
\tilde H_s+\eta^s\tilde H_u-\eta^s_n-\eta^s_u\tilde H
=H^\Phi_s-\frac{\Phi_s}{\Phi_u}H^\Phi_u+\left(\frac{\Phi_s}{\Phi_u}\right)_n+\left(\frac{\Phi_s}{\Phi_u}\right)_uH^\Phi
\\\qquad
=\frac1{\Phi_u}\left(\p_s-\frac{\Phi_s}{\Phi_u}\p_u\right)\left(\Phi_n +\Phi_uH^\Phi\right)=0,
\end{gather*}
i.e., this system is equivalent to the equation $\Phi_t+\Phi_uH^\Phi=\chi(t,\Phi)$.
Here $\chi$ is an arbitrary smooth function of its arguments
and $H^\Phi$ coincides with~$\tilde H$ under the substitution $\eta^s=-\Phi_s/\Phi_u$.
The expression for~$H^\Phi$ involves derivatives of~$\Phi$ up to order $r+1$.
The function~$\Phi$ associated with a fixed module~$Q$ is defined up to the transformation $\tilde\Phi=\theta(t,\Phi)$ with $\theta_\Phi\ne0$.
Since $\tilde\eta^s=-\tilde\Phi_s/\tilde\Phi_u=-\Phi_s/\Phi_u=\eta^s$, the functions~$H^\Phi$ and~$\smash{H^{\tilde\Phi}}$ coincide.
At the same time, if we choose~$\theta$ satisfying the equation $\theta_t+\chi\theta_\Phi=0$, 
then $\tilde\Phi_t+\tilde\Phi_uH^{\tilde\Phi}=0$.
Therefore, up to the equivalence on the set of functions parameterizing singular modules we can assume that
the function~$\Phi$ is a solution of the equation $\Phi_t+\Phi_uH^\Phi=0$.

Collecting the above arguments, we derive the following result.

\begin{proposition}\label{PropositionOnReductionOfDetEqs0ForRedModulesOfEvolEqs}
The system ${\rm DE}_0(\mathcal L)$ of
the determining equations~\eqref{EqInvCondForSingularReductionModulesOfEvolEqs}
and~\eqref{EqsDetForSingularReductionModulesOfEvolEqs}
 is reduced by the composition of the nonlocal substitution
$\eta^s=-\Phi_s/\Phi_u$, where $\Phi$ is a function of $(t,x,u)$, and the hodograph transformation
\begin{gather*}
\mbox{the new independent variables:}\qquad\tilde t=t, \quad \tilde x_i=x_i, \quad \varkappa=\Phi,
\\
\lefteqn{\mbox{the new dependent variable:}}\phantom{\mbox{the new independent variables:}\qquad }\tilde u=u
\end{gather*}
to the initial equation~$\mathcal L$ in the function $\tilde u=\tilde u(\tilde t,\tilde x,\varkappa)$
with $\varkappa$ playing the role of a parameter.
\end{proposition}

\begin{proof}
The subsystem~\eqref{EqInvCondForSingularReductionModulesOfEvolEqs} implies
the existence of a function $\Phi=\Phi(t,x,u)$ with $\Phi_u\ne0$ such that $\eta^s=-\Phi_s/\Phi_u$.
The entire system ${\rm DE}_0(\mathcal L)$ is reduced,
up to equivalence on the set traversed by the parameter-function~$\Phi$
from the above representation for~$\eta^s$,
to the single equation $\Phi_t+\Phi_uH^\Phi=0$.
The last equation is mapped by the hodograph transformation
to the initial equation~$\mathcal L$ for the function $\tilde u=\tilde u(\tilde t,\tilde x,\varkappa)$.
This directly follows from the definition of the function~$H^\Phi$
and the rule for calculating derivatives under the hodograph transformation,
$\tilde u_{\tilde x_i}=-\Phi_i/\Phi_u$, etc.
\end{proof}

\begin{proposition}\label{PropositionOnSingularAndUltraSingularModulesOfEvolEqs}
For any equation~$\mathcal L$ of the form~\eqref{EqGenEvolEq},
the following statements are equivalent:

1) The span~$Q=\langle Q_s=\p_s+\eta^s(t,x,u)\p_u, s=1,\dots,n\rangle$ is a reduction module of~$\mathcal L$.

2) The module~$\tilde Q=\langle\tilde Q_0,Q_1,\ldots,Q_n\rangle$, where $\tilde Q_0=\p_t+\tilde H\p_u$ and
the function $\tilde H=\tilde H(t,x,u)$ coincides with $H$ on the manifold~$\mathcal Q_{\ssl r\ssr}$,
is involutive.

3) There exists a vector field $\hat Q_0=\p_t+\eta^0(t,x,u)\p_u$ such that the module
$\hat Q=\langle\hat Q_0,Q_1,\ldots,Q_n\rangle$ is ultra-singular for~$\mathcal L$.

Moreover, under these equivalent conditions the coefficient~$\eta^0$ is uniquely determined as
$\eta^0=\tilde H$, i.e., we necessarily have $\hat Q_0=\tilde Q_0$.
\end{proposition}

\begin{proof}
Both the first and second statement are equivalent to the fact that the coefficients~$\eta^s$ satisfy
the system ${\rm DE}_0(\mathcal L)$ consisting of the equations~\eqref{EqInvCondForSingularReductionModulesOfEvolEqs}
and~\eqref{EqsDetForSingularReductionModulesOfEvolEqs}.
The module~$\hat Q$ is ultra-singular for~$\mathcal L$ if and only if $\eta^0=\tilde H$,
i.e. the module~$\hat Q$ coincides with $\tilde Q$, and this module is involutive.
\end{proof}

\begin{theorem}\label{TheoremUnitedOnSetsOfSolutionsAndReductionOperatorsWithTau0OfSystemsOfEqlsp}
Up to the equivalence of solution families, for any equation of the form~\eqref{EqGenEvolEq}
there exists a one-to-one correspondence between one-parameter families of its solutions
and its $n$-dimensional reduction modules formed by vector fields with vanishing $t$-components.
Namely, each module of the above kind corresponds to
the family of solutions which are invariant with respect to this module.
The problems of the construction of all one-parameter solution families of equation~\eqref{EqGenEvolEq}
and the exhaustive description of its $n$-dimensional reduction modules formed by vector fields with vanishing $t$-components
are completely equivalent.
\end{theorem}

\begin{proof}
Let $\mathcal L$ be an equation from class~\eqref{EqGenEvolEq} and $Q=\langle Q_1,\dots,Q_n\rangle\in \mathcal R^n_0(\mathcal L)$,
i.e., $Q_s=\p_s+\eta^s\p_u$, where the coefficients $\eta^s=\eta^s(t,x,u)$ satisfy the system ${\rm DE}_0(\mathcal L)$.
An ansatz constructed with~$Q$ has the form $u=f(t,x,\varphi(\omega))$,
where $f=f(t,x,\varphi)$ is a given function, $f_\varphi\ne0$,
$\varphi=\varphi(\omega)$ is the new unknown function and $\omega=t$ is the invariant independent variable.
This ansatz reduces $\mathcal L$ to a first-order ordinary differential equation $\mathcal L'$ in~$\varphi$,
solvable with respect to~$\varphi_\omega$.
The general solution of the reduced equation~$\mathcal L'$ can be represented in the form
$\varphi=\varphi(\omega,\varkappa)$, where $\varphi_\varkappa\ne0$ and $\varkappa$ is an arbitrary constant.
The form of the general solution is defined up to a transformation $\tilde\varkappa=\zeta(\varkappa)$
of the parameter~$\varkappa$.
Substituting this solution into the ansatz results in
the one-parameter family~$\mathcal F$ of solutions $u=\tilde f(t,x,\varkappa)$ of~$\mathcal L$
with $\tilde f=f(t,x,\varphi(t,\varkappa))$,
and any $Q$-invariant solution of~$\mathcal L$ belongs to this family.
Expressing the parameter~$\varkappa$ from the equality $u=\tilde f(t,x,\varkappa)$,
we obtain that $\varkappa=\Phi(t,x,u)$, where $\Phi_u\ne0$.
Then $\eta^s=u_s=-\Phi_s/\Phi_u$ for any $u\in\mathcal F$, i.e., for any admissible value of $(t,x,\varkappa)$.
This implies that $\eta^s=-\Phi_s/\Phi_u$ for any admissible value of $(t,x,u)$.

The proof of the converse assertion is similar to that of
Theorem~\ref{TheoremOnOneParFamiliesOfSolutionsAndnDReductionModules}.
Consider a one-parameter family $\mathcal F=\{u=f(t,x,\varkappa)\}$ of solutions of~$\mathcal L$.
The derivative $f_\varkappa$ is nonzero since the parameter $\varkappa$ is essential.
We express $\varkappa$ from the equality $u=f(t,x,\varkappa)$: $\varkappa=\Phi(t,x,u)$
for some function $\Phi=\Phi(t,x,u)$ with $\Phi_u\ne0$.
The span $Q=\langle Q_1,\dots,Q_n\rangle$,
where $Q_s=\p_s+\eta^s\p_u$ and the coefficients $\eta^s=\eta^s(t,x,u)$ is defined by $\eta^s=-\Phi_s/\Phi_u$,
is an $n$-dimensional involutive module.
For any $u\in\mathcal F$ we have $Q_s[u]=0$.
The ansatz $u=f(t,x,\varphi(\omega))$, where $\omega=t$, associated with~$Q$,
reduces~$\mathcal L$ to the equation $\varphi_\omega=0$.
Therefore, $Q\in\mathcal R^n_0(\mathcal L)$.

One-parameter families
$\mathcal F=\{u=f(t,x,\varkappa)\}$ and $\tilde{\mathcal F}=\{u=\tilde f(t,x,\tilde\varkappa)\}$
of solutions of~$\mathcal L$ are equivalent if and only if
the associated functions~$\Phi=\Phi(t,x,u)$ and $\tilde\Phi=\tilde\Phi(t,x,u)$
satisfy the condition $\Phi_s/\Phi_u=\tilde\Phi_s/\tilde\Phi_u$.
Therefore, equivalent one-parameter families of solutions correspond to
the same module~$Q$ from $\mathcal R^n_0(\mathcal L)$
and, conversely, any two one-parameter families of $Q$-invariant solutions are equivalent.
\end{proof}

\begin{remark}
\looseness=-1
The triviality of the above ansatz and the reduced equation results from
the above special representation for the solutions of the determining equation.
In this approach difficulties in the construction of ansatzes
and the integration of the reduced equations are replaced
by difficulties in obtaining the representation for the components of basis elements of reduction modules.
\end{remark}

\begin{remark}
In fact, Theorem~\ref{TheoremUnitedOnSetsOfSolutionsAndReductionOperatorsWithTau0OfSystemsOfEqlsp}
is a consequence of Theorem~\ref{TheoremOnOneParFamiliesOfSolutionsAndnDReductionModules}
and Proposition~\ref{PropositionOnSingularAndUltraSingularModulesOfEvolEqs}.
This observation provides a uniform background for no-go results on reduction modules.
It suffices to note that any $Q$-invariant solution of~$\mathcal L$ is $\hat Q$-invariant,
where the modules~$Q$ and~$\hat Q$ are defined in
Proposition~\ref{PropositionOnReductionOfDetEqs0ForRedModulesOfEvolEqs}.
We have given a direct proof as it leads to a deeper understanding of reduction of evolution equations
to first-order ordinary differential equations.
\end{remark}

\section{Reduction modules of codimension one\\ and singularity co-order one}\label{SectionOnReductionModulesOfSingularityCoOrder1}

Taking the previous example on evolution equations as a model case, we now proceed to studying
$(n{-}1)$-dimensional singularity-co-order-one reduction modules of general partial differential equations in
one dependent and $n$~independent variables.
In the course of considering single modules of this kind it is only possible to derive a result similar to
Proposition~\ref{PropositionOnSingularAndUltraSingularModulesOfEvolEqs}.

\begin{proposition}\label{PropositionOnWSCO1ReductionModulesAndUltraSingularModules}
Let $\mathcal L$: $L[u]=0$ be a partial differential equation in one dependent and $n$~independent variables,
let $Q$ be an $(n{-}1)$-dimensional involutive module of vector fields defined in the space of $(x,u)$
and satisfying the rank condition, as well as $\wsco_{\mathcal L}Q=1$.
We also suppose that a first-order differential function~$\hat L$
associated with~$L$ on the manifold~$\mathcal Q_{\ssl r\ssr}$ up to nonvanishing multiplier,
where $r=\ord L$,
is of maximal rank with respect to the unique first-order derivative of~$u$ appearing in~$\hat L$.
Then $Q$ is a reduction module of~$\mathcal L$ if and only if
there exists a (unique) $n$-dimensional module~$\hat Q$
which is ultra-singular for~$\mathcal L$ and contains~$Q$.
\end{proposition}

\begin{proof}
Under these assumptions, there exists a vector field~$Q_0$ such that
the equation~$\hat L=0$ is equivalent to the equation $Q_0[u]=0$, where $Q_0[u]$ is the characteristic of~$Q_0$.
Consider the module~$\hat Q$ spanned by~$Q$ and~$Q_0$.
It is $n$-dimensional and satisfies the rank condition.
Lemma~\ref{LemmaOnReformulationOfCondInvCriterion} implies that
$Q$ is a reduction module of~$\mathcal L$ if and only if
it is a reduction module of the equation $\hat L=0$.
The last condition is equivalent to the involutivity of the module~$\hat Q$,
and then the equation~$\mathcal L$ is an identity
on the well-defined manifold~$\smash{\hat{\mathcal Q}_{\ssl r\ssr}}$,
i.e., the module~$\hat Q$ is ultra-singular for~$\mathcal L$.
The conditions that the $n$-dimensional module~$\hat Q$ is ultra-singular for~$\mathcal L$ and contains~$Q$
uniquely determine~$\hat Q$.
\end{proof}

\begin{remark}
The assumption that the differential function~$\hat L$
is of maximal rank with respect to the unique first-order derivative of~$u$ appearing in~$\hat L$
can be replaced by the weaker supposition that the equation $\hat L=0$ is equivalent to the equation $\check L=0$,
where the differential function~$\check L$ satisfies the above condition of maximal rank.
\end{remark}

In order to have a richer theory, we shall consider families of
$(n{-}1)$-dimensional reduction modules of singularity co-order one that are parameterized by arbitrary functions.
Thus let $\mathcal L$: $L[u]=0$ be a partial differential equation of essential order $r>1$.
We assume that the function~$L$ admits an $n$-dimensional co-order one meta-singular module~$M$ of vector fields,
cf.\ Remark~\ref{RemarkOnWeaklyAndStronglyMetaSingularModuleForDiffEq}.
As the special case $n=2$, when meta-singular modules may be non-involutive,
was considered in~\cite{Kunzinger&Popovych2010a},
in what follows we additionally assume that the module~$M$ is involutive. 

\begin{remark}
For first-order partial differential equations, involutive modules with singularity co-order one are regular. 
The results of this section are true also for such equations 
via replacing the attributes ``co-order one singular'' and ``co-order one meta-singular'' 
by ``regular'' and ``meta-regular'', respectively.
\end{remark}

Up to a change of variables we may suppose that the module~$M$ contains
a singularity-co-order-one family $\mathfrak M=\{Q^\Phi\}$ of $(n{-}1)$-dimensional modules in reduced form
parameterized by an arbitrary smooth function~$\Phi$ of $(x,u)$, i.e.,
$Q^\Phi=\langle\p_s+\eta^s\p_u\rangle$, where $\eta^s=-\Phi_s/\Phi_u$
and the index~$s$ runs from~1 to~$p=n-1$.
We use the representation for~$\eta^s$ from the very beginning.

By Theorem~\ref{TheoremOnDiffFunctionsWithMetaSingularModule},
the differential function~$L$ can be written in the form $L=\check L(x,\Omega_{r,1,n-1})$, where
$\Omega_{r,1,n-1}=\big(u_\alpha,\,\alpha_n\leqslant 1,\,|\alpha|\leqslant r\big)$,
and $\check L$ essentially depends on some $u_\alpha$ with $\alpha_n=1$.
In this case the restriction of~$L$ to $\smash{\mathcal Q^\Phi_{\ssl r\ssr}}$ coincides with
the restriction, to the same manifold $\smash{\mathcal Q^\Phi_{\ssl r\ssr}}$,
of the function $\tilde L^\Phi=\check L(x,\tilde\Omega_{r,1,n-1})$, where
\[
\tilde\Omega_{r,1,n-1}=\big(D_n^{\alpha_n}(Q^\Phi_1)^{\alpha_1}\cdots(Q^\Phi_{n-1})^{\alpha_{n-1}}u,
\alpha_n\leqslant 1,\,|\alpha|\leqslant r\big).
\]
Consequently, the form of $\tilde L^\Phi$ is determined by the form of~$L$ and a chosen value of the parameter-function~$\Phi$.
Depending on the value of~$\Phi$, the differential function~$\tilde L^\Phi$
may either identically vanish or be a nonvanishing differential function of order not greater than 1.
Thus either the module $Q^\Phi$ is ultra-singular for~$L$ or
$\sco_L Q^\Phi\leqslant0$ with $\tilde L^\Phi$ not identically vanishing
or $\sco_L Q^\Phi=1$.
We analyze each of these cases separately.
In addition, we suppose that the equation $\tilde L^\Phi=0$ can be solved with respect to $u$ (resp. $u_n$)
if $\sco_L Q^\Phi=0$ (resp. $\sco_L Q^\Phi=1$).

The condition $\tilde L^\Phi\equiv0$, where $u$ and $u_n$ are considered as independent variables,
determines those values of~$\Phi$ for which the module~$Q^\Phi$ is ultra-singular for~$\mathcal L$.
We split this condition with respect to~$u_n$, thereby obtaining a system $\mathcal S_{\rm ultra}$ of partial differential equations in~$\Phi$
of orders not greater than~$r$, which may be incompatible in the general case.
Incompatibility here amounts to the family~$\mathfrak M$ containing no ultra-singular modules.
For example, evolution equations of orders greater than 1 have no ultra-singular modules generated by vector fields
of the general form~\eqref{EqReducedForOfSingularModulesForEvolEqs}, see~Section~\ref{SectionOnExampleOfEvolEqs}.
That the parameter-function~$\Phi$ satisfies the ultra-singularity condition $\mathcal S_{\rm ultra}$ guarantees that $Q^\Phi\in\mathcal R^{n-1}(\mathcal L)$
and the family of $Q^\Phi$-invariant solutions of~$\mathcal L$ is parameterized
by an arbitrary function of the single $Q^\Phi$-invariant variable~$x_n$.

Under the assumption $\sco_L Q^\Phi\leqslant0$ with $\tilde L^\Phi\not\equiv0$,
it follows that the parameter-function~$\Phi$ satisfies the condition $\tilde L^\Phi_{u_n}\equiv0$
with $u$ and $u_n$ viewed as independent variables, which is weaker than the ultra-singularity condition.
In this case the corresponding system~$\mathcal S_0$ of partial differential equations in~$\Phi$ of orders not greater than~$r$,
obtained by splitting the zero co-order singularity condition with respect to~$u_n$,
is more likely to be compatible than~$\mathcal S_{\rm ultra}$.

Next we provide certain sufficient conditions for the compatibility of~$\mathcal S_0$.
If $\check L_{u_n}=0$, 
then the system $\mathcal S_0$ is compatible
since it is satisfied by any~$\Phi$ with $(\Phi_s/\Phi_u)_u=0$.
Up to equivalence of functions parameterizing modules from the family~$\mathfrak M$,
we can assume that $\Phi=u-\zeta(x)$ and $\eta^s=\zeta_s$, i.e., $Q^{u-\zeta}=\langle\p_s+\zeta_s\p_u\rangle$.
In other words, we have $\sco_L Q^{u-\zeta}\leqslant0$ for any $\zeta=\zeta(x)$.

Additionally assuming $\check L_u=0$, the condition $\tilde L^{u-\zeta}\equiv0$ under the assumption $\zeta=\zeta(x)$
implies only a single partial differential equation with respect to~$\zeta$.
Any of its solutions is a solution of $\mathcal S_{\rm ultra}$ and hence
the corresponding module $Q^{u-\zeta}$ is ultra-singular for~$L$.

\looseness=1
Otherwise $\sco_L Q^{u-\zeta}=0$ and we can resolve the equation~$\mathcal L$
with respect to~$u$ as a variable of the underlying jet space.
In other words, we represent this equation in the form $u=K[u]$.
Here $K[u]$ is a differential function
which depends at most on the variables~$x$ and $u_\alpha$
with $0<|\alpha|\leqslant r$, $\alpha_n\leqslant1$ and, if $\alpha_n=1$, $|\alpha|>1$.
The equation $\tilde L^{u-\zeta}=0$ is obtained via replacing, in~$\mathcal L$, 
$u_{\alpha+\delta_s}$ by $\zeta_{\alpha+\delta_s}(x)$ for all~$\alpha$ with $|\alpha|<r$; 
see the notation in Section~\ref{SectionOnDefOfRedModules}.
Therefore, this equation can be represented in the form $u=G^\zeta(x)$,
where $\smash{G^\zeta(x)=K[u]\big|_{u=\zeta(x)}}$.
Thus, the expression for the function $G^\zeta$ involves derivatives of the parameter-function~$\zeta=\zeta(x)$ up to order~$r$.
The conditional invariance of the equation~$\mathcal L$ with respect to the module~$Q^{u-\zeta}$
is equivalent to the compatibility of the system $u=G^\zeta$, $u_s=\zeta_s$ with respect to~$u$ and hence
leads to the system of $n-1$ partial differential equations $\zeta_s=G^\zeta_s$ with respect to $\zeta$.
Here $G^\zeta$ is interpreted as a differential function of $\zeta=\zeta(x)$.
Then $G^\zeta_s$ is the total derivative of~$G^\zeta$ with respect to~$x_s$.
As the parameter-function~$\zeta$ is defined for a fixed module up to a constant summand,
the system is equivalent to the single partial differential equation $\zeta=G^\zeta$.
This means that $Q^{u-\zeta}$ is a reduction module of the equation~$\mathcal L$
if and only if up to shifts of the dependent variable the parameter-function~$\zeta$ is a solution of the same equation.
The ansatz constructed with $Q^{u-\zeta}$ can be taken in the form $u=\varphi(\omega)+\zeta(x)$,
where $\varphi=\varphi(\omega)$ is the new unknown function and $\omega=x_n$ is the invariant independent variable.
It reduces the initial equation~$\mathcal L$ to the trivial algebraic equation $\varphi=0$,
i.e., the function $u=\zeta(x)$ is the unique $Q^{u-\zeta}$-invariant solution of~$\mathcal L$.
Conversely, we fix a solution $u=\zeta(x)$ of the equation~$\mathcal L$.
Then $\zeta=G^\zeta(x)$ and hence $\zeta_s=G^\zeta_s$,
which is equivalent to the conditional invariance criterion for the case of
the module $Q^{u-\zeta}=\langle\p_s+\zeta_s\p_u\rangle$ and the equation~$\mathcal L$, i.e.,
$Q^{u-\zeta}$ is a reduction module of~$\mathcal L$, and $\zeta_u=0$.
The solution $u=\zeta(x)$ is invariant with respect to~$Q^{u-\zeta}$ by construction.
Thus we obtain:

\begin{theorem}\label{TheoremOnSolutionsAndZeroOrderSingularRedOpsForGen2DPDEs}
Suppose that an equation~$\mathcal L$: $L=0$ possesses
a singularity-co-order-one family $\mathfrak M=\{Q^\Phi\}$ of $(n{-}1)$-dimensional modules in reduced form
$Q^\Phi=\langle\p_s-(\Phi_s/\Phi_u)\p_u\rangle$ parameterized by an arbitrary smooth function~$\Phi$ of $(x,u)$, i.e.,
the left hand side~$L$ of this equation is represented in the form $L=\check L(x,\Omega_{r,1,n-1})$, where
$\Omega_{r,1,n-1}=\big(u_\alpha,\,\alpha_n\leqslant 1,\,|\alpha|\leqslant r\big)$,
$\check L_{u_\alpha}\ne0$ for some $u_\alpha$ with $\alpha_n=1$,
and additionally $\check L_{u_n}=0$ and $\check L_u\ne0$.
Then there exists a one-to-one correspondence between solutions of~$\mathcal L$
and reduction modules from~$\mathfrak M$ with $\Phi=u-\zeta(x)$.
Namely, any such module is of singularity co-order~$0$ and
corresponds to the unique solution which is invariant with respect to this module.
The problems of solving the equation~$\mathcal L$
and of the exhaustive description of its reduction modules of the above form are completely equivalent.
\end{theorem}

Next we turn to analyzing the regular values of~$\Phi$ for which the singularity co-order of $Q^\Phi$
coincides with the singularity co-order of the whole family~$\mathfrak M$ (and equals 1).
If $\sco_L Q^\Phi=1$, the parameter-function~$\Phi$ satisfies the regularity condition $\tilde L^\Phi_{u_n}\ne0$.
Thus the equation $\tilde L^\Phi=0$,
which is equivalent to~$\mathcal L$ on the manifold~$\smash{\mathcal Q^\Phi_{\ssl r\ssr}}$,
can be solved with respect to~$u_n$: $u_n=G^\Phi(x,u)$.
Here the expression for the function~$G^\Phi$ depends on derivatives of the parameter-function~$\Phi$ up to order~$r$.
The conditional invariance criterion, when applied to the equation~$\mathcal L$ and the module~$Q^\Phi$, gives the system
\begin{equation}\label{EqDEFor1stOrderSingularRedModulesOfGenPDEs}
\eta^s_n+\eta^s_uG^\Phi=G^\Phi_s+\eta^sG^\Phi_u
\end{equation}
with respect to the function~$\Phi$ (recall that $\eta^s=-\Phi_s/\Phi_u$).
This system coincides with the formal compatibility condition
of the equations $u_n=G^\Phi$ and $u_s=\eta^s$ with respect to~$u$.
As the function~$G^\Phi$ can also be expressed directly in terms of the coefficients~$\eta^s$,
the system~\eqref{EqDEFor1stOrderSingularRedModulesOfGenPDEs}
supplemented with the involutivity condition $\eta^s_{s'}+\eta^{s'}\eta^s_u=\eta^{s'}_s+\eta^s\eta^{s'}_u$
can be interpreted as a system with respect to~$\eta^s$, cf.\ Section~\ref{SectionOnExampleOfEvolEqs}.
Since
\begin{gather*}
G^\Phi_s+\eta^sG^\Phi_u-\eta^s_n-\eta^s_uG^\Phi
=G^\Phi_s-\frac{\Phi_s}{\Phi_u}G^\Phi_u+\left(\frac{\Phi_s}{\Phi_u}\right)_n+\left(\frac{\Phi_s}{\Phi_u}\right)_uG^\Phi
\\\qquad
=\frac1{\Phi_u}\left(\p_s-\frac{\Phi_s}{\Phi_u}\p_u\right)\left(\Phi_n +\Phi_uG^\Phi\right)=0,
\end{gather*}
system~\eqref{EqDEFor1stOrderSingularRedModulesOfGenPDEs} is in fact equivalent to the single equation $\Phi_n+\Phi_uG^\Phi=\chi(x_n,\Phi)$ in~$\Phi$,
where the parameter-function~$\chi$ is an arbitrary smooth function of its arguments.
The value of~$\Phi$ associated with a fixed module~$Q^\Phi$ is defined up to the transformation $\tilde\Phi=\theta(x_n,\Phi)$.
Since
$\tilde\eta^s=-\tilde\Phi_s/\tilde\Phi_u=-\Phi_s/\Phi_u=\eta^s,$
the functions~$G^\Phi$ and~$\smash{G^{\tilde\Phi}}$ coincide.
At the same time, if we choose~$\theta$ satisfying the equation $\theta_n+\chi\theta_\Phi=0$,
then $\tilde\Phi_n+\tilde\Phi_uG^{\tilde\Phi}=0$.
Therefore, up to the equivalence on the set of functions parameterizing modules from the family~$\mathfrak M$, 
we can assume that the function~$\Phi$ is a solution of the equation $\Phi_n+\Phi_uG^\Phi=0$.

The order of each equation from system~\eqref{EqDEFor1stOrderSingularRedModulesOfGenPDEs} with respect to~$\Phi$ equals $r+1$
and hence is greater than the order of the system~$\mathcal S_0$.
Under certain smoothness assumptions (e.g., analyticity) this implies
that the system~\eqref{EqDEFor1stOrderSingularRedModulesOfGenPDEs} has solutions which are not solutions of~$\mathcal S_0$.
Consequently, \emph{the equation~$\mathcal L$ necessarily possesses reduction modules of singularity co-order one
that belong to~$\mathfrak M$.}

\begin{proposition}\label{PropositionOnReductionOfDetEqsForFirstCoOrderSingularRedModulesForGenPDEs}
Suppose that any $(n{-}1)$-dimensional involutive module of the reduced form $Q=\langle\p_s+\eta^s\p_u\rangle$
is a module of singularity co-order one for an equation~$\mathcal L$.
Then the determining system for values of $\eta^s$ associated with reduction modules of~$\mathcal L$
is reduced by the composition of the nonlocal substitution $\eta^s=-\Phi_s/\Phi_u$,
where $\Phi$ is a smooth function of $(x,u)$ with~$\Phi_u\ne0$, and the hodograph transformation
\begin{gather*}
\mbox{the new independent variables:}\qquad\tilde x_i=x_i, \quad \varkappa=\Phi,
\\
\lefteqn{\mbox{the new dependent variable:}}\phantom{\mbox{the new independent variables:}\qquad }\tilde u=u
\end{gather*}
to the initial equation~$\mathcal L$ for the function $\tilde u=\tilde u(\tilde x,\varkappa)$
with $\varkappa$ playing the role of a parameter.
\end{proposition}

\begin{proof}
The possibility of representing $\eta^s$ in the form $\Phi_s/\Phi_u$ with some function $\Phi=\Phi(x,u)$ follows from
the involutivity condition $\eta^s_{s'}+\eta^{s'}\eta^s_u=\eta^{s'}_s+\eta^s\eta^{s'}_u$ for the module~$Q$.
In this representation, the system of determining equations for values of $\eta^s$ associated with reduction modules of~$\mathcal L$
has the form~\eqref{EqDEFor1stOrderSingularRedModulesOfGenPDEs} and is equivalent, up to equivalence on the set traversed by the parameter-function~$\Phi$,
to the single equation $\Phi_n+\Phi_uG^\Phi=0$.
In view of the definition of the function~$G^\Phi$ and the expressions for derivatives under the hodograph transformation,
$\tilde u_{\tilde x_i}=-\Phi_i/\Phi_u$, etc.,
the hodograph transformation maps the equation $\Phi_n+\Phi_uG^\Phi=0$ into
to the initial equation~$\mathcal L$ for the function $\tilde u=\tilde u(\tilde x,\varkappa)$.
\end{proof}

Proposition~\ref{PropositionOnReductionOfDetEqsForFirstCoOrderSingularRedModulesForGenPDEs} states
the reduction of the determining system~\eqref{EqDEFor1stOrderSingularRedModulesOfGenPDEs} to the initial equation~$\mathcal L$
and is thus a ``no-go'' assertion.
It can be reformulated in terms of a relation between one-parameter families of solutions
and $(n{-}1)$-dimensional reduction modules of singularity co-order one.
The results of Section~\ref{SectionOnReductionModulesAndParametricFamiliesOfSolutions} imply
that for each such reduction module~$Q$ of the equation~$\mathcal L$
there exists a one-parameter family of $Q$-invariant solutions of~$\mathcal L$.
If the equation~$\mathcal L$ admits an $n$-dimensional co-order one meta-singular module,
the converse statement is true as well.
It is convenient to prove this statement without transforming the meta-singular module to the reduced form.

\begin{theorem}\label{TheoremOn1parametricSolutionFamiliesAndFirstCoOrderSingularRedModulesForGenPDEsA}
Suppose that an equation~$\mathcal L$ possesses an $n$-dimensional co-order one meta-singular module~$M$.
Then for any one-parameter family~$\mathcal F$ of solutions of~$\mathcal L$
there exists an $(n{-}1)$-dimensional involutive submodule~$Q$ of~$M$
that is a reduction module of~$\mathcal L$ and each solution from~$\mathcal F$ is $Q$-invariant.
\end{theorem}

\begin{proof}
Let $\mathcal F=\{u=f(x,\varkappa)\}$ be a one-parameter family of solutions of~$\mathcal L$.
Here, $f_\varkappa$ is nonzero since the parameter $\varkappa$ is essential.
From $u=f(x,\varkappa)$ we conclude that $\varkappa=\Phi(x,u)$
with some function $\Phi=\Phi(x,u)$ with $\Phi_u\ne0$.

Let $(Q_0,\dots,Q_p)$ be a commutative basis of~$M$, $p=n-1$.
Suppose that there exists $\sigma\in\{0,\dots,p\}$ such that $Q_\sigma\Phi\ne0$.
Up to permutation of basis elements we can assume that $Q_0\Phi\ne0$. Then we set $\tilde Q_s=Q_s-(Q_s\Phi)/(Q_0\Phi)Q_0$.
If $Q_\sigma\Phi=0$ for all $\sigma\in\{0,\dots,p\}$, we set $\tilde Q_s=Q_s$.

Consider the submodule~$Q$ generated by the vector fields $\tilde Q_s$.
This submodule is involutive, 
satisfies the rank condition 
and is of dimension $p=n-1$.
Hence $\sco_L Q^\Phi\leqslant1$.
It is also obvious that $\tilde Q_s\Phi=0$.
As $f_i=-(\Phi_i/\Phi_u)|_{u=f}$, this means that any solution from the family~$\mathcal F$ is $Q$-invariant.
The case $\sco_L Q\leqslant0$ with $L|_{\mathcal Q_{\ssl r\ssr}}\not\equiv0$ is impossible as otherwise the equation~$\mathcal L$
could not have a one-parameter family of $Q$-invariant solutions.
Therefore either $Q$ is an ultra-singular module for~$L$ or $\sco_L Q=1$.
Any ultra-singular module for~$\mathcal L$ is a reduction module of~$\mathcal L$.
If $\sco_L Q=1$, then $Q$ is a reduction module of~$\mathcal L$
by Proposition~\ref{PropositionOnSingularCodimOneRedModulesAndFamiliesOfInvSolutions}.
\end{proof}

For the correspondence between one-parameter families of solutions and $(n{-}1)$-dimensional reduction modules of singularity co-order one
to be one-to-one, the related meta-singular module should satisfy additional restrictions.

\begin{theorem}\label{TheoremOn1parametricSolutionFamiliesAndFirstCoOrderSingularRedModulesForGenPDEsB}
Suppose that 
an $r$th order equation~$\mathcal L$: $L[u]=0$ possesses an $n$-dimensional co-order one meta-singular module~$M$,
where the entire module~$M$ is not ultra-singular for~$\mathcal L$,
each $(n{-}1)$-dimensional submodule~$Q$ of~$M$ is of singularity co-order one for~$\mathcal L$,
and the equation $\hat L=0$ with a first-order differential function~$\hat L=\hat L[u]$,
associated with~$L[u]$ on the manifold~$\mathcal Q_{\ssl r\ssr}$ up to a nonvanishing multiplier,
can be solved with respect to the first-order derivative of~$u$ appearing in this equation.
Then up to the equivalence of solution families
there exists a bijection between one-parameter families of solutions of~$\mathcal L$
and its $(n{-}1)$-dimensional reduction modules contained in~$M$.
Namely, each module of this kind corresponds to the family of solutions
which are invariant with respect to it.
\end{theorem}

In other words, the problems of the construction of all one-parameter solution families of the equation~$\mathcal L$
and the exhaustive description of its reduction modules of the above form are completely equivalent.

\begin{proof}
If $Q$ is an $(n{-}1)$-dimensional reduction module of~$\mathcal L$ contained in~$M$, we have $\sco_{\mathcal L}Q=1$.
Therefore, in view of Proposition~\ref{PropositionOnSingularCodimOneRedModulesAndFamiliesOfInvSolutions}
the equation~$\mathcal L$ possesses a one-parameter family~$\mathcal F$ of $Q$-invariant solutions,
and any $Q$-invariant solution of~$\mathcal L$ belongs to this family.
Each one-parameter family of $Q$-invariant solutions of~$\mathcal L$ is obtained from~$\mathcal F$ by re-parameterizing.

Conversely, consider a one-parameter family $\mathcal F=\{u=f(x,\varkappa)\}$ of solutions of~$\mathcal L$.
Theorem~\ref{TheoremOn1parametricSolutionFamiliesAndFirstCoOrderSingularRedModulesForGenPDEsA} implies that
there exists $Q\in \mathcal R^{n-1}(\mathcal L)$ such that $Q$ is a submodule of~$M$ and
any solution from~$\mathcal F$ is $Q$-invariant.
Let us prove that the module~$Q$ is unique.
Suppose that there exists an $(n{-}1)$-dimensional involutive submodule~$\tilde Q$ of~$M$ different from~$Q$ and
any solution from~$\mathcal F$ is $\tilde Q$-invariant.
This implies that the family~$\mathcal F$ consists of solutions invariant with respect to the entire module~$M$.

Therefore, the module~$M$ satisfies the rank condition.
To show this, we fix a basis $(Q_0,\dots,Q_p)$ of~$M$, where
$Q_\sigma=\xi^{\sigma i}(x,u)\p_i+\eta^\sigma(x,u)\p_u$, $\sigma=0,\dots,n$.
The function~$\Phi$ defined in the proof of Theorem~\ref{TheoremOn1parametricSolutionFamiliesAndFirstCoOrderSingularRedModulesForGenPDEsA}
is an invariant of all~$Q_\sigma$, i.e., $Q_\sigma\Phi=0$.
As $\Phi_u\ne0$, we have that $\rank(\xi^{\sigma i})=n$.

As $M$ is an $n$-dimensional involutive module satisfying the rank condition and
the family~$\mathcal F$ formed by $M$-invariant solutions of the equation~$\mathcal L$ is one-parameter,
in view of Proposition~\ref{PropositionOnUltraSingularRedOpsAndFamiliesOfInvSolutions}
the module~$M$ is ultra-singular for~$\mathcal L$, thereby contradicting an assumption of the theorem.
This means that the module~$Q$ is unique.
\end{proof}

\looseness=-1
The bijection discussed in Theorem~\ref{TheoremOn1parametricSolutionFamiliesAndFirstCoOrderSingularRedModulesForGenPDEsB}
is generally broken in the presence of $(n{-}1)$-dimensional submodules of singularity co-order less than one
or ultra-singularity of the entire meta-singular module.
Indeed, if an $(n{-}1)$-dimensional involutive module~$Q$ is ultra-singular for the equation~$\mathcal L$,
the family of $Q$-invariant solutions of~$\mathcal L$ is parameterized by an arbitrary function of a single argument,
cf.\ Proposition~\ref{PropositionOnUltraSingularRedOpsAndFamiliesOfInvSolutions}.
In the case $\sco_L Q=-\infty$ with $L|_{\mathcal Q_{\ssl r\ssr}}\not\equiv0$,
the equation~$\mathcal L$ possesses no $Q$-invariant solutions.
If there are reduction modules of~$\mathcal L$ among $(n{-}1)$-dimensional co-order zero singular submodules of~$M$,
the sets of invariant solutions of~$\mathcal L$ for these reduction modules are discrete.
If the entire $n$-dimensional module~$M$ is ultra-singular for the equation~$\mathcal L$,
this equation possesses a one-parameter family~$\mathcal F$ of $M$-invariant solutions.
Then any solution from
~$\mathcal F$ is invariant with respect to each $(n{-}1)$-dimensional involutive submodule of~$M$.

Singularity co-order one of reduction modules is also essential for the statement of
Proposition~\ref{PropositionOnReductionOfDetEqsForFirstCoOrderSingularRedModulesForGenPDEs} and
Theorem~\ref{TheoremOn1parametricSolutionFamiliesAndFirstCoOrderSingularRedModulesForGenPDEsB}.
For example, the only meta-singular module of the linear rod equation $u_{tt}+u_{xxxx}=0$
is $M=\langle\p_x,\p_u\rangle$, and its singularity co-order equals two, cf.\ point 2 in Example~\ref{ExamplesOfSingularModules}.
The one-dimensional singular reduction modules of this equation were found in explicit form in~\cite{Boyko&Popovych2013}.
They do not constitute a no-go case since all of them are of singularity co-order two,
and the corresponding families of invariant solutions are two-parameter.

\begin{remark}
Reduction modules of codimension one and singularity co-order one do not exhaust possible
no-go cases for finding reduction modules.
In particular, the system of determining equations for regular reduction operators
of any (1+1)-dimensional linear second-order evolution equation~$\mathcal L$ is reduced by a nonlocal transformation
to a system of three copies of~$\mathcal L$~\cite{Popovych2008a}.
Therefore, the regular reduction operators of~$\mathcal L$ constitute
a no-go case different from the no-go case of singular reduction operators,
which is common for all (1+1)-dimensional evolution equations.
A similar phenomenon occurs for the Burgers equation $u_t+uu_x-\mu u_{xx}=0$,
where a no-go case arises for regular reduction operators
of the form $\p_t+\xi(t,x,u)\p_x+\eta(t,x,u)\p_u$ with $\xi_u=1/2$
\cite{Arrigo&Hickling2002,Mansfield1999,Pocheketa&Popovych2013}.
We believe that these no-go cases are effected by the coupling
of several properties of related equations such as
the evolution form, the second order, and the linearity or linearizability.
For a class of differential equations, the study of reduction modules may lead to no-go cases
due to the appearance of arbitrary elements parameterizing equations of the class
in the corresponding determining equations.
We refer to~\cite{Pocheketa&Popovych2012} for the calculation of
one-dimensional regular reduction modules
spanned by vector fields of the form $\p_t+\xi(t,x)\p_x+\eta(t,x,u)\p_u$ with $\xi_{xx}\ne0$
for the class of generalized Burgers equations $u_t+uu_x+f(t,x)u_{xx}=0$.%
\noprint{
Is there a no-go theorem for one-dimensional singular modules
of (1+2)-dimensional degenerate nonlinear diffusion--convection equation?
No, there is not!
}
\end{remark}

\section{Singular modules for quasi-linear second-order PDEs}
\label{SectioOnSingularModulesForQuasiLinearSecondOrderPDEs}

For some classes of differential equations it is possible to exhaustively describe
the associated singular modules.
We study certain quasi-linear second-order PDEs from this point of view.
It is natural to distinguish elliptic, evolution and generalized wave equations.
In all cases, $Q$ is an involutive module of vector fields
defined in the corresponding space of independent and dependent variables and satisfies the rank condition,
and the dimension of~$Q$ is less than the number of independent variables.

\medskip

\noindent {\bf Elliptic equations.}
Consider an equation~$\mathcal L$ for the single unknown function~$u$ of the independent variables $x=(x_1,\dots,x_n)$,
having the general form
\[
L[u]:=a^{ij}(x,u_{\ssl1\ssr})u_{ij}+b(x,u_{\ssl1\ssr})=0,
\]
where the coefficients $a^{ij}$ and $b$ are defined on the same domain~$\Omega$ of the first-order jet space,
$a^{ij}=a^{ji}$ and the matrix-function~$(a^{ij})$ is positive definite in each point of~$\Omega$.
We will prove that the equation~$\mathcal L$ possesses no singular modules of dimensions less than~$n$.

Denote $\dim Q$ by~$p$, $0<p<n$.
Up to permutation of~$x$'s we can locally choose a basis of~$Q$ which consists of vector fields of the form
$\hat Q_s=\p_s+\hat\xi^{s\iota}(x,u)\p_\iota+\hat\eta^s(x,u)\p_u$,
cf.\ Section~\ref{SectionOnSingularModulesOfVectorFieldsForDiffFunctions}.
Here and in what follows the index~$\iota$ runs from $p+1$ to~$n$.
Then any derivative of~$u$ of order one or two is expressed, on the manifold~$\smash{\mathcal Q_{\ssl2\ssr}}$,
via derivatives of~$u$ with respect to $x_{p+1}$, \dots, $x_n$ only
and the coefficients $\hat\xi^{s\iota}$ and $\hat\eta^s$,
see equation~\eqref{EqSingularRedModuleExpressionsForU_2}.
As only second-order terms in the expressions for second-order derivatives of~$u$ are essential here,
we use the representations
\[
u_{s\iota}=-\hat\xi^{s\iota'} u_{\iota\iota'}+R^{s\iota}(x,u_{\ssl1\ssr}), \quad
u_{ss'}=\hat\xi^{s\iota}\hat\xi^{s'\iota'} u_{\iota\iota'}+R^{ss'}(x,u_{\ssl1\ssr}),
\]
where the $R^{si}$ denote the terms without second-order derivatives, and $R^{ss'}=R^{s's}$.
Substituting the above expressions for the derivatives~$u_{si}$ into~$L[u]$,
we obtain the differential function
$\hat L[u]:=\hat a^{\iota\iota'}u_{\iota\iota'}+\hat b$,
which is associated with~$L[u]$ on the manifold~$\smash{\mathcal Q_{\ssl2\ssr}}$.
Here $\hat a^{\iota\iota'}=a^{\iota\iota'}-a^{s\iota}\hat\xi^{s\iota'}-a^{s\iota'}\hat\xi^{s\iota}+a^{ss'}\hat\xi^{s\iota}\hat\xi^{s'\iota'}$
and $\hat b=b+2a^{s\iota}R^{s\iota}+a^{ss'}R^{ss'}$.
In other words, for each fixed~$\iota$ and $\iota'=\iota$ the coefficient~$\hat a^{\iota\iota'}$ coincides
with the value of the quadratic form whose matrix is~$(a^{ij})$ at the tuple $(z^1,\dots,z^n)$,
where $z^s=-\hat\xi^{s\iota}$, $z^\iota=1$ and the other~$z$'s equal zero.
As the matrix~$(a^{ij})$ is positive definite, the coefficient~$\hat a^{\iota\iota'}$ is nonvanishing (and, moreover, positive).
This implies that the differential function~$L[u]$ cannot coincide on the manifold~$\smash{\mathcal Q_{\ssl2\ssr}}$
with a differential function of order less than two, not even up to a nonvanishing multiplier.
Therefore, $\wsco_{\mathcal L}Q=\sco_{\mathcal L}Q=2$.

\medskip

\noindent {\bf Evolution equations.}
Similarly to Section~\ref{SectionOnExampleOfEvolEqs}, for this and the next class of equations
we set the number of independent variables to $n+1$ instead of~$n$
and additionally single out the variable~$t=x_0$, i.e.,
the unknown function $u$ depends on the variables~$t$ and $x=(x_1,\dots,x_n)$.
The general form of a quasi-linear second-order evolution equation~$\mathcal L$ we intend to study is
\begin{equation}\label{EqQuasiLinear2ndOrderEvolution}
u_t=H[u]:=a^{ij}(t,x,u_{\ssl1,x\ssr})u_{ij}+b(t,x,u_{\ssl1,x\ssr}),
\end{equation}
where $u_{\ssl1,x\ssr}=(u,u_1,\dots,u_n)$,
the coefficients $a^{ij}$ and $b$ are defined on the same domain~$\Omega$ of the first-order jet space,
and the matrix-function~$(a^{ij})$ is symmetric and positive definite in each point of~$\Omega$.
Up to permutation of the variables~$x_i$, any involutive module~$Q$
of vector fields defined in the space of $(t,x,u)$ and satisfying the rank condition,
where $\dim Q<n+1$, can be locally assumed to be spanned either by the vector fields
$\p_t+\xi^{0\iota}\p_\iota+\eta^0\p_u$, $\p_s+\xi^{s\iota}\p_\iota+\eta^s\p_u$ with $p=\dim Q-1$
or by the vector fields $\p_s+\xi^{s\iota}\p_\iota+\eta^s\p_u$ with $p=\dim Q$.
Here $\xi^{0\iota}$, $\eta^0$, $\xi^{s\iota}$ and $\eta^s$ are smooth functions of $(t,x,u)$.
We can show similarly as for elliptic equations that
in the second case with $p=n$ and only in this case the module~$Q$ is singular for~$\mathcal L$,
and $\wsco_{\mathcal L}Q=\sco_{\mathcal L}Q=1$.
Then the basis elements of~$Q$ take the form $Q^s=\p_s+\eta^s\p_u$.
This is the case that has been studied in Section~\ref{SectionOnExampleOfEvolEqs}.
In contrast to general evolution equations, cf.\ Corollary~\ref{CorollaryOSingularVectorFieldsOfEvolEqs},
we can guarantee that $M=\langle\p_1,\dots,\p_n,\p_u\rangle$ is the only meta-singular module
of the equation~$\mathcal L$.

\medskip

\noindent {\bf Generalized wave equations.}
Consider the equation~$\mathcal L$: $u_{tt}=H[u]$
for the single unknown function~$u$ of the independent variables $t$ and $x=(x_1,\dots,x_n)$,
where the differential function $H=H[u]$ is defined analogously to~\eqref{EqQuasiLinear2ndOrderEvolution}
but the coefficients~$a^{ij}$ and~$b$ may additionally depend on~$u_t$.
We partition the set of appropriate involutive modules
in a way different from that used for evolution equations.
Up to permutation of the variables~$x$, any involutive module~$Q$
of vector fields defined in the space of $(t,x,u)$ and satisfying the rank condition,
where $\dim Q<n+1$, can be locally assumed to be spanned either by the vector fields
$\p_t+\eta^0\p_u$, $\p_s+\xi^{s\iota}\p_\iota+\eta^s\p_u$
with $p=\dim Q-1$
or by the vector fields $\p_s+\tau^s\p_t+\xi^{s\iota}\p_\iota+\eta^s\p_u$ with $p=\dim Q$.
Here $\eta^0$, $\tau^s$, $\xi^{s\iota}$ and $\eta^s$ are smooth functions of $(t,x,u)$.
Only in the second case with $p=n$ the module~$Q$ may be singular for~$\mathcal L$.
We consider this case in detail.

As $p=n$, the basis elements of~$Q$ take the form $Q^s=\p_s+\tau^s\p_t+\eta^s\p_u$.
On the manifold~$\smash{\mathcal Q_{\ssl2\ssr}}$ we have that
$
u_s=\eta^s-\tau^s u_t$,
$
u_{ss'}
=\tau^s\tau^{s'}u_{tt}+R^{ss'}$,
where $R^{ss'}$ denotes the corresponding collection of the terms depending at most on $t$, $x$, $u$, and~$u_t$.
Substituting the above expressions for the derivatives of~$u$ with respect to~$x$ into the differential function $L[u]=u_{tt}-H[u]$,
we obtain the differential function
$\hat L[u]:=(1-\hat a^{ij}\tau^i\tau^j)u_{tt}+\hat b(t,x,u,u_t)$,
which is associated with~$L[u]$ on the manifold~$\smash{\mathcal Q_{\ssl2\ssr}}$ and depends at most on $t$, $x$, $u$, $u_t$ and~$u_{tt}$.
Here the coefficient $\hat a^{ij}=\hat a^{ij}(t,x,u,u_t)$ is obtained from $a^{ij}$ by the substitution $u_s=\eta^s-\xi^s u_t$
and the precise form of the coefficient $\hat b=\hat b(t,x,u,u_t)$ is not essential.
The module~$Q$ is singular for~$\mathcal L$ if and only if $\ord\hat L[u]<2$, i.e.,
the coefficients $\tau^s=\tau^s(t,x,u)$ satisfy the equation $\hat a^{ij}\tau^i\tau^j=1$.
If some of the coefficients~$\hat a^{ij}$ depend on~$u_t$, this equation should be split with respect to this derivative and hence it may be inconsistent.

As the module~$Q$ is involutive and, therefore, the vector fields~$Q^s$ commute,
the coefficients $\tau^s$ and $\eta^s$ jointly satisfy the system $Q^s\tau^{s'}=Q^{s'}\tau^s$, $Q^s\eta^{s'}=Q^{s'}\eta^s$, i.e.,
\begin{gather*}
\tau^{s'}_s+\tau^s\tau^{s'}_t+\eta^s\tau^{s'}_u=\tau^s_{s'}+\tau^{s'}\tau^s_t+\eta^{s'}\tau^s_u,\\
\eta^{s'}_s+\tau^s\eta^{s'}_t+\eta^s\eta^{s'}_u=\eta^s_{s'}+\tau^{s'}\eta^s_t+\eta^{s'}\eta^s_u.
\end{gather*}
In view of the Frobenius theorem this implies that
the system $Q^s\Phi=0$ with respect to the unknown function~$\Phi=\Phi(t,x,u)$
admits solutions $\Phi^l$, $l=1,2$, such that $\Phi^1_t\Phi^2_u-\Phi^1_u\Phi^2_t\ne0$.
Solving the pair of the equations $\Phi^l_s+\tau^s\Phi^l_t+\eta^s\Phi^l_u=0$ for each fixed~$s$
as a system of linear algebraic equations with respect to $(\tau^s,\eta^s)$, we derive the representation
\begin{gather}\label{EqRepresentationForCoeffsOfInvModulesForNWEs}
\tau^s=-\frac{\Phi^1_s\Phi^2_u-\Phi^2_s\Phi^1_u}{\Phi^1_t\Phi^2_u-\Phi^2_t\Phi^1_u},\quad
\eta^s=-\frac{\Phi^1_t\Phi^2_s-\Phi^2_t\Phi^1_s}{\Phi^1_t\Phi^2_u-\Phi^2_t\Phi^1_u}.\quad
\end{gather}
Conversely, it can be checked by direct calculation that
for arbitrary functions $\Phi^1$ and $\Phi^2$ with $\Phi^1_t\Phi^2_u-\Phi^1_u\Phi^2_t\ne0$
the module~$Q$ spanned by the vector fields~$Q^s=\p_s+\tau^s\p_t+\eta^s\p_u$
with the coefficients defined by~\eqref{EqRepresentationForCoeffsOfInvModulesForNWEs} is involutive.
Up to functional dependence of pairs of functions $\Phi^1$ and $\Phi^2$ with $\Phi^1_t\Phi^2_u-\Phi^1_u\Phi^2_t\ne0$,
there exists a bijection between such pairs and involutive modules spanned by $n$~vector fields of the form~$Q^s=\p_s+\tau^s\p_t+\eta^s\p_u$.

If the coefficients~$\hat a^{ij}$ do not depend on~$u_t$
(this is the case if, e.g., $a^{ij}=a^{ij}(t,x,u)$ and then $\hat a^{ij}=a^{ij}$),
the substitution of the expressions~\eqref{EqRepresentationForCoeffsOfInvModulesForNWEs} for~$\tau^s$
into the equation $\hat a^{ij}\tau^i\tau^j=1$ gives
a single equation in the two unknown functions $\Phi^1$ and $\Phi^2$.
Each solution of this equation is associated with a singular module~$Q$ for the equation~$\mathcal L$.
In general, in the multidimensional case $n>1$, the coefficients $\tau^s$ and $\eta^s$
of the corresponding basis elements~$Q^s$ are coupled in a nonlocal and nonlinear way.
This is why \emph{it is impossible to exhaustively describe singular modules
of multidimensional nonlinear wave equations within the framework of meta-singular modules.}
At the same time, under additional conditions for the coefficients~$a^{ij}$, nonlinear wave equations possess
families of singular modules which fit well into the above framework.

\looseness=-1
Thus, let the coefficients~$\hat a^{ij}$ only depend on~$t$ and~$x$.
We look for singular modules of~$\mathcal L$ for which the corresponding coefficients~$\tau^s$ also do not depend on~$u$.
It is then sufficient to consider pairs of functions~$\Phi^1$ and~$\Phi^2$ with $\Phi^1_u=0$
and, therefore, $\Phi^1_t\Phi^2_u\ne0$.
Then the expressions~\eqref{EqRepresentationForCoeffsOfInvModulesForNWEs} for~$\tau^s$
are reduced to $\tau^s=-\Phi^1_s/\Phi^1_t$.
The equation $a^{ij}\tau^i\tau^j=1$ is equivalent to the equation $(\Phi^1_t)^2=a^{ij}\Phi^1_i\Phi^1_j$,
which is the eikonal equation associated with the wave equation~$\mathcal L$.
We fix a solution~$\Psi=\Psi(t,x)$ of the eikonal equation with $\Psi_t\ne0$ and consider the module~$M^\Psi$ spanned by
the vector fields $\Psi_t\p_s-\Psi_s\p_t$, $s=1,\dots,n$, and~$\p_u$.
The module~$M^\Psi$ is meta-singular for the equation~$\mathcal L$, and $\sco_{\mathcal L}M^\Psi\leqslant1$ for each~$\Psi$.
Hence, the results of Section~\ref{SectionOnReductionModulesOfSingularityCoOrder1} are relevant in this case.
The number of such meta-singular modules of singularity co-order one is infinite
since they are parameterized by~$\Phi$ traversing the set of solutions of the above eikonal equation with $\Phi_t\ne0$.

\begin{example}
To consider a more concrete example, we briefly describe singular one-dimensional modules
of (1+1)-dimensional nonlinear wave equations of the general form
\[
u_{tt}-(G(u)u_x)_x-F(u)=0,
\]
where~$F$ and~$G$ are arbitrary smooth functions of~$u$ with $G>0$.
For convenience we use the specific notation of variables, $t=x_0$ and $x=x_1$.
The case $G=\const$ in the characteristic, or light-cone, variables
was exhaustively studied in~\cite[Section~6]{Kunzinger&Popovych2010a}.

Denote a nonlinear wave equation with fixed~$F$ and~$G$ by~$\mathcal L$.
Let a vector field $Q_1=\tau\p_t+\xi\p_x+\eta\p_u$,
where the components~$\tau$, $\xi$ and~$\eta$ are smooth functions of~$(t,x,u)$ with $(\tau,\xi)\ne(0,0)$, constitute a basis of a one-dimensional module~$Q$.
The module~$Q$ is singular for~$\mathcal L$ if and only if
it is strongly singular for~$\mathcal L$, i.e. $\sco_{\mathcal L}Q\leqslant1$.
Both the components~$\tau$ and~$\xi$ are then necessarily nonvanishing.
Changing the basis of~$Q$, we set $\tau=1$.
Then the condition $\sco_{\mathcal L}Q\leqslant1$ is equivalent to the following constraint for the component~$\xi$:
\[
\xi=\pm g(u), \quad\mbox{where}\quad g:=\sqrt G.
\]
This means that the equation~$\mathcal L$ possesses exactly two two-dimensional meta-singular modules,
$M_{\pm}=\langle\p_t\pm g\p_x,\p_u\rangle$ with $\sco_{\mathcal L}M_{\pm}=1$,
and thus the results of Section~\ref{SectionOnReductionModulesOfSingularityCoOrder1}
are relevant here.

The singularity co-order of~$Q$ for~$\mathcal L$ is nonpositive,
$\mathop{\rm (w)sco}_{\mathcal L}Q\leqslant0$,
if and only if additionally the component~$\eta$ satisfies the equation
\[
(\sqrt g\eta)_u=0, \quad\mbox{i.e.}\quad \eta=\frac{\alpha(t,x)}{\sqrt g}
\]
for some smooth function~$\alpha$ of~$(t,x)$.
If we assume~$u_t$ as the principal derivative of the characteristic equation $Q_1[u]=0$,
then the left hand side of the equation~$\mathcal L$ is associated,
on the manifold~$\smash{\mathcal Q_{\ssl 2\ssr}}$,
with the function
\[
\hat L=\frac{\alpha_t-g\alpha_x}{\sqrt g}-\frac{\alpha^2g_u}{2g^2}-F,
\]
which depends only on~$(t,x,u)$.
If the equation $\hat L=0$ can be solved with respect to~$u$, then $\wsco_{\mathcal L}Q=0$;
otherwise $\wsco_{\mathcal L}Q=-\infty$.
The condition singling out ultra-singular modules
from the modules of nonpositive co-order of singularity
is $\hat L\equiv0$ with respect to~$(t,x,u)$.
In the case $g_u\ne0$ this implies that there exist constants $b$, $a_0$ and~$a_1$ such that
$F=bg_u/(2g^2)$,
$\alpha=\alpha(z)$, where $z=a_0t+a_1x$ and $\alpha_z=\alpha^2+b$.
A~nonconstant value of the parameter-function~$\alpha$ is possible only for~$g$
satisfying the equation $g_u=2g^{3/2}(a_0-a_1g)$.

For~$\alpha=\const$, there is no constraint for~$g$.
Therefore, for any positive smooth function $g=g(u)$ and any constant~$\alpha$,
each of the vector fields $Q_{\pm}=\p_t\pm g\p_x+\alpha g^{-1/2}\p_u$ spans an ultra-singular reduction module
of the equation $u_{tt}=(g^2u_x)_x-\alpha^2g_u/(2g^2)$,
which displays the factorization
\[
u_{tt}-(g^2u_x)_x+\frac{\alpha^2g_u}{2g^2}=
\left(D_t\mp D_x\circ g-\frac{\alpha g_u}{2g^{3/2}}\right)
\left(u_t\pm gu_x-\frac{\alpha}{g^{1/2}}\right).
\]
For the particular value $g=1/u$, this factorization was indicated in~\cite{Manganaro&Pavlov2014} in terms of "first order reductions".
Interpreting and developing the above results for (1+1)-dimensional nonlinear wave equations
similarly to the consideration of the particular case $G=\const$ in \cite[Section~6]{Kunzinger&Popovych2010a}
will be the subject of a forthcoming paper.
\end{example}

\noindent{\bf Acknowledgements.}
This work was supported by the Austrian Science Fund (FWF), project P25064.
The authors thank the referees for useful comments and remarks.

\end{document}